\documentclass[journal]{IEEEtran}
\usepackage[cmex10]{amsmath}
\usepackage[mathscr]{euscript}
\usepackage[noadjust]{cite}
\usepackage{%
	amsmath,
	amssymb,%
	amsthm,%
	bbm,%
	cite,%
	enumerate,%
	mathtools,%
	mathrsfs,%
	pgf,%
	pgfplots,%
	siunitx,%
	stmaryrd,%
	tabto,%
	tikz,%
	url,%
	xcolor,%
}

\usepgfplotslibrary{patchplots}
\usepgflibrary{patterns}
\usetikzlibrary{decorations.pathreplacing}
\theoremstyle{plain}
\newtheorem{thm}{Theorem}
\newtheorem{lem}{Lemma}
\newtheorem{prop}{Proposition}
\newtheorem{cor}{Corollary}

\theoremstyle{definition}

\newtheorem{prob}{Problem}

\theoremstyle{remark}
\newtheorem{rem}{Remark}

\usepgflibrary{shapes}
\usetikzlibrary{%
	arrows,%
	decorations.text,%
	positioning,%
	scopes,%
	shapes%
}

\newcommand{\N}{\mathbb{N}}

\newcommand{\eq}[1]{\begin{align*}#1\end{align*}}

\newcommand{\EQ}[1]{\begin{equation*}#1\end{equation*}}
\newcommand{\EQN}[1]{\begin{equation}#1\end{equation}}
\newcommand{\meq}[2]{\begin{xalignat*}{#1}#2\end{xalignat*}}
\newcommand{\meqn}[2]{\begin{xalignat}{#1}#2\end{xalignat}}

\newcommand{\pd}[2]{\frac{\partial #1 }{\partial #2}}
\newcommand{\set}[1]{\left\{#1\right\}}
\newcommand{\SetIn}[1]{\mathbbm{1}_{\set{#1}}}
\DeclareMathOperator{\sgn}{sgn}

\newcommand{\E}{\mathbb{E}}

\newcommand{\R}{\mathbb{R}}

\newcommand{\tI}{\tilde{I}}

\renewcommand{\le}{\leqslant}
\renewcommand{\ge}{\geqslant}

\begin{document}
\title{Optimal Server Selection for Straggler Mitigation}
	
\author{
Ajay~Badita~\IEEEmembership{Student~Member,~IEEE}, 
\and Parimal~Parag~\IEEEmembership{Member,~IEEE},  
\and Vaneet Aggarwal~\IEEEmembership{Senior Member,~IEEE}
\thanks{ 
Ajay Badita and Parimal Parag are with the Department of Electrical and Communications Engineering, Indian Institute of Science, Bangalore, Karnataka 560012, India. 
Vaneet Aggarwal is with the School of Industrial Engineering and the School of Electrical and Computer Engineering, Purdue University, West Lafayette, IN, USA. 
Email: \{ajaybadita, parimal\}@iisc.ac.in, vaneet@purdue.edu. 
}
\thanks{ 
This work is supported in part by the Science and Engineering Research Board (SERB) under Grant~DSTO-1677 and the VAJRA Fellowship, the National Science Foundation under Grant CNS-1618335, and CISCO. 
Any opinions, findings, and conclusions or recommendations expressed in this material are those of the authors. 
}
}

\maketitle
\begin{abstract}
The performance of large-scale distributed compute systems is adversely impacted by stragglers when the execution time of a job is uncertain. 
To manage stragglers, we consider a multi-fork approach for job scheduling, where additional parallel servers are added at forking instants. 
In terms of the forking instants and the number of additional servers, we compute the job completion time and the cost of server utilization when the task processing times are assumed to have a shifted exponential distribution. 
We use this study to provide insights into the scheduling design of the forking instants and the associated number of additional servers to be started. 
Numerical results demonstrate orders of magnitude improvement in cost in the regime of low completion times as compared to the prior works. 
\end{abstract}

	

\section{Introduction}
\label{sec:Intro}
Large scale computing jobs require multi-stage computation, 
where computation per stage is performed in parallel over a large number of servers. 
The execution time of a task on a machine has stochastic variations due to many contributing factors such as co-hosting, virtualization, hardware and network variations~\cite{Cheng2014IMC}.
A slow server can delay the onset of next stage computation, 
and we call it a \emph{straggling} server. 
One of the key challenges in cloud computing is the problem of straggling servers, 
which can significantly increase the job completion time~\cite{Garraghan2016TSC, Ouyang2016SCC, Guo2017TPDS}. 
Straggler mitigation is a particularly important problem, 
considering this the organizations such as VMWare and Amazon have spent substantial effort optimizing the operation of virtualization technologies for massive-scale systems~\cite{Garraghan2016TSC}. 
This paper aims to find efficient scheduling mechanisms for straggler mitigation by analyzing how the replication of straggling tasks affects the mean service completion time and the mean server utilization cost of computing resources. 

The idea of replicating tasks in parallel computing has been adopted at a large scale via the speculative execution in both Hadoop MapReduce~\cite{Cheng2014IMC}, and Apache Spark~\cite{Vernik2017SSC}. 
The use of redundancy to reduce mean service completion time has also attracted attention in other contexts such as cloud storage and networking~\cite{Shah2016TCOM, Xiang2016TNET}. 
These works focus on the queuing aspects at the storage servers. 
Replication is a special case of general redundancy mechanism and is considered in this paper. 
Replication is also referred to as forking in popular scheduling parlance. 
Replicating a job on multiple servers affords us the parallelism gains, 
while it comes at the cost of server utilization. 
We consider a dynamic replication strategy, 
where an unfinished task is sequentially forked over multiple servers at certain forking times. 
We thus provide an efficient tradeoff between the mean service completion time and the mean utilization cost of computing resources. 

Recently, the authors of~\cite{Wang2015Sigmetrics} provided a framework for analyzing straggling tasks for a computing job. 
The authors of~\cite{Wang2015Sigmetrics} considered executing $K$ jobs (or tasks), 
where one copy for each job was started at time $t=0$. 
They had a single forking point at the instant of job completion of a fraction $(1-p)$ of all $K$ jobs. 
At this forking point, each of the remaining $pK$ incomplete jobs is replicated $r$ times. 
Two variants, where the original tasks were killed or kept at the forking point were considered. 
In this setting, the mean service completion time and the mean server utilization cost of computing resources per job were computed in the limit as $K\to\infty$, 
where the execution time follows either a shifted exponential or a Pareto distribution. 
The analysis assumes a single forking point, 
corresponding to the time where multiple replicas are run for an unfinished job.  

In contrast, we provide a multi-fork analysis of the computing jobs, 
with a selection of number of servers for replication at each forking point. 
Specifically, we assume $K$ jobs, all starting at $t_0 = 0$ and an identical sequence of $m$ forking points for each job, 
denoted by $t_i$ for $i \in [m] \triangleq \set{1, \dots, m}$.  
We initialize each task on $n_0$ parallel servers at instant $t_0=0$. 
At each forking point $t_i$, we start additional $n_i$ replicas for each unfinished job. 
If a job is unfinished for any time $t \in [t_i, t_{i+1})$, then it has $N_i = \sum_{j=0}^i n_j$ active replicas. 
This procedure is illustrated in Fig.~\ref{Fig:Numberofserversforked}, where we plot the time-evolution of number of active replicas for a single unfinished task. 
With multiple forking points, the mean service completion time and average server utilization cost are evaluated where the server execution times are assumed to be \emph{i.i.d.} following a shifted exponential distribution, 
and the forking points are separated by at least the shift of the distribution. 
\begin{figure}[hhh]
\centering
\pgfplotsset{
every axis/.append style={
grid=both,
grid style={densely dotted,line width=.1pt, draw=gray!10},
major grid style={line width=.2pt,draw=gray!50},
axis x line=middle,    
axis y line=middle,    
ticks=both,
xtick={0,1,2,3,4,5},
ytick={2,4,6,8,10,12},
yticklabels={2,4,6,8,10,12},
xticklabels={0,1,2,3,4,5},
xlabel near ticks,
ylabel near ticks,
xlabel={Time $t$},
xmajorgrids,
ylabel={Number of active replicas $N(t)$},
ymajorgrids,
}
}
\begin{tikzpicture}
\begin{axis} [
xmin=-.5,xmax=6, ymin=-1, ymax=13, unbounded coords=jump]   
    \addplot[blue, thick, mark=][const plot] coordinates{(0,4)(2,9)(4,12)(5,12) (5,0)(6,0)};


\end{axis}

\draw[|<->|, densely dashed, thin] (0.52,1.5) -- (5.799,1.5) node [above, text width=3cm,align=center,midway] {\textbf{$S_1$}};

\draw[|<->|, densely dashed, thin] (1.2,0.4) -- (1.2,2.05) node [text width=3cm, align=center, midway] {\textbf{$n_0$}};

\draw[|<->|, densely dashed, thin] (4.2,2.045) -- (4.2,4.08) node [text width=3cm, align=center, midway] {\textbf{$n_1$}};

\draw[|<->|,densely dashed, thin] (5.2,4.07) -- (5.2,5.3) node [text width=3cm, align=center, midway] {\textbf{$n_2$}};

\fill[blue] (0.5199,0.41) circle (0.6 mm) node[above right] {$t_0$};
\fill[blue] (2.639,0.41) circle (0.6 mm) node[above] {$t_1$};
\fill[blue] (4.75,0.41) circle (0.6 mm) node[above] {$t_2$};
\end{tikzpicture}
\caption{
We illustrate the two-forking for a single task with total number of servers $N=12$,  
by plotting the time-evolution of number of active replicas $N(t)$. 
We consider the example when the sequence of number of forked servers is $(n_0, n_1, n_2) = (4,5,3)$, 
the sequence of forking times is $(t_0, t_1, t_2) = (0,2,4)$, 
and the service completion time is $S_1 = 5$.  
For this case, the server utilization cost $W = n_0S_1 + n_1(S_1 - t_1) + n_2(S_1-t_2)$. 
}
\label{Fig:Numberofserversforked}
\end{figure}
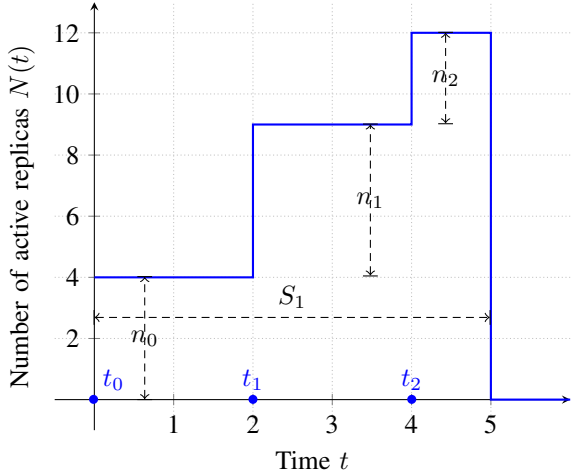

The results of single forking point analysis show that starting with multiple copies per job at time $t_0=0$ can perform much better than starting with a single copy per job as proposed in~\cite{Wang2015Sigmetrics}, 
when the forking time is below a certain threshold. 
Numerical evaluations show orders of magnitude improvement in the average server utilization cost for a fixed service completion time. 
The proposed framework thus shows that the single forking point strategies used in the literature may be significantly suboptimal, 
and one must judiciously select the number of servers to run at each forking time.   
Further, having more forking points help achieve a better tradeoff between the mean service completion time and the mean server utilization cost.  


\subsection{Related Work}
\label{sec:related}
It has been observed that task execution times have significant variability, 
partly due to resource sharing by multiple jobs~\cite{Dean2013ACM}.  
The slowest tasks that determine the job execution time are known as ``stragglers''. 
One of the key approaches to mitigate the effect of stragglers is to either re-launch a delayed task,  
or pre-emptively assigning each such task to multiple servers and taking the result of first completing server per task and canceling the same completed task at remaining servers. 
It is known that cancellation overhead can reduce the parallelism gains afforded by the additional servers~\cite{Lee2017TNET}. 
However, for simplicity of analysis and to obtain insight into optimistic performance gains, 
we assumed idealized assumption of negligible cancellation overhead. 
 
Speculative execution have been studied in~\cite{Dean2008ACM}, which acts after the tasks have already slowed down. 
Proactive approaches launch redundant copies of a task in a hope that at least one of them will finish in a timely manner. 
The authors of~\cite{Anantha2013NSDI} perform cloning to mitigate the effect of stragglers. 
The authors of~\cite{Wang2015Sigmetrics} analyzed the latency and cost for replication-based strategies for straggler mitigation. 
A machine learning approach for predicting and avoiding these stragglers has been studied in~\cite{Yadwadkar2016JMLR}. 

The problem of analyzing the completion of replicated parallel tasks is equivalent to having multiple redundant requests. 
The authors of~\cite{Gardner2015Sigmetrics} present an analysis of redundant requests where each job enters the queue at multiple servers. 
Service time completion can be generalized to finding mean waiting time of  
a stream of arriving redundant requests,  
and has been studied in the context of distributed storage.  
We note that the queueing studies for streaming arrival of requests exist only for fixed redundancy per request, and are difficult to characterize analytically even for this case. 
This implies that each job is forked to the identical number of servers, 
and job is completed by joining identical number of service completions. 
Tight numerical bounds are provided in~\cite{Shah2016TCOM}, 
analytical bounds are presented in~\cite{Joshi2014JSAC,Xiang2014Sigmetrics,Xiang2016TNET,al2018video}, 
analytical approximations appear in~\cite{Badita2019TIT}, 
exact analysis for small systems in~\cite{Gardner2016QueueingSystems},  
 exact analysis for random independent scheduling for asymptotically large number of servers in~\cite{wang2019delay}, and an exact analysis of tail index for Pareto-distributed file sizes in \cite{al2019ttloc}.

Even though we are not considering the streaming arrival of requests, 
our setting is a generalization of the fixed redundancy scheduling approach studied in the above-mentioned works, 
since the number of parallel servers available to each task is a time-varying function in our problem setting.   



\subsection{Main contributions} 
Our main contribution is the design of a multi-forking straggler mitigation policy that can efficiently trade-off mean service completion time and mean server utilization cost, 
by sequentially starting a number of replicas at forking points. 
The key contributions are summarized below. 
\begin{enumerate}

\item We analytically compute the mean service completion time and mean server utilization cost for any finite number of forking points when the completion time of each job on any server is independent and identically distributed according to a shifted exponential distribution with shift $c$ and rate $\mu$, 
and the inter-forking times $t_i -t_{i-1} \ge c$ for each $i \in [m] \triangleq \set{1, 2, \dots, m}$. 

\item For a single forking point, the mean service completion time and mean server utilization cost are analytically computed for all values of forking instants $t_1$, initial number of replicas $n_0$, and additional replicas $n_1$.  
We demonstrate that for single forking point $t_1$, having initial number of replicas $n_0=1$ is sub-optimal since both the performance metrics decrease with initial number of replicas $n_0 \le n_0^\ast$, 
where the inflection point $n_0^\ast \ge 1$ when the forking point $t_1 \le t_1^\ast$. 

\item Numerical results for multi-forking show orders of magnitude improvement in the tradeoff between the two metrics when compared to the baseline case of single-forking with single replica initialization of~\cite{Wang2015Sigmetrics}. 

\item We performed numerical studies for single and multi-forking when the job execution times are assumed to have heavy-tailed distributions such as Pareto and Weibull.  
We also studied single and multi-forking on a real compute cluster.  
We verified that the insights derived from the analytical studies for the shifted exponential distribution continue to hold in all three cases.

\end{enumerate}



\subsection{Organization}
\label{sec:Organization}
The rest of the paper is organized as follows. 
Section~\ref{sec:model} describes the model used in the paper. 
Section~\ref{sec:analysis}  provides the analytical results, 
where the mean service completion time and the mean server utilization are characterized for multiple forking points, 
with single forking being a special case. 
Section \ref{sec:OptSingleFork} explores further properties with single forking point.
Sections~\ref{Num_SingleFork} and \ref{sec:MultiFork}  provide a tradeoff between the mean service completion time and the mean server utilization for single and multiple forking points, respectively. 
We also compare our approach with that in~\cite{Wang2015Sigmetrics}. 
Section \ref{sec:IntelDevCloud} provides the experimental results on a real compute cluster, Intel DevCloud.
Section~\ref{sec:concl} concludes the paper, with directions for future work.




\section{System Model}
\label{sec:model}

We consider a distributed computation system with $K$ jobs and $KN$ identical servers, with the cost of server utilization $\lambda$ per unit time. 
Each server $n \in [KN] \triangleq \{1, \dots, KN\}$ has an independent and identically distributed (\emph{i.i.d.}) random service time $T_n$ with distribution function $F$ for each scheduled job on this server.
Uncertainty in execution time at various servers due to independent background processes, motivates our assumption of independently random execution time at each server. 
Identical distribution at each server is motivated primarily by analytical tractability, 
and the fact that we expect similar randomness at each identical server in a homogeneous cloud. 
Thus, following the existing literature~\cite{Xiang2016TNET,Shah2016TCOM,Wang2015Sigmetrics,al2019ttloc,Lee2017TNET,Badita2019TIT}, 
we adopted this commonly-used assumption for analysis. 

It has been shown in~\cite{Lee2018TIT, Bitar2017ISIT,Xiang2016TNET,al2018video} that shifted exponential well models the service time distribution in distributed computation networks. 
That is, it suggests that service time for each computation task can be modeled by aggregation of two components: a constant overhead and a random exponentially distributed component. 
Motivated by these studies together with the goal of analytical tractability, 
we assume the service time distribution to be a shifted exponential with rate $\mu$ and shift $c$, 
such that the complementary distribution function $\bar{F} = 1 - F$ can be written 
\EQN{
\label{eqn:ShiftedExp}
\bar{F}(x) \triangleq P\{t_0 > x\} = 
\begin{cases}
1, & x \in [0, c],\\
e^{-\mu(x-c)}, & x \ge c.
\end{cases} 
}

We assume that $KN$ servers are partitioned into $K$ disjoint sets of $N$ servers, where each set of $N$ servers can be utilized by a single job. 
The service completion time for job $k \in [K]$ sequentially scheduled over $N$ servers is denoted by $S_k$ and its server utilization cost is denoted by $W_k$. 
Then the service completion time for all $K$ jobs (also known as the makespan of the jobs) is the maximum of service completion times of all $K$ jobs, and is denoted by
\EQN{
\label{eqn:CompletionTime}
S = \max_{k \in [K]}S_k. 
}

Similarly, the average server utilization cost for $K$ independent jobs is defined as the average of server utilization cost for all $K$ jobs, and is denoted by
\EQN{
\label{eqn:UtilizationCost}
W = \frac{1}{K}\sum_{k \in [K]} W_k. 
}

We are interested in the optimal trade-off between mean service completion time $\E S$ and mean server utilization cost $\E W$ for $K$ jobs over these $KN$ servers. 
We will see that starting all the servers initially minimizes the mean service completion time, whereas it leads to maximum server utilization cost. 
Hence, we adopt an identical sequential policy for each of the $K$ jobs. 
A job $k \in [K]$ starts with $n_0$ parallel servers at time $t_0 = 0$, 
and sequentially adds $n_i$ servers at instant $t_i > t_{i-1}$ until we utilize all the $N$ servers. 
We let $m$ denote the number of sequential addition of servers such that $n_0 + \dots +n_m = N$. 

That is, we are considering $K$ parallel jobs, where each job is replicated on $N$ servers sequentially. 
Sequential addition of servers is motivated by the fact that service times are random and there is a cost associated with the on-time of each server. 
Hence, we should commission additional service only when absolutely necessary. 
For analytical tractability, we have further assumed $K$ parallel jobs to be uncoupled and we add extra servers in an identical fashion for each unfinished job at the same forking times. 
One can couple the $K$ jobs, by adding additional servers performing coded version of the tasks,
 such that  any $K$ task completions suffice~\cite{Joshi2014JSAC}. 
However, this can incur encoding and decoding delay of the computational tasks~\cite{Lee2017TNET}, and requires mixing of $K$ sub-tasks which may not always be desirable. 

We will consider the general case of $m \ge 1$, and find the mean service completion time and the mean server utilization cost for the case when the inter server addition interval $t_i - t_{i-1} \ge c$. 
Next, we will consider the specific case of single forking when $m = 1$ and $t_i - t_{i-1} > 0$. 

We note that the problem is important even when there are stochastic arrivals since this procedure of forking can be used for any arriving job. 
Even though the exact queueing analysis for multi-forking with stochastic arrivals remains open, 
we provide insights on sequential scheduling of $K$ initial jobs assigned to total $N$ servers each. 
In particular, the results in this paper  can provide an understanding of how many servers to use at each forking time to optimize the mean service completion time $\E S$ and the mean server utilization cost $\E W$.



\section{Analysis} 
\label{sec:analysis}

We observe that service completion time $S_k$ for each job $k \in [K]$ is independent due to independence of server completion times. 
Further, since we employ the identical forking strategy for each job, 
the service completion time $S_k$ for each job $k \in [K]$ has an identical distribution as well. 
From the \emph{i.i.d.} service completion times for individual job, it follows from~\eqref{eqn:CompletionTime} that $F_S(x) = F_{S_1}^K(x)$. 
From the positivity of service completions times, we have 
\EQN{
\label{eqn:MeanCompletionTime}
\E S = \int_{\R_+}\bar{F}_S(x)dx = \int_{\R_+}(1 - (1-\bar{F}_{S_1}(x))^K)dx. 
} 
From the similar arguments, we can conclude that the server utilization costs $(W_k: k \in [K])$ are \emph{i.i.d.}, and from the linearity of expectations, we have 
\EQN{
\label{eqn:MeanUtilizationCost}
\E W = \E W_1. 
}
It follows that we should first find the complementary distribution of service completion time $F_{S_1}(x)$ and the mean server utilization cost $\E W_1$ for any single task.


\subsection{Single Task}
\label{subsec:SingleTask}
At instant $t_i$, we switch on $n_i$ servers that continue being utilized until the service completion time $S_1$ for a single task. 
Hence, the total cost of server utilization in terms of service completion times $S_1$ for single task is 
\EQN{
\label{eqn:ServiceCost}
W_1 = \lambda\sum_{i =0}^mn_i (S_1 - t_i)_+,
}
where $(x)_+ \triangleq \max\{x, 0\}$. 

Let the time-interval $I_i \triangleq [t_i, t_{i+1})$ and we define $t_{m+1} = \infty$. 
Clearly, the disjoint intervals $I_i$ partition the positive reals and any $t \in \R_+$ belongs to a unique interval $I_i$ for some $i \in [m]_0 = \{0,1,\dots,m\}$.  
Let $t \in I_i$, then we have $n_\ell$ servers switched on at time $t_\ell$ for $l \le i$. 
The event that the service completion time is longer than duration $t$ is identical to the event that none of the servers started before this time $t$ have finished until this time $t$. 
Let $T_{\ell,p}$ denote the service completion time for the $p$th server started at time $t_\ell$, 
then for time $t \in I_i$ we can write 
\eq{
P\{S_1 > t \} &= P\bigcap\limits_{\ell=0}^i\{ \min_{p \in [n_\ell]}(T_{\ell,p} + t_\ell) > t\}\\
&= P\bigcap\limits_{\ell=0}^i\bigcap\limits_{p \in [n_\ell]}\{T_{\ell,p} > t -t_\ell \}. 
}
From the \emph{i.i.d.} service completion time for all servers,  we can write the complementary distribution function of service completion time $S_1$ as
\EQN{
\label{eqn:ServiceDist}
\bar{F}_{S_1}(t) = \prod_{\ell = 0}^i\bar{F}(t-t_\ell)^{n_\ell},~ t \in I_i.
} 
For a single task, we have $N_i \triangleq \sum_{\ell=0}^in_\ell$ servers working in parallel during the interval $[t_i, t_{i+1})$. 
If the task is unfinished until time $t_i$, 
then $n_\ell$ servers switched on at instant $t_\ell < t_i$ have been working on this task since then. 
Hence the server utilization until time $t_i$ is denoted by 
\EQN{
\label{eqn:tau}
\tau_i \triangleq \sum_{\ell = 0}^in_\ell(t_i-t_\ell). 
}
Shifted exponential distribution of server completion time $T_n$ defined in~\eqref{eqn:ShiftedExp}, is akin to a constant start-up time $c$ for the server after which the random service time $T_n-c$ is distributed exponentially with rate $\mu$. 
Hence, the servers switched on at time instant $t_i$ only begin the random part of the service at time $t_i+c$. 
Accordingly, we define shifted intervals $\tI_i \triangleq [t_i+c, t_{i+1}+c) = c + I_i$ where $N_i$ servers are working in parallel. 
In the following, we use the notation $[m]_0 = \set{0, 1, \dots, m}$. 


\begin{lem}
\label{lem:SingleTaskServiceDist}
Consider a single task being served by $N$ servers started sequentially at times $(t_j: j \in [m]_0)$ in batches of $(n_j: j \in [m]_0)$. 
When the job completion time for each server has an \emph{i.i.d.} shifted exponential distribution as defined in~\eqref{eqn:ShiftedExp}, then the complementary distribution of service completion time for a single task is given by 
\EQN{
\label{eqn:SingleTaskServiceDist}
\bar{F}_{S_1}(t) 
= e^{\left(-\mu N_i(t-t_i-c) - \mu\tau_i\right)},~t \in \tI_i. 
}
\end{lem}
\begin{proof}
Let $t \in \tI_i$, then from the definition of service completion time, we can write 
\EQ{
P\{S_1 > t\} = P\bigcap_{\ell =0}^i\bigcap_{p=1}^{n_\ell}\{T_{\ell,p}> c + (t-t_\ell-c)\}.
}
Since the job completion time at each server is \emph{i.i.d.} with the common shifted exponential distribution defined in~\eqref{eqn:ShiftedExp}, 
we get 
\EQ{
P\{S_1 > t\} = \exp(-\mu\sum_{\ell=0}^{i}n_\ell(t-t_\ell-c)),~t \in \tI_i.
}
The result follows from the definition of $\tau_i$ from equation~\eqref{eqn:tau}, 
and the definition of aggregate number of forked servers $N_i = \sum_{\ell = 0}^in_\ell$ at $i$th forking time $t_i$. 
\end{proof}



\begin{lem}
\label{lem:SingleTaskUtilization}
Consider a single task being served by $N$ servers started sequentially at times $(t_j: j \in [m]_0)$ in batches of $(n_j: j \in [m]_0)$. 
When the job completion time for each server has an \emph{i.i.d.} shifted exponential distribution as defined in~\eqref{eqn:ShiftedExp}, 
then the mean server utilization cost is given by 
\EQN{
\label{eqn:SingleTaskUtilization}
\E W_1= \lambda\sum_{i =0}^mn_i\int_{t_i}^{t_i+c}\bar{F}_{S_1}(t)dt + \lambda\sum_{i =0}^mN_i\int_{\tI_i}\bar{F}_{S_1}(t)dt. 
}
\end{lem}
\begin{proof}
From the equation~\eqref{eqn:ServiceCost} for the service utilization cost for a single task, 
the linearity of expectations, and positivity of random variables $(S_1-t_i)_+$, 
we can write the mean server utilization cost as 
\EQ{
\E W_1 = \lambda\sum_{i=0}^mn_i\E(S_1-t_i)_+ = \lambda\sum_{i=0}^mn_i\int_{t_i}^{\infty}\bar{F}_{S_1}(t)dt. 
}
We can write the integral over $[t_i, \infty)$ as the sum of integrals over its partition $\{[t_i, t_i+c),\tI_{i}, \tI_{i+1}, \dots, \tI_m\}$. 
Exchanging summations over indices $i \in [m]_0$ and $j \ge i$, we get the result.
\end{proof}
For general $m$, there is no straightforward way to evaluate the integral $\int_{t_i}^{t_i+c}\bar{F}_{S_1}(t)dt$ when $t_{i+1} - t_i \in (0, c)$. This is because the integration has to account for the servers started between $t_i$ and $t_i+c$, which makes the integral evaluation cumbersome.  
For simplicity, we stick with the case when $t_{i+1} - t_i \ge c$ for all $i \in [m]_0$. 
The results for $\E S$ and $\E W$ in this case will be provided in Corollary \ref{cor:SingleTaskMeans}.
However for the single forking case when $m=1$, we will derive the results when $t_1 -t_0 > 0$ and not necessarily larger than $c$ in Section \ref{sec:singlefork}.  




\subsection{Parallel Tasks}
Next, we find the mean of service completion time and the mean of server utilization cost for $K$ parallel tasks on $N$ servers each, using the complementary service distribution $\bar{F}_{S_1}$ for a single task, defined in~\eqref{eqn:SingleTaskServiceDist}. 
Formally, we describe our setup below.

\begin{prob} 
\label{prob:TimingGreaterC} 
Consider $K$ parallel tasks, where each single task is being served by $N$ servers starting in batches of $(n_j: j \in [m]_0)$, 
sequentially at times $(t_j: j \in [m]_0)$ such that the total number of servers is $N$ and timing thresholds are at least $c$ distance apart. 
That is, we have the following constraints, $t_0 = 0, t_{m+1} = \infty$, and 
\meq{2}{
&\sum_{j = 0}^m n_j = N, && t_{j+1} - t_j \ge c,~j \in [m]_0. 
} 
When the job completion time for each server has an \emph{i.i.d.} shifted exponential distribution as defined in~\eqref{eqn:ShiftedExp}, 
find the mean of the service completion time to finish all $K$ parallel tasks and the mean of the server utilization cost. 
\end{prob}

The time evolution of number of active replicas for a single task $S_k$ is illustrated in Fig. \ref{Fig:Numberofserversforked}. 
When a task is completed from any replica, the number of active replicas for that task becomes zero. 
The overall service completion time $S$ of the $K$ tasks is the maximum of the completion of each of the $K$ tasks, i.e. $S = \max_{k \in [K]}S_k$.
We need the following Lemma to evaluate the mean service completion time. 
\begin{lem} 
\label{lem:MultiTaskIntegralTerms}
We can write the following integrals for complementary distribution of service completion times. 
For $i \in [m]_0$, we have 
\EQN{
\label{eqn:MultiTaskStandardIntegralTerm}
\int_{t \in \tI_i}\bar{F}_{S}(t)dt = -\frac{1}{N_i\mu}\sum_{k=1}^K\binom{K}{k}\frac{(-1)^k}{k}\left(e^{-k\mu\tau_i} - e^{-k\mu\tau_{i+1}}\right).
}
For $1 \le i \le m$, we can write 
\EQN{
\label{eqn:MultiTaskLeadingIntegralTerm}
\int_{t_i}^{t_i+c}\bar{F}_{S}(t)dt = -\sum_{k=1}^K\binom{K}{k}\frac{(-e^{-\mu\tau_i})^k}{kN_{i-1}\mu}\left(e^{k\mu N_{i-1} c} - 1\right),
}
where the total number of active servers in interval $\tI_i$ is $N_i = \sum_{\ell = 0}^{i}n_\ell$ 
and server utilization until time $t_i+c$ is $\tau_i = \sum_{\ell = 0}^in_\ell(t_i-t_\ell)$.
\end{lem}
\begin{proof}
From the fact that $F_S(x) = F_{S_1}^K(x)$ and the binomial expansion of $(1-x)^K$, we can write 
\EQ{
\bar{F}_S(t) = 1 - (1- \bar{F}_{S_1}(t))^K = -\sum_{k=1}^K\binom{K}{k}(-1)^k\bar{F}^k_{S_1}(t).
}
Using the definition of single task service distribution in~\eqref{eqn:SingleTaskServiceDist} 
and definitions of $N_i$ and $\tau_i$, 
we can integrate $\bar{F}^k_{S_1}(t)$ over interval $\tI_i$, to get 
\EQ{
\int_{t \in \tI_i}\bar{F}^k_{S_1}(t)dt = \frac{1}{k N_i\mu}(e^{-k\mu\tau_i}-e^{-k\mu\tau_{i+1}}). 
}
To integrate $\bar{F}^k_{S_1}(t)$ over the interval $[t_i, t_i+c)$, 
we notice that $[t_i, t_i+c) \subseteq \tI_{i-1}$ since $t_{i-1}+c \le t_i$ by hypothesis. 
Therefore, we can write 
\EQ{
\int_{t_i}^{t_i+c}\bar{F}^k_{S_1}(t)dt = \frac{1}{k N_{i-1}\mu}(e^{-k\mu(\tau_i-N_{i-1}c)}-e^{-k\mu\tau_{i}}). 
}
The result follows from combining the above expressions. 
\end{proof}



\begin{cor}
\label{cor:MultiTaskIntegralTerms}
We can futher simplify the above integrals for complementary distribution of service completion times of Lemma \ref{lem:MultiTaskIntegralTerms} . 
For integers $0 \le i \in m$, we have 
\EQN{
\label{eqn:MultiTaskStandardIntegralTermSimplified}
\int_{t \in \tI_i}\bar{F}_{S}(t)dt =  \frac{1}{N_i\mu}\sum_{k=1}^K\frac{1}{k}\left((1-e^{-\mu\tau_{i+1}})^k- (1-e^{-\mu\tau_i})^k\right).
}
For $1 \le i \le m$, we can write 
\EQN{
\label{eqn:MultiTaskLeadingIntegralTermSimplified}
\int_{t_i}^{t_i+c}\bar{F}_{S}(t)dt =\frac{(e^{\mu N_{i-1} c} - 1)}{N_{i-1}\mu}\sum_{k=1}^K\frac{1}{k}(1-(1-e^{-\mu\tau_i})^k).
}
\end{cor}
\begin{proof}
We define the following integrals as a function of number of tasks
\meq{2}{
&h_1(K) = \int_{t \in \tI_i}\bar{F}_S(t)dt, &&h_2(K) = \int_{t=t_i}^{t_i+c}\bar{F}_S(t)dt.
}
We next observe the following identity for binomial coefficients
\EQ{
\frac{1}{k}\binom{K}{k} = \frac{1}{k}\binom{K-1}{k} + \frac{1}{K}\binom{K}{k},~ k \in [K]. 
}
Multiplying with a geometric term in $k$ and summing over all $k \in [K]$, we get 
\EQ{
-\sum_{k=1}^K\binom{K}{k}\frac{\alpha^k}{k} 
= -\sum_{k=1}^{K-1}\binom{K-1}{k}\frac{\alpha^k}{k} +\frac{1- (1+\alpha)^K}{K}.
}
Hence, we conclude that 
\eq{
h_2(K) &= h_2(K-1) + \frac{(1-e^{-\mu\tau_i})^K-(1-e^{-\mu(\tau_i-N_{i-1}c)})^K}{KN_{i-1}\mu},\\
h_1(K) &= h_1(K-1) + \frac{(1-e^{-\mu\tau_{i+1}})^K- (1-e^{-\mu\tau_i})^K}{KN_i\mu}.
}
The results follow by taking the summation of $h_1(k)$ and $h_2(k)$ over $k \in [K]$ with initial conditions $h_1(0) = h_2(0) = 0$.  
\end{proof}


Now, we have all the necessary results to compute the means of service completion time and cost server utilization for $K$ parallel tasks. 


\begin{thm}
\label{thm:general_forking}
For the Problem~\ref{prob:TimingGreaterC}, the mean service completion time is 
\EQN{
\label{eqn:MultiTaskMeanCompletionTime}
\E S 
= c + \frac{1}{\mu}\sum_{k=1}^K\frac{1}{k}\left(\frac{1}{N_m} + \sum_{i=1}^{m}\frac{n_i}{N_iN_{i-1}}(1-e^{-\mu \tau_i})^k\right),
}
and the mean server utilization cost is 
\EQN{
\label{eqn:SingleTaskMeanUtilizationCost}
\E W_1 = \lambda c n_0 + \frac{\lambda}{\mu} + \frac{\lambda}{\mu}\sum_{i = 1}^m n_ie^{-\mu\tau_{i}}\left(\frac{e^{\mu N_{i-1}c} - 1}{N_{i-1}}\right).
}
\end{thm}
\begin{proof}
We will first find the mean server utilization cost for single task. 
From~\eqref{eqn:SingleTaskUtilization}, we have 
\eq{
\frac{1}{\lambda}\E W_1 &= n_0\int_{t_0}^{t_0+c}\bar{F}_{S_1}(t)dt + \sum_{i=1}^mn_i\int_{t_i}^{t_i+c}\bar{F}_{S_1}(t)dt \nonumber\\
&+ \sum_{i=0}^mN_i\int_{\tI_i} \bar{F}_{S_1}(t)dt.
}
First, we notice that $\int_{t_0}^{t_0+c}\bar{F}_{S_1}(t) dt= c$ since $t_0 = 0$ and there is initial startup delay of $c$ for all shifted exponential job completion times. 
Taking $K=1$, and substituting equation~\eqref{eqn:MultiTaskStandardIntegralTerm} for integers $0 \le i \le m$ and equation~\eqref{eqn:MultiTaskLeadingIntegralTerm} for integer $1 \le i \le m$, in the above equation, we get 
\eq{
\frac{1}{\lambda}\E W_1 &= n_0c + \frac{1}{\mu}\sum_{i=1}^m\frac{n_ie^{-\mu\tau_{i}}}{N_{i-1}}(e^{\mu N_{i-1}c}-1) \\&+ \frac{1}{\mu}\sum_{i=0}^m(e^{-\mu\tau_i} - e^{-\mu\tau_{i+1}}). 
}
The result for mean server utilization cost follows from the telescopic sum and the fact that $\tau_0 = 0, \tau_{m+1} = \infty$.

To compute the mean of service completion time $S$, we use its positivity to write $\E S = \int_{\R_+}\bar{F}_S(t)dt$. 
By writing the integral over positive reals, as the sum of integrals over the partition $\{[0, t_0+c),\tI_0, \tI_1, \dots,\tI_m\}$, we get 
\EQ{
\E S = \int_{0}^{t_0+c}\bar{F}_S(t)dt + \sum_{i=0}^m \int_{\tI_i}\bar{F}_S(t)dt. 
}
Substituting the fact that $t_0 = 0, \tau_0 = 0, \tau_{m+1} = \infty$, $\int_{0}^{c}\bar{F}_{S_1}(t) = c$, and equation~\eqref{eqn:MultiTaskStandardIntegralTermSimplified} in the above equation, 
followed by exchanging summations over indices $k$ and $i$, we get the result. 
\end{proof}

As a special case of Theorem \ref{thm:general_forking}, we can obtain the mean service completion time and the mean server utilization cost for a single task, as is given in the following corollary.

\begin{cor}
	For a single task served by $N$ servers with multiple forks, the mean service completion time is 
	\EQN{
		\label{eqn:SingleTaskMeanCompletionTime}
		\E S = c + \frac{1}{N\mu} +\frac{1}{\mu}\sum_{i=1}^m\frac{n_i}{N_iN_{i-1}}(1 - e^{-\mu\tau_i}),
	}
	and the mean server utilization cost for single task is 
	\EQN{
		\label{eqn:SingleTaskMeanUtilizationCost1}
		\E W = \lambda c n_0 + \frac{\lambda}{\mu} + \frac{\lambda}{\mu}\sum_{i = 1}^mn_ie^{-\mu\tau_{i}}\left(\frac{e^{\mu N_{i-1}c} - 1}{N_{i-1}}\right).
	}
	\label{cor:SingleTaskMeans} 
\end{cor}


We show that making the forking instants smaller and increasing number of servers at any forking instant can reduce the service completion time, 
irrespective of the common service time distribution. 


\begin{prop}
\label{thm:Monotone}
For $K$ parallel tasks, each forked sequentially on $N$ identical servers with random \emph{i.i.d.} execution times with the common distribution function $F$, 
the following statements are true. 
\begin{enumerate}[(i)]
\item Consider two increasing sequences of forking times $t = (t_0, \dots, t_m)$ and $t' = (t'_0, \dots, t'_m)$ each with identical sequence of forked replicas such that $t'_i  \ge t_i$ at each stage $0 \le i \le m$. 
Then $\E S^{(t)} \le \E S^{(t')}$. 
\item Consider sequences of forked replicas $n = (n_0, \dots, n_m)$ and $n' = (n'_0, \dots, n'_m)$ with identical sequence of forking instants $t = (t_0, \dots, t_m)$ such that $n'_j \le n_j$ for stages $0 \le j \le m$. 
Then $\E S^{(n)} \le \E S^{(n')}$. 
\end{enumerate}
\end{prop}
\begin{proof}
The detailed proof is given in Appendix \ref{apdx:proof:stoc}, which uses stochastic dominance.
\end{proof} 


Second condition in the above theorem 
is very strict in that for a fixed forking time sequence $t$, 
the two forked replica sequence is such that the number of forked replicas at each forking time are always larger for one sequence.  
We would like the theorem to hold for the following weaker condition: 
for a fixed forking time sequence $t$ and the two server sequences $n, n'$ such that the cumulative number of server sequences $N \le N'$ are point-wise ordered. 
Notice that, in this case we would have to use specific properties of the service-time distribution at each server, 
and it links the forking instant sequence and the server sequence. 
In the following result, we will show that the result could be refined for the shifted exponential distribution. 


\begin{thm}
\label{thm:cor_par}
Let there be $K$ parallel tasks, 
each forked sequentially on $N$ identical servers with random \emph{i.i.d.} execution times 
with the common distribution function $F$ being the shifted exponential as defined in~\eqref{eqn:ShiftedExp}. 
Consider sequences of forked replicas $n = (n_0, \dots, n_m)$ and 
$n' = (n'_0, \dots, n'_m)$ with identical sequence of forking instants 
$t = (t_0, \dots, t_m)$ such that for each stage $0 \le i \le m$, 
\EQ{
\sum_{j=0}^in'_j \le \sum_{j=0}^in_j, \text{ and } \sum_{j=0}^in'_jt_j \ge \sum_{j=0}^in_jt_j. 
}  
Then $\E S^{(n)} \le \E S^{(n')}$. 
\end{thm}
\begin{proof}
Following the arguments in Theorem~\ref{thm:Monotone}, it suffices to show the monotonicity of the complementary distribution function of service times for single task. 
It follows from the theorem hypothesis 
that $N_i = \sum_{\ell=0}^in_\ell \ge \sum_{\ell=0}^in'_\ell = N'_i$ and $\tau'_i = \sum_{\ell=0}^in'_\ell(t_i-t_\ell) \le \sum_{\ell=0}^in_\ell(t_i-t_\ell) = \tau_i$ for all stages $i \in [m]_0$.  
Therefore, for any time $u \in \tI_i $, 	
\eq{
\bar{F}_{S_1^{(n)}}(u) &= e^{-\mu N_i(u-t_i-c)  - \mu\tau_i}\\
&\le e^{-\mu N'_i (u-t_i -c) - \mu\tau'_i} =  \bar{F}_{S_1^{(n')}}(u).
}
Hence, the result follows. 
\end{proof}


\begin{rem}
For single-fork case starting with forking points $0 = t_0 < t_1$, 
the condition $\sum_{j=0}^in'_jt_j \ge \sum_{j=0}^in_jt_j$ in Theorem \ref{thm:cor_par} 
reduces to $n'_1 \ge n_1$. 
Hence, if both the systems have identical number of servers, 
i.e. $n_0+n_1=n'_0+n'_1$, 
then $n_0'\le n_0$, and both the theorem conditions hold.  
\end{rem}

\subsection{Single Forking Parallel Tasks} \label{sec:singlefork}
We consider the single forking case for $K$ parallel tasks when $m=1$ and $t_1 > 0$. 
Formally, we define the problem below. 

\begin{prob} 
\label{prob:SingleFork} 
Consider $K$ parallel tasks, where each single task is being served by $N$ servers starting in two batches of $(n_0, n_1)$, sequentially at times $(0, t_1)$ such that the total number of servers is $N = n_0 + n_1$. 
When the job completion time for each server has an \emph{i.i.d.} shifted exponential distribution as defined in~\eqref{eqn:ShiftedExp}, 
find the mean of the service completion time to finish all $K$ parallel tasks and the mean of the server utilization cost. 
\end{prob}

Since $n_1 = N-n_0$, we have only two variables $n_0$ and $t_1$ in this case. 
Further, we have $t_2 = \infty$ and we can write $\tau = n_0t_1$. 
For the ease of further analysis, we would define following normalized constants.  
We define the amount of work done by all servers $N_1 = N$ in parallel each having independent random execution time distributed exponentially with rate $\mu$ in the shift-interval $c$ as $\alpha \triangleq c\mu N$. 
We denote the normalized forking time by $u \triangleq t_1/c$ and 
the initial fraction of servers by $x \triangleq n_0/N$.


\begin{thm}
\label{thm:SingleFork}
For the Problem~\ref{prob:SingleFork}, the scaled mean service completion time is 
\EQN{
\label{eqn:MultiTaskMeanCompletionTimeSingleFork}
\frac{1}{c}\E S = 1 + \frac{1}{\alpha}\sum_{k=1}^K\frac{1}{k}\left(1 + \frac{1-x}{x}(1-e^{-\alpha xu})^k\right),
}
and the scaled and shifted mean server utilization cost $\frac{\mu}{\lambda}\E W_1 -(1+\alpha)$ for single task equals
\EQN{
\label{eqn:SingleTaskMeanUtilizationCostSingleFork}
\begin{cases}
\alpha(1- x)\left(\frac{e^{-\alpha x (u-1)} - e^{-\alpha xu})}{ \alpha x}-1\right), & u \ge 1,\\
\alpha (1-x)\left(\frac{(1-e^{-\alpha x u})}{\alpha x}  -u\right), & u \le 1.
\end{cases}
}
\end{thm}
\begin{proof}
The result for the mean service completion time $\E{S}$ can be obtained by substituting $m=1$ in the equation \eqref{eqn:MultiTaskMeanCompletionTime}. 
To compute the mean server utilization cost for $t_1 \ge c$,  we substitute $m=1$ in the equation \eqref{eqn:SingleTaskMeanUtilizationCost}.  
For $t_1 < c$, we need to evaluate the integral $\int_{t_1}^{t_1+c}\bar{F}^k_{S_1}(t)dt$. 
In this case, we have 
\EQ{
\int_{t_1}^{t_1+c}\bar{F}_{S_1}(t)dt = \int_{t_1}^{c}\bar{F}_{S_1}(t)dt + \int_{c}^{t_1+c}\bar{F}_{S_1}(t)dt.
}
Since $\bar{F}_{S_1}(t) = 1$ for $t \le c$ due to initial startup delay $c$, 
and there are $n_0$ parallel independent servers working at the exponential rate $\mu$ in the interval $[t_1, t_1+c)$, 
we have 
\EQ{
\int_{t_1}^{t_1+c}\bar{F}_{S_1}(t)dt = c - t_1 + \frac{1}{n_0\mu}(1 - e^{-\mu n_0 t_1}). 
}
The result follows from aggregating both the cases. 
\end{proof}





\section{Optimal Single Forking}
\label{sec:OptSingleFork} 
We have the expression for mean of service completion and server utilization for single forking case in Theorem~\ref{thm:SingleFork}. 
We study the impact of forking time and initial number of servers on these two performance metrics. 


\begin{prop} 
\label{prop:PartialDer}
Consider the single forking for $K$ parallel tasks, 
each forked sequentially over $N$ parallel servers, 
each forked task having \emph{i.i.d.} random service times with the common shifted exponential  distribution with shift $c$ and rate $\mu$.  

The partial derivative of the mean service completion time with respect to normalized forking time $u$ is
\EQ{
\pd{\E S}{u} 
= c(1-x)(1 - (1 - e^{-\alpha x u})^K).
}

The partial derivative of the mean service completion time with respect to the initial fraction of servers $x$ is  
\EQ{
\pd{\E S}{x} 
= -\frac{c}{\alpha x^2}\sum_{k=1}^K\frac{1}{k}(1-e^{-\alpha xu})^k+ \frac{u}{x}\pd{\E S}{u}. 
}

The partial derivative of the mean server utilization cost with respect to the normalized forking time $u$ is 
\EQ{
\pd{\E W_1}{u} = \begin{cases}
-\frac{\lambda}{\mu} \alpha (1-x)e^{-\alpha xu}(e^{\alpha x} - 1), & u \ge 1,\\
-\frac{\lambda}{\mu} \alpha (1-x)(1- e^{-\alpha xu}), & u \le 1.
\end{cases} 
}

The scaled partial derivative $\frac{\mu}{\alpha\lambda}\pd{\E W_1}{x}$ of the mean server utilization cost with respect to the initial fraction of servers $x$ equals 
\EQ{
 \begin{cases}
1 - e^{-\alpha x u}\left[(\frac{1}{x}-1)((u-1)e^{\alpha  x}-u) +\frac{(e^{\alpha  x}-1)}{\alpha x^2}\right], & u \ge 1\\
u  + (\frac{1}{x}-1)ue^{-\alpha  x u}-\frac{1}{\alpha x^2}(1-e^{-\alpha xu}),& u \le 1.
\end{cases}
}
\end{prop}

\begin{proof}
Results follow by taking partial derivatives of mean server utilization and mean service completion task with respect to normalized forking time $u$ and initial fraction of servers $x$. 
\end{proof}


Even though the initial fraction of servers $x$ lie in the set $\{\frac{1}{N}, \dots, 1\}$, 
we approximate it by a real number $x \in [\frac{1}{N},1]$ to get insight on the dependence of the above two performance metrics on this fraction.


\begin{thm}
\label{thm:OptSingleFork}
The following statements are true for the single forking problem. 
\begin{enumerate}[(i)]
\item The mean service completion time is an increasing function of forking time $t_1$. 
\item The mean service completion time is a decreasing function of initial fraction $x$.
\item The mean server utilization cost is a decreasing function of forking time $t_1$. 
\item There exists a unique optimal initial fraction of servers $x^\ast \in [\frac{1}{N}, 1]$ that minimizes the mean server utilization cost. 
For normalized forking time $u \ge v_3$, the optimal initial fraction is $x^\ast = 1/N$. 
For normalized forking time $u < v_3$, the optimal initial fraction is the unique solution to the following implicit equation, where $e^{\alpha xu}$ equals
\EQN{
\label{eqn:ImplicitFraction}
 \begin{cases}
(\frac{1}{x}-1)((u-1)e^{\alpha  x}-u)+\frac{(e^{\alpha  x}-1)}{\alpha x^2}, & u \in [1, v_3),\\
-(\frac{1}{x}-1)+\frac{1}{\alpha x^2u}(e^{\alpha xu}-1), & u  < v_3\wedge1,
\end{cases}
}
where the normalized forking point threshold $v_3$ is the unique solution to the implicit equation, 
where $\frac{c\mu}{N} e^{c\mu v_3} + \frac{(N-1)c\mu}{N}$ equals
\EQ{
\begin{cases}
\left(\frac{c\mu(N-1)(v_3-1)}{N}+1\right)(e^{c\mu}-1),&(1-\frac{c\mu}{N})\frac{(e^{c\mu}-1)}{c\mu} > 1,\\
\frac{1}{v_3}(e^{c\mu v_3} -1), & (1-\frac{c\mu}{N})\frac{(e^{c\mu}-1)}{c\mu} \le 1.
\end{cases}
}

\end{enumerate}
\end{thm}
\begin{proof} 
	The proof is provided in Appendix \ref{apdx:OptSingleFork}. 
\end{proof}


\section{Numerical studies: Single Forking}\label{Num_SingleFork}

We numerically evaluate the behavior of mean service completion time and mean server cost utilization for single forking below, 
with total number of servers $N = 12$ for each of the $K = 10$ parallel tasks, taking $\lambda = 1$. 
We have analytically studied the case when service time at each server is an \emph{i.i.d.} random variable having a shifted exponential distribution, 
and we present the numerical studies for the single forking case with shifted exponential distribution. 
We note that the insights obtained from this study hold for heavy-tailed distribution such as Pareto distribution as well, and the supporting numerical results are presented in Appendix~\ref{subsec:ParetoSFP}. 



We have taken the job completion times at each server to be an \emph{i.i.d.} random variable having a shifted exponential distribution with the shift parameter $c = 8$ and the exponential rate $\mu = 0.01$. 
From the discussion in Appendix~\ref{subsec:ImpactParam}, we observe that $c\mu = 0.08 < N-2 = 10 < x'$, and hence $cv_3 \ge 1$. 
Specifically, we can numerically compute the forking time threshold $v_3 \approx \frac{y}{c\mu} \approx 47$ where $e^{y} = 1 + 11y$ is satisfied by $y \approx 3.741$.  
That is, for any forking point $t_1 \le 47c$, we can have optimal number of initial servers $n_0^\ast \ge 1$. 
For the given system parameters, we plot the mean service completion time in Fig.~\ref{Fig:MeanServiceShiftedExpSFP} and mean server utilization in Fig.~\ref{Fig:MeanCostShiftedExpSFP} as a function of initial servers $n_0 \in \{1, \dots, 11\}$ for values of forking times in $t_1 \in \{c, 2c, \cdots, 9c\}$. 
We corroborate the analytical results obtained in Theorem~\ref{thm:OptSingleFork}, 
by observing that mean service completion time increases and the mean server utilization cost decreases with increase in the forking time $t_1$. 
We further observe that the optimal number of initial servers $n_0^\ast \ge 1$ for mean server utilization cost for different values of forking time $t_1$.  
In addition, we notice the decrease in the mean service completion time as the number of initial server $n_0$ increases. 

These results point to an interesting tradeoff between the two metrics. 
First observation is that forking time gives a true tradeoff between these two metrics. 
Second and more interesting observation is that there exist a minimum number of initial servers for each forking time, 
until which point we can decrease both the mean service completion time and the mean server utilization cost. 
This also points to the sub-optimality of single-forking with unit server in~\cite{Wang2015Sigmetrics}. 

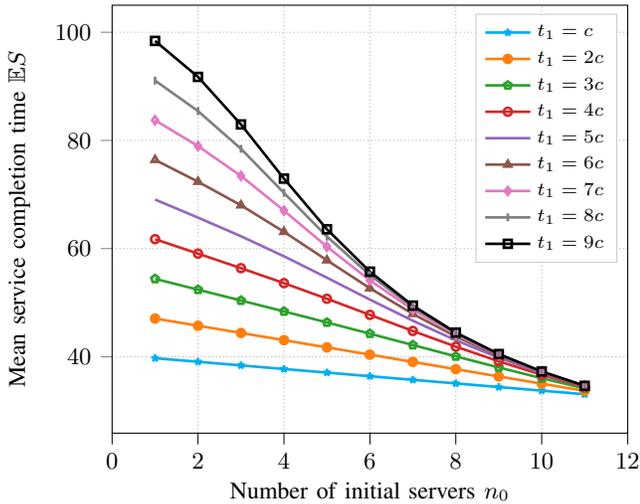
\begin{figure}[hhh]
\centering
\begin{tikzpicture}

\definecolor{color0}{rgb}{0.12156862745098,0.466666666666667,0.705882352941177}
\definecolor{color1}{rgb}{1,0.498039215686275,0.0549019607843137}
\definecolor{color2}{rgb}{0.172549019607843,0.627450980392157,0.172549019607843}
\definecolor{color3}{rgb}{0.83921568627451,0.152941176470588,0.156862745098039}
\definecolor{color4}{rgb}{0.580392156862745,0.403921568627451,0.741176470588235}
\definecolor{color5}{rgb}{0.549019607843137,0.337254901960784,0.294117647058824}
\definecolor{color6}{rgb}{0.890196078431372,0.466666666666667,0.76078431372549}
\definecolor{color7}{rgb}{0.737254901960784,0.741176470588235,0.133333333333333}

\begin{axis}[
font=\small,
legend cell align={left},
legend style={draw=white!80.0!black, font=\scriptsize},
tick align=outside,
tick pos=left,
x grid style={white!69.01960784313725!black, densely dotted},
xlabel={Number of initial servers $n_0$},
xmajorgrids,
xmin=0,xmax=12,
xtick style={color=black},
y grid style={white!69.01960784313725!black, densely dotted},
ylabel={Mean service completion time $\E S$},
ymajorgrids,
ymax=105,
]
\addplot [semithick, cyan, mark=star, mark size=1.5, mark options={solid}, line width=1pt]
table {%
1	39.7414021164
2	39.0747354465
3	38.4080686651
4	37.7414008012
5	37.0747278397
6	36.4080399522
7	35.7413215718
8	35.0745587627
9	34.4077567532
10	33.7409672971
11	33.0743223684
};
\addlegendentry{$t_1=c$}
\addplot [semithick, color1, mark=*, mark size=1.5, mark options={solid}, line width=1pt]
table {%
1	47.0747354426
2	45.7413988285
3	44.4079822904
4	43.0740287016
5	41.7383583811
6	40.3993569707
7	39.0560788708
8	37.7095849697
9	36.3637877082
10	35.0254658401
11	33.7035572006
};
\addlegendentry{$t_1=2c$}
\addplot [semithick, color2, mark=pentagon, mark size=1.5, mark options={solid}, line width=1pt]
table {%
1	54.4080683505
2	52.4079246287
3	50.4052605146
4	48.3906451583
5	46.3497162425
6	44.2752255585
7	42.1769536167
8	40.0815025174
9	38.0245917959
10	36.0420010412
11	34.1631231974
};
\addlegendentry{$t_1=3c$}
\addplot [semithick, color3, mark=o, mark size=1.5, mark options={solid}, line width=1pt]
table {%
1	61.7413948831
2	59.0729685795
3	56.3819333459
4	53.6141335294
5	50.7298481825
6	47.7549362518
7	44.7725099217
8	41.8815910512
9	39.1634560957
10	36.668337999
11	34.4169237928
};
\addlegendentry{$t_1=4c$}
\addplot [semithick, color4, mark=, mark size=1.5, mark options={solid}, line width=1pt]
table {%
1	69.0746756566
2	65.730531633
3	62.2830276247
4	58.5820393536
5	54.6224783662
6	50.5777300737
7	46.6685007636
8	43.0587418228
9	39.8295736159
10	36.9966018018
11	34.5364881043
};
\addlegendentry{$t_1=5c$}
\addplot [semithick, color5, mark=triangle, mark size=1.5, mark options={solid}, line width=1pt]
table {%
1	76.4077516434
2	72.3645097212
3	68.0095391092
4	63.1018037205
5	57.84559459
6	52.674230721
7	47.9268803048
8	43.758237418
9	40.1854490406
10	37.1550231554
11	34.5888431678
};
\addlegendentry{$t_1=6c$}
\addplot [semithick, color6, mark=diamond, mark size=1.5, mark options={solid}, line width=1pt]
table {%
1	83.7401617295
2	78.9441393969
3	73.4373850844
4	67.0285039713
5	60.3585154649
6	54.1344049135
7	48.7122974158
8	44.1516804104
9	40.3668548161
10	37.2285609054
11	34.6110579737
};
\addlegendentry{$t_1=7c$}
\addplot [semithick, white!49.80392156862745!black, mark=|, mark size=1.5, mark options={solid}, line width=1pt]
table {%
1	91.0708483351
2	85.4232306489
3	78.4486711892
4	70.3021578556
5	62.2299532945
6	55.1084060529
7	49.1846568222
8	44.3664150076
9	40.4571779648
10	37.2620931423
11	34.6203598682
};
\addlegendentry{$t_1=8c$}
\addplot [semithick, black, mark=square, mark size=1.5, mark options={solid}, line width=1pt]
table {%
1	98.3977717988
2	91.74385266
3	82.9567758989
4	72.940392659
5	63.578389081
6	55.7402095558
7	49.4626102823
8	44.4817325557
9	40.5016324944
10	37.2772603173
11	34.6242333286
};
\addlegendentry{$t_1=9c$}
\end{axis}

\end{tikzpicture}
\caption{
Mean service completion time $\E S$ as a function of initial number of servers $n_0 \in \{1, \dots, 11\}$ for single forking of $K = 10$ parallel tasks at different forking times $t_1 \in \{c, 2c, \dots, 9c\}$. 
}
\label{Fig:MeanServiceShiftedExpSFP}
\end{figure}

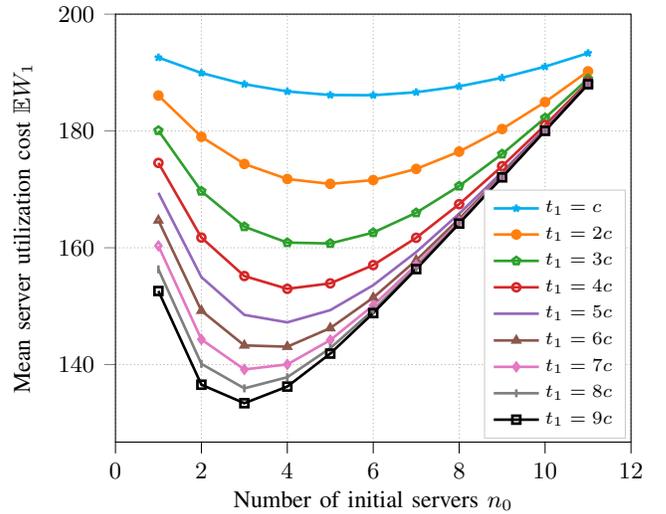
\begin{figure}[hhh]
\centering
\begin{tikzpicture}

\definecolor{color0}{rgb}{0.12156862745098,0.466666666666667,0.705882352941177}
\definecolor{color1}{rgb}{1,0.498039215686275,0.0549019607843137}
\definecolor{color2}{rgb}{0.172549019607843,0.627450980392157,0.172549019607843}
\definecolor{color3}{rgb}{0.83921568627451,0.152941176470588,0.156862745098039}
\definecolor{color4}{rgb}{0.580392156862745,0.403921568627451,0.741176470588235}
\definecolor{color5}{rgb}{0.549019607843137,0.337254901960784,0.294117647058824}
\definecolor{color6}{rgb}{0.890196078431372,0.466666666666667,0.76078431372549}
\definecolor{color7}{rgb}{0.737254901960784,0.741176470588235,0.133333333333333}

\begin{axis}[
font=\small,
legend cell align={left},
legend style={at={(0.99,0.59)},draw=white!80.0!black, font=\scriptsize},
tick align=outside,
tick pos=left,
x grid style={white!69.01960784313725!black, densely dotted},
xlabel={Number of initial servers $n_0$},
xmajorgrids,
xmin=0,xmax=12,
xtick style={color=black},
y grid style={white!69.01960784313725!black, densely dotted},
ylabel={Mean server utilization cost $\E W_1$},
ymajorgrids,
ymax=200
]

\addplot [semithick, cyan, mark=star, mark size=1.5, mark options={solid}, line width=1pt]
table {%
1	192.572018974701
2	189.928105516894
3	188.011641680034
4	186.770192585262
5	186.15519355501
6	186.121660819386
7	186.627924010799
8	187.635378797848
9	189.108258134668
10	191.013420717656
11	193.320155348349
};
\addlegendentry{$t_1=c$}
\addplot [semithick, color1, mark=*, mark size=1.5, mark options={solid}, line width=1pt]
table {%
1	186.069813162467
2	178.99737594626
3	174.353340778124
4	171.771322606128
5	170.938751468578
6	171.589050583103
7	173.494947801841
8	176.462756179493
9	180.327483242595
10	184.948648922451
11	190.206709525987
};
\addlegendentry{$t_1=2c$}
\addplot [semithick, color2, mark=pentagon, mark size=1.5, mark options={solid}, line width=1pt]
table {%
1	180.067520689624
2	169.682822633775
3	163.609340753851
4	160.879907613587
5	160.738865308703
6	162.596512729299
7	165.993272755974
8	170.571516916142
9	176.053421254802
10	182.223571294105
11	188.915305402424
};
\addlegendentry{$t_1=3c$}
\addplot [semithick, color3, mark=o, mark size=1.5, mark options={solid}, line width=1pt]
table {%
1	174.526706392149
2	161.745483881546
3	155.157810995458
4	152.971117104384
5	153.901677148457
6	157.032079655177
7	161.708247975726
8	167.465111084353
9	173.973011940131
10	180.999114986221
11	188.379653039895
};
\addlegendentry{$t_1=4c$}
\addplot [semithick, color4, mark=, mark size=1.5, mark options={solid}, line width=1pt]
table {%
1	169.411890141857
2	154.981729962913
3	148.509602218873
4	147.228156491708
5	149.318572866126
6	153.588900884094
7	159.260602982431
8	165.827126823247
9	172.960368012895
10	180.448931301793
11	188.157473593317
};
\addlegendentry{$t_1=5c$}
\addplot [semithick, color5, mark=triangle, mark size=1.5, mark options={solid}, line width=1pt]
table {%
1	164.690319652448
2	149.218039071055
3	143.279935969024
4	143.057911172861
5	146.246426192608
6	151.458319045528
7	157.862485977177
8	164.963430131664
9	172.467461296828
10	180.201717836794
11	188.065317355549
};
\addlegendentry{$t_1=6c$}
\addplot [semithick, color6, mark=diamond, mark size=1.5, mark options={solid}, line width=1pt]
table {%
1	160.331760753059
2	144.306545676036
3	139.166134792814
4	140.029691550219
5	144.187104692987
6	150.13995038894
7	157.063868871455
8	164.508009409521
9	172.227537840805
10	180.090637666651
11	188.027092522918
};
\addlegendentry{$t_1=7c$}
\addplot [semithick, white!49.80392156862745!black, mark=|, mark size=1.5, mark options={solid}, line width=1pt]
table {%
1	156.308303786343
2	140.121247084923
3	135.930104172718
4	137.830752787191
5	142.806700210559
6	149.324165759965
7	156.607691542122
8	164.267869512983
9	172.110754557328
10	180.040726128866
11	188.011237515541
};
\addlegendentry{$t_1=8c$}
\addplot [semithick, black, mark=square, mark size=1.5, mark options={solid}, line width=1pt]
table {%
1	152.594184891385
2	136.554770885536
3	133.384552327686
4	136.233995521833
5	141.88138741435
6	148.819371780265
7	156.347118916884
8	164.141245564828
9	172.053910030637
10	180.018299429296
11	188.004661129416
};
\addlegendentry{$t_1=9c$}
\end{axis}

\end{tikzpicture}
\caption{
Mean server utilization cost $\E W$ as a function of initial number of servers $n_0 \in \{1, \dots, 11\}$ for single forking of $K = 10$ parallel tasks at different forking times $t_1 \in \{c, 2c, \dots, 9c\}$. 
}
\label{Fig:MeanCostShiftedExpSFP}
\end{figure}


%

The authors of~\cite{Wang2015Sigmetrics} considered a single fork analysis where at $t=0$, one copy of the task is started and when $pn$ jobs are complete, each unfinished job is replicated $r$ times. 
The analysis considered two possibilities, one where the currently running job is kept running at the forking point and second where it is killed. 
It was shown that keeping the currently running job performed better for both mean service completion time and mean server utilization cost, when the service distribution is shifted exponential. 
We compare our results with the baseline results obtained in the~\cite[Theorem~2]{Wang2015Sigmetrics} for the case when the straggler job is kept running at the forking point. 
We restate the above-mentioned Theorem, adapted to our notation, for easy reference. 

\begin{lem}\cite[Theorem~2]{Wang2015Sigmetrics}
\label{lem:SingleForkSingleStart}
Consider $K$ parallel computing tasks, each started on a single server each, i.e. $t_0=0, n_0=1$. 
If $r$ replicas of each unfinished task are started, after $(1-p)K$ tasks are completed, 
and the execution time of each task is assumed to be \emph{i.i.d.} ShiftedExp$(c,\mu)$, 
then the mean service completion time and the mean server utilization cost metrics for $K \to \infty$,  are  
\eq{
\E S &= \frac{2r+1}{r+1}c + \frac{1}{(r+1)\mu}(\ln K - r \ln p + \gamma_{EM})\\
\E W &= c + \frac{1}{\mu} + pc + pr\frac{(1-e^{-\mu c})}{\mu},
}
where $\gamma_{EM}\approx 0.577$ is the Euler-Mascheroni constant.
\end{lem}

Though our model is quite different than the one studied here, we will make broad comparisons. 
We let $n_0 =1$ for this model and let $t_1$ to be the mean time to finish $(1-p)K$ tasks with $K$ parallel servers working at rate $\mu$. 
Then
\EQ{
{\mu K (t_1-c) \approx K (1-p)}. 
}
Further, at instant $t_1$, we have $n_1 = N-1 = rp$ new servers being started per job. 
Therefore, we can take the forking point to be $t_1 = c + \frac{(1-p)}{\mu}$ and the total number of servers to be $N = 1 + rp$. 
Given total number of available servers $N$ and forking time $t_1$, 
we can compute the fraction of completed tasks $p = 1 - \mu(t_1-c)$ and the number of replicas $r = \frac{(N-1)}{p}$.  
In Fig.~\ref{Fig:TradeoffShiftedExpSFP}, we have plotted the mean of service completion times with respect to mean server utilization cost when $\lambda =1$ for the single forking proposed in~\cite{Wang2015Sigmetrics} as the baseline curve and our proposed single forking varying the initial number of servers $n_0 \in \{1, \dots, 11\}$, when the forking time $t_1 \in \{2c, 4c, 6c, 8c\}$. 
We see that our trade-off curves are well inside the baseline curve. 
Specifically, we observe significant reduction in the mean server utilization cost for optimal server initialization when compared to single-server initialization of~\cite{Wang2015Sigmetrics}, for the identical mean service completion time in both the cases. 

\begin{figure}[hhh]
\centering
\begin{tikzpicture}

\definecolor{color0}{rgb}{0.12156862745098,0.466666666666667,0.705882352941177}
\definecolor{color1}{rgb}{1,0.498039215686275,0.0549019607843137}
\definecolor{color2}{rgb}{0.172549019607843,0.627450980392157,0.172549019607843}
\definecolor{color3}{rgb}{0.83921568627451,0.152941176470588,0.156862745098039}
\definecolor{color4}{rgb}{0.580392156862745,0.403921568627451,0.741176470588235}
\definecolor{color5}{rgb}{0.549019607843137,0.337254901960784,0.294117647058824}
\definecolor{color6}{rgb}{0.890196078431372,0.466666666666667,0.76078431372549}
\definecolor{color7}{rgb}{0.737254901960784,0.741176470588235,0.133333333333333}

\begin{axis}[
font=\small,
legend cell align={left},
legend style={at={(0.99,0.72)}, draw=white!80.0!black, font=\scriptsize},
tick align=outside,
tick pos=left,
x grid style={white!69.01960784313725!black, densely dotted},
xlabel={Mean service completion time $\E S$},
xmajorgrids,
xmin=38,xmax=100,
xtick style={color=black},
y grid style={white!69.01960784313725!black, densely dotted},
ylabel={Mean server utilization cost $\E W_1$},
ymajorgrids,
ymin=125,ymax=215
]

\addplot [semithick, color1, mark=*, mark size=1.5, mark options={solid}, line width=1pt]
table {%
47.0747354426	186.069813162467
45.7413988285	178.99737594626
44.4079822904	174.353340778124
43.0740287016	171.771322606128
41.7383583811	170.938751468578
40.3993569707	171.589050583103
39.0560788708	173.494947801841
37.7095849697	176.462756179493
36.3637877082	180.327483242595
35.0254658401	184.948648922451
33.7035572006	190.206709525987
};
\addlegendentry{$t_1=2c$}
\addplot [semithick, color3, mark=o, mark size=1.5, mark options={solid}, line width=1pt]
table {%
61.7413948831	174.526706392149
59.0729685795	161.745483881546
56.3819333459	155.157810995458
53.6141335294	152.971117104384
50.7298481825	153.901677148457
47.7549362518	157.032079655177
44.7725099217	161.708247975726
41.8815910512	167.465111084353
39.1634560957	173.973011940131
36.668337999	180.999114986221
34.4169237928	188.379653039895
};
\addlegendentry{$t_1=4c$}
\addplot [semithick, color5, mark=triangle, mark size=1.5, mark options={solid}, line width=1pt]
table {%
76.4077516434	164.690319652448
72.3645097212	149.218039071055
68.0095391092	143.279935969024
63.1018037205	143.057911172861
57.84559459	146.246426192608
52.674230721	151.458319045528
47.9268803048	157.862485977177
43.758237418	164.963430131664
40.1854490406	172.467461296828
37.1550231554	180.201717836794
34.5888431678	188.065317355549
};
\addlegendentry{$t_1=6c$}
\addplot [semithick, white!49.80392156862745!black, mark=|, mark size=1.5, mark options={solid}, line width=1pt]
table {%
91.0708483351	156.308303786343
85.4232306489	140.121247084923
78.4486711892	135.930104172718
70.3021578556	137.830752787191
62.2299532945	142.806700210559
55.1084060529	149.324165759965
49.1846568222	156.607691542122
44.3664150076	164.267869512983
40.4571779648	172.110754557328
37.2620931423	180.040726128866
34.6203598682	188.011237515541
};
\addlegendentry{$t_1=8c$}
\addplot [semithick, color=black, mark=square, mark size=1.5, mark options={solid}, line width=1pt]
table {%
39.3298757749504	200.572018974701
45.3021475158061	199.932018974701
52.0602933843543	199.292018974701
59.7626698417928	198.652018974701
68.6199071247242	198.012018974701
78.9209734243304	197.372018974701
91.0779113700539	196.732018974701
105.708072672076	196.092018974701
123.799422150789	195.452018974701
};
\addlegendentry{Baseline}

\end{axis}

\end{tikzpicture}
\caption{
Mean server utilization cost $\E W$ as a function of mean service completion time $\E S$ when we vary the number of initial servers $n_0 \in \{1, \dots, 11\}$ for single forking of $K = 10$ parallel tasks at different forking times $t_1 \in \{2c, 4c, 6c, 8c\}$. 
The service distribution of each replica is assumed to be \emph{i.i.d.} shifted exponential with shift parameter $c = 8$ and exponential rate $\mu = .01$.  
We have plotted the same curve for initial servers $n_0 = 1$ for single forking for $K=10$ parallel tasks for these different values of forking times $t_1$, for the baseline $(r, p)$ model where $r = \frac{(N-1)}{p}$ and $p = 1 - \mu (t_1-c)$.
}
\label{Fig:TradeoffShiftedExpSFP}
\end{figure}
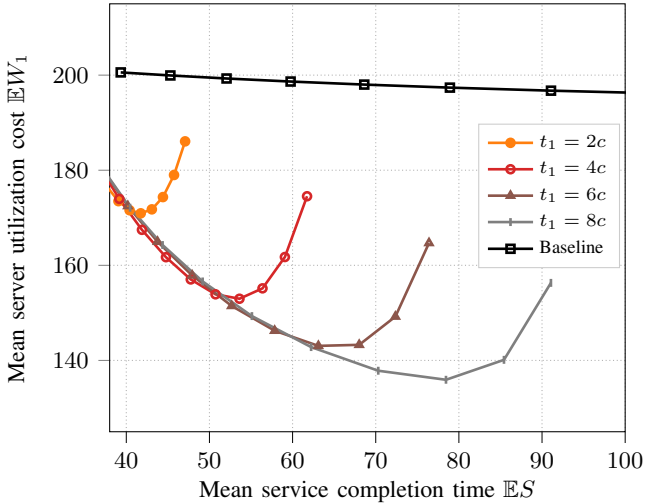




\section{Numerical Studies: Multiple Forking}
\label{sec:MultiFork}
In a multi-forking scenario the free variables are number of forked servers $(n_0, \dots, n_{m-1}, n_m)$ under the constraint of finite number of servers $N$ per task, i.e. $\sum_{i=0}^mn_i = N$, 
and the forking instants $(t_0=0, t_1, \dots, t_m)$. 
It is a multi-dimensional optimization problem and not easy to evaluate. 
The single forking results in Section~\ref{sec:OptSingleFork} leads us to believe that even for the general case of multiple forking points with \emph{i.i.d.} execution times having shifted exponential distribution, 
there should be a tradeoff between the two metrics of mean server utilization cost $\E W$ and the mean service completion time $\E S$. 
We attempt two approaches to understand this tradeoff.

\subsection{Joint Cost for large $N$}
To explore this tradeoff, we formulate the joint optimization in terms of a tradeoff parameter $\beta$ as  
\begin{eqnarray}
MP: && \min \E S + \beta \E W \label{opt_form}\\
& \text{such that} &  \eqref{eqn:MultiTaskMeanCompletionTime},  \eqref{eqn:SingleTaskMeanUtilizationCost}, t_0=0 \nonumber\\
& \text{variables} & n_0,  \cdots n_{m},  t_{0}, \cdots, t_m  \nonumber
\end{eqnarray}
We note from Fig.~\ref{Fig:TradeoffShiftedExpSFP} that based on the value of $\beta$, 
the tradeoff point chosen will be different. 
Thus, finding the forking instants and the number of servers added at each forking point, are important.  
For the optimization problem, we chose the total number of servers $N$, to be unbounded. 
For $(n_0, \dots, n_m)$ an integer sequence, the above problem is a mixed-integer programming problem, and known to be hard. 
As such, we relax the integer constraints and allow $n_i$ to be real valued, 
in which case the problem reduces to a linear programming problem and can be solved using interior point algorithm~\cite{Byrd1999SJO}. 
We round off the values of $n_i$ to nearest integers to get a heuristic integral solution. 


For this multi-objective optimization defined in~\eqref{opt_form}, 
changing the value of $\beta$ provides a tradeoff between the two metrics. 
For numerically solving the multi-objective optimization, 
we take the parameters of shifted exponential distribution as shift $c=1$ and service rate $\mu=1$, 
the server utilization cost per unit time $\lambda=1$, 
the number of parallel tasks $K=25$, and the number of forking points $m=4$.  
We depict the tradeoff between mean service completion time and mean server utilization cost for the proposed heuristic algorithm in  Fig.~\ref{Fig:tradeoffK25}. 

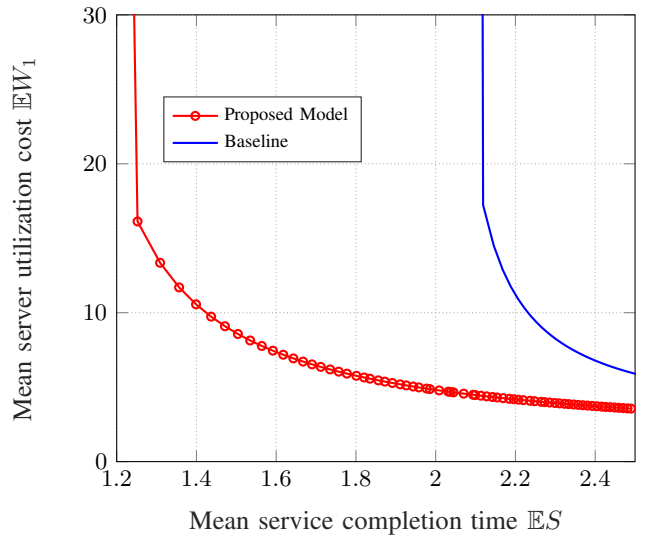
\begin{figure}[hhh]
\centering
%
%
\begin{tikzpicture}

\begin{axis}[%
font=\small,
width=.38\textwidth,
at={(1.011in,0.65in)},
scale only axis,
xmin=1.2,xmax=2.5,
xlabel style={font=\color{white!15!black}},
xlabel={Mean service completion time $\E S$},
xmajorgrids,
x grid style={white!69.01960784313725!black, densely dotted},
ymin=0,ymax=30,
ylabel style={font=\color{white!15!black}},
ylabel={Mean server utilization cost $\E W_1$},
ymajorgrids,
y grid style={white!69.01960784313725!black, densely dotted},
axis background/.style={fill=white},
legend style={at={(0.09,0.67)}, anchor=south west, legend cell align=left, align=left, draw=white!15!black, font=\scriptsize}
]
\addplot [color=red,mark=o,thick, mark size=1.5]
  table[row sep=crcr]{%
1.00917623883405	416.852098735475\\
1.25227370932592	16.12626182045\\
1.30891255001364	13.3528751987717\\
1.35659803205205	11.7010073939055\\
1.39900041899084	10.5637949136272\\
1.4368174283858	9.73581942868348\\
1.47151684717597	9.09294132446147\\
1.50405886735417	8.57046135842363\\
1.53478523305827	8.13549653538723\\
1.56390957840487	7.76696818749409\\
1.59170392223562	7.44910069762601\\
1.6177298615035	7.17738985187484\\
1.64318700432579	6.93289067098871\\
1.66742204497086	6.7174589997561\\
1.6898522792524	6.53155841106533\\
1.71180715420784	6.36094383866832\\
1.73515879025206	6.19065843764254\\
1.75674426070497	6.04259942998588\\
1.77735860689073	5.9088775036643\\
1.80105181388769	5.76368458510118\\
1.82056068458248	5.65042774946927\\
1.83514373746858	5.56922324449205\\
1.85624883399265	5.45659956093562\\
1.87251646550195	5.37350849927893\\
1.89136993712275	5.28100390061144\\
1.91103534836442	5.18859507980321\\
1.92629194670911	5.11960634136355\\
1.94372249597603	5.04351723517373\\
1.95778100327758	4.98416565438304\\
1.9771036262667	4.90537715249598\\
1.98548765306008	4.8721521937167\\
2.00883083819036	4.78255504612763\\
2.03166838405504	4.69882244784334\\
2.03745944704775	4.67817574803445\\
2.04460929078445	4.65300041973356\\
2.07019101135234	4.5656795256406\\
2.09376319145498	4.48883397020418\\
2.09904186627793	4.47207717444888\\
2.11315028332057	4.42807097555288\\
2.12779173991109	4.38356634711127\\
2.14235333724661	4.34043596955226\\
2.15383734115338	4.30718901320221\\
2.16915821551749	4.26385097161539\\
2.18447775087502	4.22163770053289\\
2.1957083286734	4.19137877208008\\
2.20749308527746	4.16023190876239\\
2.21975464748735	4.12846373293953\\
2.23665471763329	4.08571028219182\\
2.24821394887919	4.05713470119787\\
2.26457922365106	4.01757146274657\\
2.2732190879187	3.99709469792006\\
2.28592293812841	3.967485892445\\
2.29910985709269	3.93736373164738\\
2.31048724080631	3.91186213710568\\
2.32360389209382	3.88300616255641\\
2.33395000989965	3.86064556342953\\
2.34611207453824	3.83479975396933\\
2.35746637088651	3.81108855401744\\
2.36911195458585	3.78717760434382\\
2.38069451154672	3.76379615170078\\
2.39247781209588	3.74040860436228\\
2.40433435072126	3.71727183467715\\
2.41647782364804	3.69397664687183\\
2.42632716669644	3.67537369180968\\
2.43761954936651	3.65435885268111\\
2.44856056496393	3.63431041132272\\
2.45908192423459	3.61531454422696\\
2.47025299723267	3.59544322296988\\
2.48115392812755	3.57634139519312\\
2.48989888184028	3.56121957279944\\
};
\addlegendentry{Proposed Model}

\addplot [color=blue,mark=,thick, mark size=1.5]
  table[row sep=crcr]{%
2.00133942524772	1331.80835994475\\
2.11891527874686	17.2310324186215\\
2.14563864061222	14.5035899987267\\
2.16816744010092	12.8777136593177\\
2.18801593415112	11.7681660399904\\
2.20596023938237	10.94912369424\\
2.22246169821821	10.3125658739148\\
2.23782090028425	9.7994510134975\\
2.25224636847176	9.37443634692747\\
2.26589020239503	9.01488111241357\\
2.27886753827335	8.70554593486093\\
2.29126714443009	8.4357340643761\\
2.30316032527967	8.19768658367201\\
2.31460417730449	7.98561761129356\\
2.32564598855372	7.79512846629463\\
2.33632539084899	7.62279058014798\\
2.34667599414625	7.46589142711712\\
2.3567263936138	7.3222558898486\\
2.36650124933118	7.19011520933922\\
2.37602208487298	7.06801329119694\\
2.38530774926655	6.95474000011643\\
2.39437448762842	6.84928443028041\\
2.40323736045078	6.75078009915608\\
2.41190987775265	6.65849839885789\\
2.42040388663571	6.57181013709964\\
2.42872943171512	6.49017149740688\\
2.43689633632504	6.41311158229782\\
2.44491372735986	6.3402143190143\\
2.45278961586135	6.27111730501971\\
2.4605286793933	6.20552288879224\\
2.4681411228854	6.14311751489969\\
2.4756320209753	6.08365859031337\\
2.48300648262594	6.02692542632587\\
2.49026954283002	5.97271744737874\\
2.49742702506606	5.92084626128082\\
2.5044835375175	5.8711478538989\\
2.51144186645306	5.8234837613931\\
2.51830673295001	5.77771402909861\\
2.5250819537112	5.73371544827206\\
2.53177122229897	5.69137498231266\\
2.53837818476752	5.65058831062116\\
2.54490478916308	5.61126884196312\\
2.55135185230865	5.57334233323531\\
2.55772602710245	5.53670680927167\\
2.56402813311735	5.50129957288359\\
2.57026059065509	5.46705332243257\\
2.57642565682058	5.43390597858012\\
2.58252556461106	5.40179951553617\\
2.58856195256567	5.37068254589162\\
2.59453788875647	5.34049961560555\\
2.60045386405097	5.31121131398537\\
2.60631059234159	5.28277781586902\\
2.61211295654588	5.2551463166038\\
2.61786023728095	5.22828874898072\\
2.623555021493	5.20216484210081\\
2.62919850304385	5.17674280546362\\
2.63479079384921	5.15199720137027\\
2.64033488759505	5.12789167994156\\
2.64583207885107	5.10439851573969\\
2.65128234053783	5.08149781008077\\
2.65669272273854	5.05914077304214\\
2.66205077795945	5.03735908249693\\
2.6674108292116	5.015919933161\\
2.67907859147558	4.97001223359574\\
2.69167395596555	4.92119126243714\\
2.70407492818838	4.87385762532696\\
2.71627967502164	4.82797352614435\\
2.72829827557106	4.78345862589434\\
2.74013553250708	4.74025545913885\\
2.75179645241834	4.69830827910769\\
};
\addlegendentry{Baseline}

\end{axis}
\end{tikzpicture}%
\caption{Tradeoff between mean service completion time and mean server utilization cost, obtained by changing the value of $\beta$. 
The service distribution of each replica is assumed to be \emph{i.i.d.} shifted exponential with shift parameter $c = 1$ and exponential rate $\mu = 1$. 
}
\label{Fig:tradeoffK25}
\end{figure}

We compare the performance of multi-forking obtained by the proposed heuristic algorithm to the baseline single-forking approach proposed in~\cite{Wang2015Sigmetrics}. 
We can compute the linear cost of the optimization problem in~\eqref{opt_form} for any tradeoff parameter $\beta$, by obtaining the mean service completion time and the mean server utilization cost from Lemma~\ref{lem:SingleForkSingleStart}. 
Fig.~\ref{Fig:tradeoffK25} shows a significant improvement of the proposed model as compared to that in~\cite{Wang2015Sigmetrics} for the tradeoff between the two metrics. 
For a service completion time lower than $2$, there is significant reduction in the mean server utilization cost, 
thus showing the huge savings that the multi-forking can provide. 
The performance gains are due to two factors, 
initializing the task on multiple servers at time $t=0$ 
and multi-forking. 
We observed that the improvement due to multi-forking was small in this setting and the corresponding tradeoff curve for single forking looks very similar, 
and hence we do not provide the tradeoff curve for single-forking in this setting.  

We note that the lowest mean service completion point in Fig.~\ref{Fig:tradeoffK25} corresponds to starting large number of servers at $t=0$ since having large number of servers at $t=0$ achieves the lowest completion time. 
However, if the number of total servers is bounded by a number $N$ as is the case in our single forking analysis, 
the points on the very left in the mean service completion time may not be achievable. 
In other words, the curve will get truncated on the left side with an upper bound on $N$. 

\subsection{Comparison with optimal single forking}
\label{subsec:MultiVsSingle}
In general, finding the optimal forking points and the corresponding number of servers to be forked at each forking point, is not an easy task. 
In the following, we compare optimal single forking to sub-optimal two-forking to quantify potential gains of multi-forking for \emph{i.i.d.} shifted exponential execution times. 
We assume the system parameters to be $c = 8, \mu= 0.01, K=10, N=12, \lambda = 1$. 
We consider two different setups for the comparison, depending on the location of the other forking with respect to $t$. 
The single forking can be thought of as two-forking with zero forked servers at this other forking point. 

As a first case, we take the other forking point $s < t$. 
In this case, the single forking can be thought of as a two-forking sequence $((0, n_0),(s, 0),(t, N-n_0))$. 
For all possible values of $0 \le m_0, m_1$ such that $m_0+ m_1 \le N$, we consider two-forking sequences $((0, m_0),(s, m_1),(t, N-m_0-m_1))$. 
We plot the tradeoff curve between mean service completion time and  mean server utilization cost for the single and two-forking sequences in Fig.~\ref{Fig:TradeoffShiftedExpTFPBefore} for the values of forking points $t = 9c$ and $s \in  \{c,  2c, 3c, 4c, 5c\}$, varying the number of forked servers $n_0 \in [N]$ in single-forking case and $m_0, m_1$ in two-forking case. 
We observe that for \emph{some} feasible choice of forked servers $m_0, m_1$ and forking point $s < t$, 
the two-forking system achieves better tradeoff points as compared to the single-forking system. 

\begin{figure}[hhh]
\centering
\input{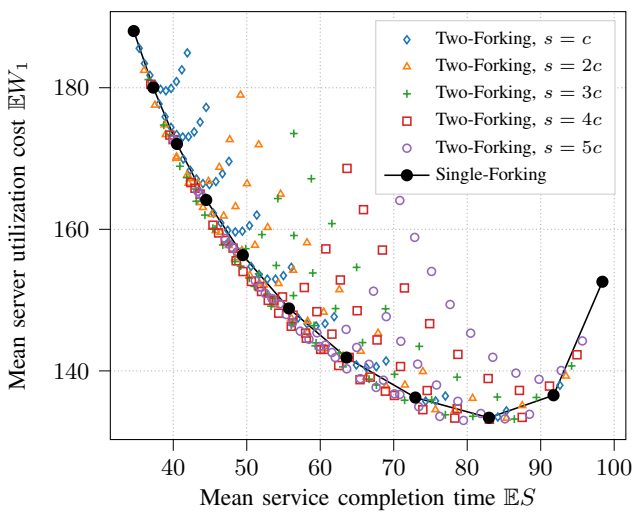}
\caption{ 
This figure illustrates  achievable points with shifted exponential
execution times for mean server utilization cost and mean service completion time for two-forking with different values of forked servers $m_0$, $m_1$ and $N-m_0-m_1$ at forking points $t_0=0$, $s \in \{c, 2c, 3c, 4c, 5c\}$ and $t = 9c$, respectively. 
For comparison, we also plot the tradeoff points for single-forking at forking time $t=9c$ varying the number of initial servers $n_0$. 
}
\label{Fig:TradeoffShiftedExpTFPBefore}
\end{figure}

For the other case, we take the second forking point $s > t$. 
In this case, the single forking can be thought of as a two-forking sequence $((0, n_0),(t, N-n_0),(s, 0))$. 
For all possible values of $0 \le m_0, m_1$ such that $m_0+ m_1 \le N$, we consider two-forking sequences $((0, m_0),(t, m_1),(s, N-m_0-m_1))$. 
We plot the tradeoff curve between mean service completion time and  mean server utilization cost for the single and two-forking sequences in Fig.~\ref{Fig:TradeoffShiftedExpTFPAfter} for the values of forking points $t = 9c$ and $s \in \{10c, 12c, \dots , 18c\}$, varying the number of forked servers $n_0 \in [N]$ in single-forking case and $m_0, m_1$ in two-forking case. 
We observe that for \emph{any} choice of forked servers $m_0, m_1$ and forking point $s > t$, 
the two-forking system achieves better tradeoff points as compared to the single-forking system. 

\begin{figure}[hhh]
\centering
\input{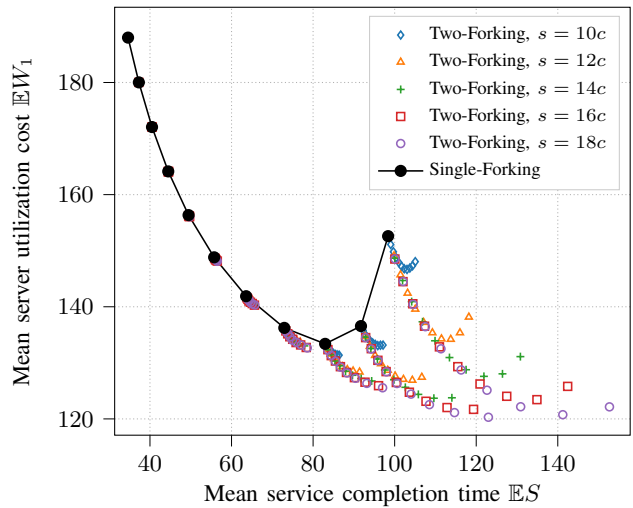}
\caption{ 
This figure illustrates achievable points with shifted exponential execution times for mean server utilization cost and mean service completion time for two-forking with different values of forked servers $m_0$, $m_1$ and $N-m_0-m_1$ at forking points $t_0=0$, $t=9c$ and $s \in \{10c, 12c, \dots , 18c\}$, respectively. 
For comparison, we also plot the tradeoff points for single-forking at forking time $t=9c$ varying the number of initial servers $n_0$. 
}
\label{Fig:TradeoffShiftedExpTFPAfter}
\end{figure}

Looking closely, we observe that setting the other forking point $s < t$ in two-forking can achieve better tradeoff points for the mean service completion time below a threshold. 
In contrast, setting the other forking point $s > t$ helps two-forking achieve significantly better tradeoff points when the mean service completion time is above that threshold. 
We also remark that at this threshold, the mean server utilization cost is minimum for single-forking. 
Hence, two-forking can further reduce the mean server utilization cost when compared to single-forking. 
Thus, an investigation of optimal forking points and the number of forked-servers at the different forking points is an important future research direction. 


\section{Experiments on Intel DevCloud servers}
\label{sec:IntelDevCloud}

Intel DevCloud is  a cloud computing service made available by Intel~\cite{intel} for several profiles of researchers, students and professional engineers.  
Intel DevCloud is a compute cluster, consisting of multiple servers called compute nodes, storage servers, and the login node. 
Each node has Intel Xeon processor of the Skylake architecture (Intel Xeon Scalable Processors family), an Intel Xeon Gold 6128 CPU, on-platform memory of 192 GB and a Gigabit Ethernet interconnect. 
To maximize the utilization of the compute cluster, one can submit jobs either by running Jupyter Notebook session on one of the compute nodes or by accessing the login node using an SSH client in a text-based terminal to a job queue dedicated to the authenticating account. 
For best performance, we created new environments with core Python 3 using Intel Distribution for Python. 
When a job is submitted to the queue, the scheduler picks the first available compute node for that job. 


\subsection{Setup}
In our experiment, 
we reserved one node per job and submitted multiple single-node jobs at forking points  
by launching a distributed-memory parallel job explicitly requesting multiple compute nodes, which correspond to the servers on which the job is forked. 
This ensures that all the forked jobs start at the same forking time on the compute nodes to which the jobs are forked. 
Single node jobs are submitted through a job script file using the \emph{qsub} command. 
We submitted a parallel job using the command \emph{mpirun}, which launches the single node job at multiple nodes, by creating MPI program using Message Passing Interface (MPI) library called Intel MPI which is installed on all nodes. 
From here after, in this section,  we refer \emph{parallel job} to \emph{replicated single-node jobs} which is \emph{requesting multiple compute nodes at once}. 

\subsection{Objective}
In this experimental set up, we have $K$ jobs. 
For both single-forking and two-forking, we take each of the $K$ jobs to  be an identical algebraic computation with approximate mean completion time of 600 seconds. 
As a test-case, each algebraic job is taken as the repeated addition of two numbers in a loop, that runs \num{6e9} times. 
This section aims at answering the following questions. 
\begin{enumerate}
\item Given $K$ jobs, $KN$ servers, and a forking mechanism, is it possible to get a tradeoff between the avearge server utilization cost and average service completion time on real cloud setup? 
\item  Are the tradeoff curves for this practical setup qualitatively similar to the one predicted by the analytical study for random execution times modeled to be distributed as a shifted exponential?
\end{enumerate}


\subsection{Experiment}
To cater to this requirement, we initialized a parallel job requesting $n_0$ compute nodes at time $t_0= 0$ for each of the $K$ jobs.
In the single-forking experiment, at  time $t_1$ seconds, we initialized a parallel job requesting $n_1$ compute nodes for each of the unfinished jobs and waited for the completion. 
Similarly, in the two-forking experiment, at times $t_1,t_2$ seconds, we again initialized parallel jobs requesting $n_1,n_2$ compute nodes, respectively, for each of the unfinished jobs and waited for the completion. 
As soon as one of the replica of a $i$th job is finished we logged that time stamp $S_i$ into a log file and killed the other replicas of that particular $i$th job immediately using the \emph{qdel} command. 

Using the observed job completion times $(S_i)$ and forking time and server sequences, 
we compute the two performance metrics: the server utilization cost and the service completion time, by using the equations~\eqref{eqn:ServiceCost}, \eqref{eqn:UtilizationCost}, and \eqref{eqn:CompletionTime} for each run $j \in [J]$ for $J = \num{1e4}$ runs. 
In addition, we also computed the empirical distribution of job completion times, which is plotted in Fig.~\ref{Fig:fitDistIntelDevCloud}. 

\subsection{Evaluations}
We evaluate the single forking setup on Intel DevCloud with $K=10$ parallel tasks and with each task being replicated on total number of $N=12$ servers. 
We ran this experiment $J = \num{1e4}$ times for the given system parameters. 
For the $k$th task in $j$th run, we denote the service completion time  and the server utilization cost by $S^{(j)}_k$ and $W^{(j)}_k$ respectively. 
Hence, we computed the empirical average of service completion time and server utilization costs as 
\meqn{2}{
&\hat{S} \triangleq \frac{1}{J}\sum_{j=1}^J \max_{k \in [K]}S^{(j)}_k,&
&\hat{W} \triangleq \frac{1}{J}\sum_{j=1}^J \frac{1}{K}\sum_{k=1}^KW^{(j)}_k.
}
We plot the empirical average of service completion time $\hat{S}$ in Fig.~\ref{Fig:MeanServiceIntelDevCloudSFPjob600IFT10} and empirical average of server utilization cost $\hat{W}$ in Fig.~\ref{Fig:MeanCostIntelDevCloudSFPjob600IFT10} as a function of initial servers $n_0 \in \{1, \dots, 11\}$ for values of forking times  $t_1 \in \{10, 20, \cdots, 90\}$ seconds. 
In Fig.~\ref{Fig:TradeoffIntelDevCloudSFPjob600IFT10}, we plot the empirical average of service
completion times with respect to empirical average of server utilization cost when $\lambda=1$ for the single forking varying the initial number of servers $n_0 \in \{1, \dots, 11\}$ for the forking times $t_1 \in \{10, 20, \cdots, 90\}$ seconds. 
We also evaluate the two-forking setup on Intel DevCloud with the same parameters.
Given the first forking point at $t$, the second forking point at $s>t$, two-forking sequences $((0, m_0),(t, m_1),(s, N-m_0-m_1))$, the tradeoff is plotted in Fig.~\ref{Fig:TradeoffIntelDevCloudTFPjob600IFT10}.

\subsection{Results}
From Fig.~\ref{Fig:fitDistIntelDevCloud}, 
we observe that the empirical distribution of the job execution times at each node has characteristics of a shifted exponential distribution.  
The empirical distribution has a distinct constant shift corresponding to the start delay, 
and the random part of the job execution time doesn't have long tails.

\begin{figure}[th]
\centering
\scalebox{0.9}{\input{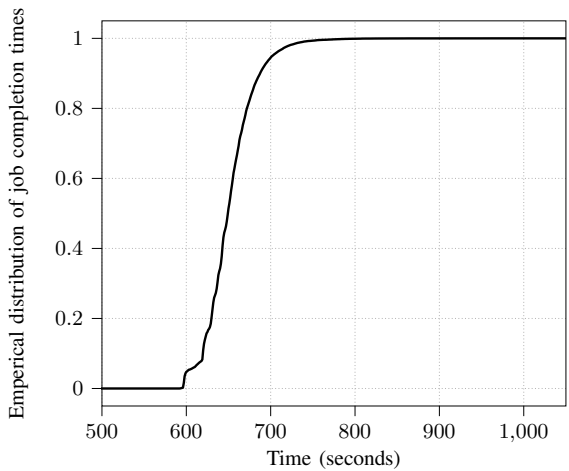}}
\caption{ 
This figure illustrates the empirical distribution of job completion times that are collected during the Intel DevCloud experiments. 
The job here is a algebraic computation of addition of two numbers, repeated \num{6e9} times.}
\label{Fig:fitDistIntelDevCloud}
\end{figure}

We substantiate the analytical results obtained in Theorem~\ref{thm:OptSingleFork} for single-forking, by observing that the mean service completion time $\E{S}$ decreases with increase in initial number of servers $n_0$. 
Further, the tradeoff suggests that the number of initial servers $n_0$ is an important consideration for an efficient system design. 
The insights obtained in this experiment for two-forking are identical to those obtained from the shifted exponential service distribution. 
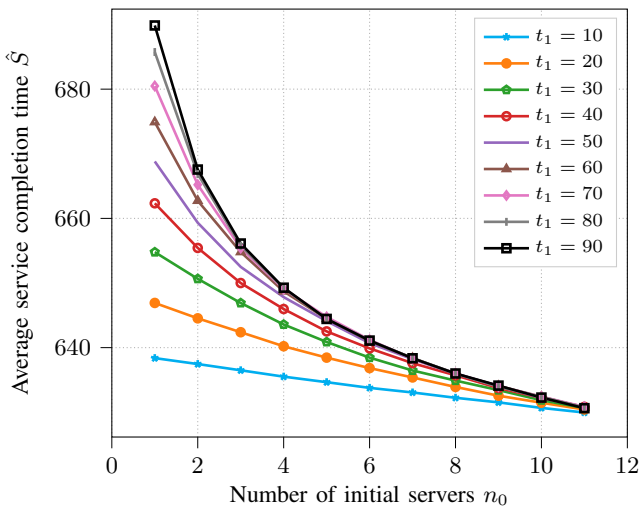
\begin{figure}[hhh]
\centering
\begin{tikzpicture}

\definecolor{color0}{rgb}{0.12156862745098,0.466666666666667,0.705882352941177}
\definecolor{color1}{rgb}{1,0.498039215686275,0.0549019607843137}
\definecolor{color2}{rgb}{0.172549019607843,0.627450980392157,0.172549019607843}
\definecolor{color3}{rgb}{0.83921568627451,0.152941176470588,0.156862745098039}
\definecolor{color4}{rgb}{0.580392156862745,0.403921568627451,0.741176470588235}
\definecolor{color5}{rgb}{0.549019607843137,0.337254901960784,0.294117647058824}
\definecolor{color6}{rgb}{0.890196078431372,0.466666666666667,0.76078431372549}
\definecolor{color7}{rgb}{0.737254901960784,0.741176470588235,0.133333333333333}
\definecolor{color8}{rgb}{0.0901960784313725,0.745098039215686,0.811764705882353}

\begin{axis}[
font=\small,
legend cell align={left},
legend style={draw=white!80.0!black, font=\scriptsize},
tick align=outside,
tick pos=left,
x grid style={white!69.01960784313725!black, densely dotted},
xlabel={Number of initial servers $n_0$},
xmajorgrids,
xmin=0, xmax=12,
xtick style={color=black},
y grid style={white!69.01960784313725!black, densely dotted},
ylabel={Average service completion time $\hat{S}$},
ymajorgrids,
ymin=626.167189548699, ymax=692.388325459635,
ytick style={color=black}
]
\addplot [semithick, cyan, mark=star, mark size=1.5, mark options={solid}, line width=1pt]
table {%
1 638.378658101251
2 637.460747963631
3 636.492769168902
4 635.509995835308
5 634.654544388505
6 633.772520556717
7 633.062316968387
8 632.237930033438
9 631.530410464861
10 630.671655688999
11 629.957007888997
};
\addlegendentry{$t_1=10$}
\addplot [semithick, color1, mark=*, mark size=1.5, mark options={solid}, line width=1pt]
table {%
1 646.914954018928
2 644.53527589139
3 642.384857798996
4 640.211981528136
5 638.452963917803
6 636.835530702397
7 635.385611631356
8 633.915689039942
9 632.560790596517
10 631.47323860812
11 630.426269325911
};
\addlegendentry{$t_1=20$}
\addplot [semithick, color2, mark=pentagon, mark size=1.5, mark options={solid}, line width=1pt]
table {%
1 654.780890738964
2 650.649936166435
3 646.907542452888
4 643.58176011062
5 640.876091818277
6 638.454402620934
7 636.46109371816
8 634.898456778064
9 633.425867317188
10 631.862001073811
11 630.659644858147
};
\addlegendentry{$t_1=30$}
\addplot [semithick, color3, mark=o, mark size=1.5, mark options={solid}, line width=1pt]
table {%
1 662.332681409663
2 655.427566809016
3 649.995944125862
4 645.97415824616
5 642.514121112189
6 639.88070366068
7 637.596766720947
8 635.70968209406
9 633.744290443865
10 632.224203564704
11 630.838035379965
};
\addlegendentry{$t_1=40$}
\addplot [semithick, color4, mark=, mark size=1.5, mark options={solid}, line width=1pt]
table {%
1 668.812468453207
2 659.323833570841
3 652.497636832661
4 647.817737948413
5 644.107329123314
6 640.697020319931
7 638.219186602921
8 635.945001253521
9 634.095201933457
10 632.213263906098
11 630.821864474689
};
\addlegendentry{$t_1=50$}
\addplot [semithick, color5, mark=triangle, mark size=1.5, mark options={solid}, line width=1pt]
table {%
1 674.899067694738
2 662.731592365307
3 654.789744278347
4 648.705937144218
5 644.465188797369
6 641.007159818525
7 638.39102940696
8 635.921757775856
9 634.072804906206
10 632.402629046642
11 630.722563854433
};
\addlegendentry{$t_1=60$}
\addplot [semithick, color6, mark=diamond, mark size=1.5, mark options={solid}, line width=1pt]
table {%
1 680.483991662043
2 665.204865576183
3 655.414538629173
4 649.066472696908
5 644.749471460323
6 641.208277378277
7 638.288406308252
8 636.006499900092
9 633.994862979797
10 632.406054602755
11 630.833990407614
};
\addlegendentry{$t_1=70$}
\addplot [semithick, white!49.80392156862745!black, mark=|, mark size=1.5, mark options={solid}, line width=1pt]
table {%
1 685.745231278547
2 666.870416359735
3 655.774699661082
4 649.178285305811
5 644.469112583095
6 641.102739098247
7 638.384537810041
8 636.008103315065
9 633.949838917115
10 632.366943022823
11 630.732266117769
};
\addlegendentry{$t_1=80$}
\addplot [semithick, black, mark=square, mark size=1.5, mark options={solid}, line width=1pt]
table {%
1 689.853191449812
2 667.558067478614
3 656.111264691326
4 649.273526051673
5 644.44208362044
6 641.092644525419
7 638.352126392742
8 635.998671789611
9 634.122405592344
10 632.2727418556
11 630.63006029068
};
\addlegendentry{$t_1=90$}
\end{axis}

\end{tikzpicture}
\caption{Empirical average of service completion time $\hat{S}$ as a function of initial number of servers $n_0 \in \{1, \dots, 11\}$ for single forking of $K = 10$ parallel tasks at different forking times $t_1 \in \{10, 20, \dots, 90\}$ when jobs are executed on Intel DevCloud.
}
\label{Fig:MeanServiceIntelDevCloudSFPjob600IFT10}
\end{figure}

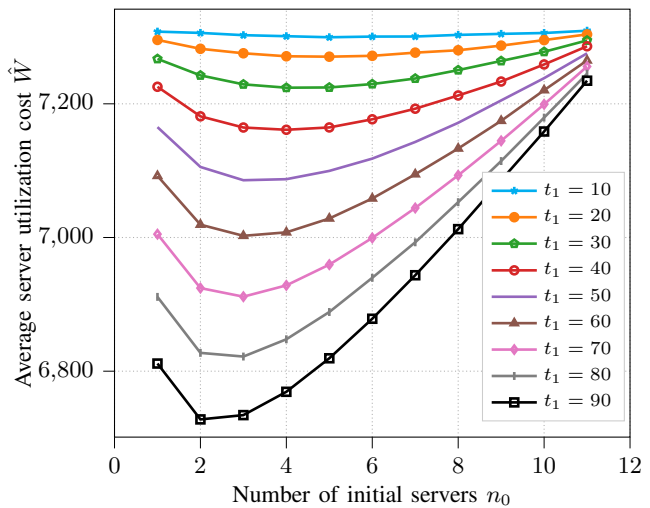
\begin{figure}[hhh]
\centering
\begin{tikzpicture}

\definecolor{color0}{rgb}{0.12156862745098,0.466666666666667,0.705882352941177}
\definecolor{color1}{rgb}{1,0.498039215686275,0.0549019607843137}
\definecolor{color2}{rgb}{0.172549019607843,0.627450980392157,0.172549019607843}
\definecolor{color3}{rgb}{0.83921568627451,0.152941176470588,0.156862745098039}
\definecolor{color4}{rgb}{0.580392156862745,0.403921568627451,0.741176470588235}
\definecolor{color5}{rgb}{0.549019607843137,0.337254901960784,0.294117647058824}
\definecolor{color6}{rgb}{0.890196078431372,0.466666666666667,0.76078431372549}
\definecolor{color7}{rgb}{0.737254901960784,0.741176470588235,0.133333333333333}
\definecolor{color8}{rgb}{0.0901960784313725,0.745098039215686,0.811764705882353}

\begin{axis}[
font=\small,
legend cell align={left},
legend style={at={(0.99,0.62)},draw=white!80.0!black, font=\scriptsize},
tick align=outside,
tick pos=left,
x grid style={white!69.01960784313725!black, densely dotted},
xlabel={Number of initial servers $n_0$},
xmajorgrids,
xmin=0, xmax=12,
xtick style={color=black},
y grid style={white!69.01960784313725!black, densely dotted},
ylabel={Average server utilization cost $\hat{W}$},
ymajorgrids,
ymin=6701.37330816515, ymax=7341.59697167051,
ytick style={color=black}
]
\addplot [semithick, cyan, mark=star, mark size=1.5, mark options={solid}, line width=1pt]
table {%
1 7307.9892764678
2 7306.07230436371
3 7302.69302813559
4 7301.15840131874
5 7299.56136304972
6 7300.58589821215
7 7300.69124736053
8 7303.11916501983
9 7304.63148378023
10 7305.94016694025
11 7309.35325389488
};
\addlegendentry{$t_1=10$}
\addplot [semithick, color1, mark=*, mark size=1.5, mark options={solid}, line width=1pt]
table {%
1 7295.58984964736
2 7282.31355027532
3 7275.60009098331
4 7271.13408964781
5 7270.50105096655
6 7271.79885272985
7 7276.56326113669
8 7280.16554598951
9 7287.06208608747
10 7295.35014518034
11 7303.90951062514
};
\addlegendentry{$t_1=20$}
\addplot [semithick, color2, mark=pentagon, mark size=1.5, mark options={solid}, line width=1pt]
table {%
1 7267.21126218898
2 7242.53747196635
3 7229.13545496
4 7224.11093104681
5 7224.47109561416
6 7229.56602110402
7 7237.87885499171
8 7250.4349804771
9 7264.12615521876
10 7277.82228469412
11 7294.83387922979
};
\addlegendentry{$t_1=30$}
\addplot [semithick, color3, mark=o, mark size=1.5, mark options={solid}, line width=1pt]
table {%
1 7225.49952730707
2 7181.26692130312
3 7164.54027415197
4 7161.16283729591
5 7164.66596088697
6 7176.84397149558
7 7192.6846103366
8 7212.49944537383
9 7233.34418277136
10 7258.89553793101
11 7286.08521111305
};
\addlegendentry{$t_1=40$}
\addplot [semithick, color4, mark=, mark size=1.5, mark options={solid}, line width=1pt]
table {%
1 7165.11043650016
2 7105.55174760541
3 7085.68494331664
4 7087.2510642473
5 7099.63305077226
6 7118.01600202755
7 7142.93741193424
8 7171.59002868464
9 7205.05774687772
10 7238.38127646829
11 7275.51622955031
};
\addlegendentry{$t_1=50$}
\addplot [semithick, color5, mark=triangle, mark size=1.5, mark options={solid}, line width=1pt]
table {%
1 7092.17095600215
2 7019.20937717582
3 7002.51060763006
4 7007.70965581953
5 7028.42661767629
6 7058.35660674568
7 7094.57179676663
8 7133.12551048772
9 7174.55982622696
10 7220.36695507708
11 7264.88378500596
};
\addlegendentry{$t_1=60$}
\addplot [semithick, color6, mark=diamond, mark size=1.5, mark options={solid}, line width=1pt]
table {%
1 7004.99955634886
2 6924.33922744118
3 6911.47552632934
4 6928.3704112745
5 6959.45256433905
6 6999.37910391562
7 7044.15660736222
8 7093.05039444712
9 7144.60708369146
10 7199.19513500962
11 7255.84807241737
};
\addlegendentry{$t_1=70$}
\addplot [semithick, white!49.80392156862745!black, mark=|, mark size=1.5, mark options={solid}, line width=1pt]
table {%
1 6911.16257179051
2 6827.40049082258
3 6821.89822626605
4 6847.74102167389
5 6888.54380907057
6 6939.7151166973
7 6992.82682945825
8 7052.87828554049
9 7114.37951805918
10 7179.36039554578
11 7245.02668916475
};
\addlegendentry{$t_1=80$}
\addplot [semithick, black, mark=square, mark size=1.5, mark options={solid}, line width=1pt]
table {%
1 6811.44200714113
2 6727.8719920613
3 6734.29896812102
4 6769.17575958707
5 6819.31446597747
6 6878.44263415537
7 6943.54518419414
8 7012.52238048582
9 7085.67452053799
10 7158.58063462978
11 7234.56531322731
};
\addlegendentry{$t_1=90$}
\end{axis}

\end{tikzpicture}
\caption{Empirical and task average of server utilization cost $\hat{W}$ as a function of initial number of servers $n_0 \in \{1, \dots, 11\}$ for single forking of $K = 10$ parallel tasks at different forking times $t_1 \in \{10, 20, \dots, 90\}$ when jobs are executed on Intel DevCloud.
}
\label{Fig:MeanCostIntelDevCloudSFPjob600IFT10}
\end{figure}

\begin{figure}[hhh]
\centering
\begin{tikzpicture}

\definecolor{color0}{rgb}{0.12156862745098,0.466666666666667,0.705882352941177}
\definecolor{color1}{rgb}{1,0.498039215686275,0.0549019607843137}
\definecolor{color2}{rgb}{0.172549019607843,0.627450980392157,0.172549019607843}
\definecolor{color3}{rgb}{0.83921568627451,0.152941176470588,0.156862745098039}
\definecolor{color4}{rgb}{0.580392156862745,0.403921568627451,0.741176470588235}
\definecolor{color5}{rgb}{0.549019607843137,0.337254901960784,0.294117647058824}
\definecolor{color6}{rgb}{0.890196078431372,0.466666666666667,0.76078431372549}
\definecolor{color7}{rgb}{0.737254901960784,0.741176470588235,0.133333333333333}
\definecolor{color8}{rgb}{0.0901960784313725,0.745098039215686,0.811764705882353}

\begin{axis}[
font=\small,
legend cell align={left},
legend style={draw=white!80.0!black, font=\scriptsize},
tick align=outside,
tick pos=left,
x grid style={white!69.01960784313725!black, densely dotted},
xlabel={Average service completion time $\hat{S}$},
xmajorgrids,
xmin=626.167189548699, xmax=692.388325459635,
xtick style={color=black},
y grid style={white!69.01960784313725!black, densely dotted},
ylabel={Average server utilization cost $\hat{W}$},
ymajorgrids,
ymin=6701.37330816515, ymax=7341.59697167051,
ytick style={color=black}
]
\addplot [semithick, cyan, mark=star, mark size=1, mark options={solid}, line width=1pt]
table {%
638.378658101251 7307.9892764678
637.460747963631 7306.07230436371
636.492769168902 7302.69302813559
635.509995835308 7301.15840131874
634.654544388505 7299.56136304972
633.772520556717 7300.58589821215
633.062316968387 7300.69124736053
632.237930033438 7303.11916501983
631.530410464861 7304.63148378023
630.671655688999 7305.94016694025
629.957007888997 7309.35325389488
};
\addlegendentry{$t_1=10$}
\addplot [semithick, color1, mark=*, mark size=1, mark options={solid}, line width=1pt]
table {%
646.914954018928 7295.58984964736
644.53527589139 7282.31355027532
642.384857798996 7275.60009098331
640.211981528136 7271.13408964781
638.452963917803 7270.50105096655
636.835530702397 7271.79885272985
635.385611631356 7276.56326113669
633.915689039942 7280.16554598951
632.560790596517 7287.06208608747
631.47323860812 7295.35014518034
630.426269325911 7303.90951062514
};
\addlegendentry{$t_1=20$}
\addplot [semithick, color2, mark=pentagon, mark size=1, mark options={solid}, line width=1pt]
table {%
654.780890738964 7267.21126218898
650.649936166435 7242.53747196635
646.907542452888 7229.13545496
643.58176011062 7224.11093104681
640.876091818277 7224.47109561416
638.454402620934 7229.56602110402
636.46109371816 7237.87885499171
634.898456778064 7250.4349804771
633.425867317188 7264.12615521876
631.862001073811 7277.82228469412
630.659644858147 7294.83387922979
};
\addlegendentry{$t_1=30$}
\addplot [semithick, color3, mark=o, mark size=1, mark options={solid}, line width=1pt]
table {%
662.332681409663 7225.49952730707
655.427566809016 7181.26692130312
649.995944125862 7164.54027415197
645.97415824616 7161.16283729591
642.514121112189 7164.66596088697
639.88070366068 7176.84397149558
637.596766720947 7192.6846103366
635.70968209406 7212.49944537383
633.744290443865 7233.34418277136
632.224203564704 7258.89553793101
630.838035379965 7286.08521111305
};
\addlegendentry{$t_1=40$}
\addplot [semithick, color4, mark=*, mark size=1, mark options={solid}, line width=1pt]
table {%
668.812468453207 7165.11043650016
659.323833570841 7105.55174760541
652.497636832661 7085.68494331664
647.817737948413 7087.2510642473
644.107329123314 7099.63305077226
640.697020319931 7118.01600202755
638.219186602921 7142.93741193424
635.945001253521 7171.59002868464
634.095201933457 7205.05774687772
632.213263906098 7238.38127646829
630.821864474689 7275.51622955031
};
\addlegendentry{$t_1=50$}
\addplot [semithick, color5, mark=, mark size=1, mark options={solid}, line width=1pt]
table {%
674.899067694738 7092.17095600215
662.731592365307 7019.20937717582
654.789744278347 7002.51060763006
648.705937144218 7007.70965581953
644.465188797369 7028.42661767629
641.007159818525 7058.35660674568
638.39102940696 7094.57179676663
635.921757775856 7133.12551048772
634.072804906206 7174.55982622696
632.402629046642 7220.36695507708
630.722563854433 7264.88378500596
};
\addlegendentry{$t_1=60$}
\addplot [semithick, color6, mark=triangle, mark size=1, mark options={solid}, line width=1pt]
table {%
680.483991662043 7004.99955634886
665.204865576183 6924.33922744118
655.414538629173 6911.47552632934
649.066472696908 6928.3704112745
644.749471460323 6959.45256433905
641.208277378277 6999.37910391562
638.288406308252 7044.15660736222
636.006499900092 7093.05039444712
633.994862979797 7144.60708369146
632.406054602755 7199.19513500962
630.833990407614 7255.84807241737
};
\addlegendentry{$t_1=70$}
\addplot [semithick, white!49.80392156862745!black, mark=diamond, mark size=1, mark options={solid}, line width=1pt]
table {%
685.745231278547 6911.16257179051
666.870416359735 6827.40049082258
655.774699661082 6821.89822626605
649.178285305811 6847.74102167389
644.469112583095 6888.54380907057
641.102739098247 6939.7151166973
638.384537810041 6992.82682945825
636.008103315065 7052.87828554049
633.949838917115 7114.37951805918
632.366943022823 7179.36039554578
630.732266117769 7245.02668916475
};
\addlegendentry{$t_1=80$}
\addplot [semithick, black, mark=square, mark size=1, mark options={solid}, line width=1pt]
table {%
689.853191449812 6811.44200714113
667.558067478614 6727.8719920613
656.111264691326 6734.29896812102
649.273526051673 6769.17575958707
644.44208362044 6819.31446597747
641.092644525419 6878.44263415537
638.352126392742 6943.54518419414
635.998671789611 7012.52238048582
634.122405592344 7085.67452053799
632.2727418556 7158.58063462978
630.63006029068 7234.56531322731
};
\addlegendentry{$t_1=90$}
\end{axis}

\end{tikzpicture}
\caption{
Empirical and task average of server utilization cost $\hat{W}$ as a function of empirical average of service completion time $\hat{S}$ by varying the number of initial servers $n_0 \in \{1, \dots, 11\}$ for single forking of $K = 10$ parallel tasks at different forking times $t_1$  when jobs are executed on Intel DevCloud.
}
\label{Fig:TradeoffIntelDevCloudSFPjob600IFT10}
\end{figure}
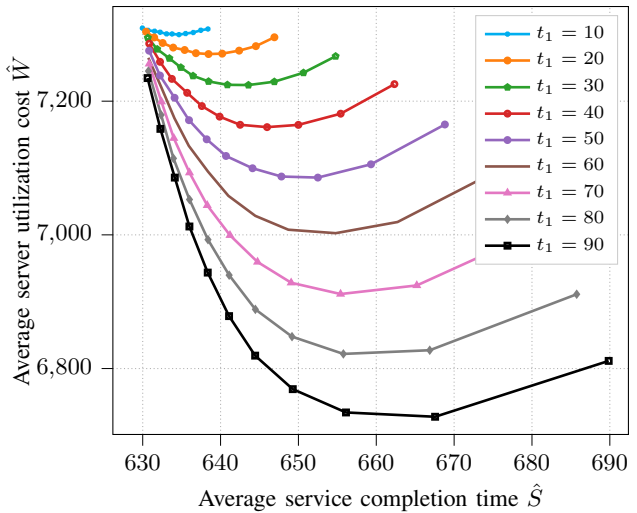

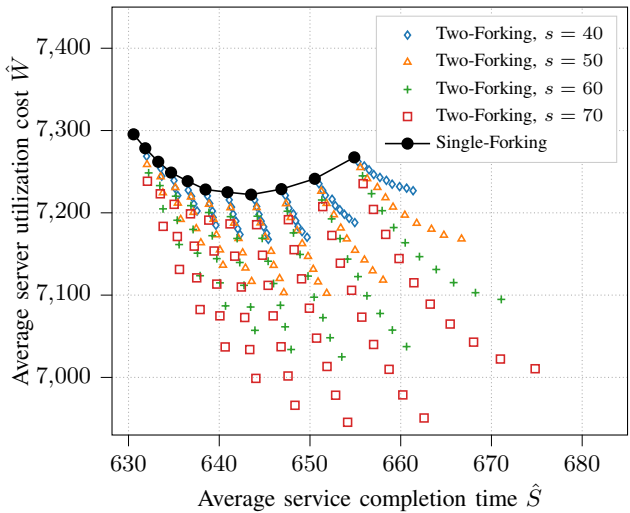
\begin{figure}[hhh]
\centering
\begin{tikzpicture}

\definecolor{color0}{rgb}{0.12156862745098,0.466666666666667,0.705882352941177}
\definecolor{color1}{rgb}{1,0.498039215686275,0.0549019607843137}
\definecolor{color2}{rgb}{0.172549019607843,0.627450980392157,0.172549019607843}
\definecolor{color3}{rgb}{0.83921568627451,0.152941176470588,0.156862745098039}
\definecolor{color4}{rgb}{0.580392156862745,0.403921568627451,0.741176470588235}

\begin{axis}[
font=\small,
legend cell align={left},
legend style={draw=white!80.0!black, font=\scriptsize},
tick align=outside,
tick pos=left,
x grid style={white!69.01960784313725!black, densely dotted},
xlabel={Average service completion time $\hat{S}$},
xmajorgrids,
xmin=628.180830227019, xmax=685,
xtick style={color=black},
y grid style={white!69.01960784313725!black, densely dotted},
ylabel={Average server utilization cost $\hat{W}$},
ymajorgrids,
ymin=6930, ymax=7450,
ytick style={color=black}
]
\addplot [semithick, color0, mark=diamond, mark size=1.5, mark options={solid}, only marks]
table {%
631.998068711223 7268.70724178534
633.547875248274 7254.11634699918
633.661052156203 7244.53828275843
635.01313748663 7239.3245968806
635.127633388791 7229.96504477318
635.430265252223 7220.97978009435
636.582156185993 7227.49948918544
636.948925378616 7220.01750509526
637.201100765683 7210.78074747088
637.555521030801 7202.24514464579
638.698303612767 7220.5128531866
638.947775976329 7211.39810296894
639.248115182594 7202.09518285697
639.326178645859 7194.32661105631
639.632311180294 7184.94555462526
641.10063123057 7215.71726293455
641.35150087591 7206.92498171507
641.528042654749 7198.40294192708
641.828671275151 7189.88210506727
642.043312992962 7180.50034360049
642.287621580695 7173.61820504659
643.915271521803 7215.65456292849
644.126622945895 7207.60797074336
644.384333392429 7197.93987291567
644.61900121776 7189.8872134322
644.899842186307 7182.91190360005
645.229148474105 7174.92507458789
645.407842980821 7167.72888720377
647.286952811038 7221.12489341405
647.522648905993 7213.12717563727
647.690202525458 7205.52452262451
648.169350954884 7198.66416574494
648.429426261805 7190.2787957014
648.682424860698 7183.4829217986
649.131251974839 7177.16669333035
649.683115280462 7170.30832668254
650.993447849442 7235.24626610889
651.401147248937 7227.31696645488
651.702535992513 7221.2636686503
652.437061025711 7216.47878161329
652.788087671708 7208.78731030326
653.289920639898 7203.71459098399
653.787552628409 7198.23770499978
654.314916016965 7193.30028073961
654.939351286743 7188.16205801411
655.294913166469 7262.03743492641
655.981907800922 7256.94244563207
656.511753953945 7252.22213999679
657.027434616602 7246.47601655434
657.710780003251 7242.97971745098
658.419961485911 7239.43213666639
659.058119790138 7234.78287821704
659.874874767337 7231.84191603338
660.652974965594 7229.97611963108
661.367347951027 7226.86831851319
};
\addlegendentry{Two-Forking, $s=40$}
\addplot [semithick, color1, mark=triangle, mark size=1.5, mark options={solid}, only marks]
table {%
632.035911977633 7258.71222430443
633.53748147722 7243.87098245753
633.786273064202 7224.38641599602
635.163037867853 7230.72347057505
635.432960485213 7211.81491579798
635.750371688489 7192.00314424099
636.876325210081 7219.24867636524
637.157239501322 7200.54106658691
637.467062782933 7180.86387397589
637.971959035016 7163.66502061524
638.866246593555 7210.43432508498
639.368015859799 7192.12095359307
639.687195351003 7173.28295026276
640.078568050524 7154.96667370052
640.412382903233 7136.3591927072
641.231422966218 7206.22658028821
641.632666176943 7187.43337652591
642.124109742804 7168.82837248258
642.580847782349 7151.49764713034
642.89139832111 7132.97170391833
643.567551794218 7116.69338310717
644.123463456521 7206.40622587513
644.543006259573 7187.8174238112
645.069836468582 7170.56797315389
645.713087805102 7153.25117619514
645.99469333673 7136.15813925332
646.653093765431 7120.21317759751
647.118496937812 7103.10609012653
647.356050039631 7211.15754392975
647.861104980843 7194.2705423229
648.580266580085 7178.46851663707
649.140547798689 7162.56718848658
649.782953042908 7146.99341456765
650.361435394088 7131.43963152543
650.888334697372 7116.09020178425
651.84281169979 7102.2531641866
651.245005882279 7226.25441944186
651.824231561946 7211.42517776496
652.634180489027 7197.65070114808
653.390058514378 7181.59605833268
654.059125472939 7168.32215459548
654.965732023032 7154.48404984606
656.004931479558 7142.78808807812
656.905069177307 7129.5705181833
658.058272966604 7118.26512836633
655.5602887963 7255.06915702618
656.300375285237 7241.51431059263
657.426883713111 7230.84544583028
658.320304695356 7218.38535206483
659.155063507154 7207.01124809102
660.453957456327 7197.21760890104
661.762555401291 7187.30661594633
663.389201984141 7180.11300655011
664.75973795422 7172.83685554474
666.703019661435 7168.29075104893
};
\addlegendentry{Two-Forking, $s=50$}
\addplot [semithick, color2, mark=+, mark size=1.5, mark options={solid}, only marks]
table {%
632.217543233144 7248.51464614548
633.45158376227 7233.08568843693
633.803945123972 7204.82331662645
635.204207398658 7220.19873338795
635.372411500497 7191.00096020642
635.579667388019 7161.45187059232
636.884279078307 7208.75878829764
637.119639197257 7179.80178791585
637.604831957662 7150.95894799583
637.890056130609 7123.57518085821
638.859925849191 7200.34689944118
639.233511072164 7172.18958639783
639.727001020912 7144.30092238535
640.079164247552 7114.89078951509
640.690751190864 7086.94412307359
641.249557519963 7195.3570486076
641.974154055615 7168.83818689951
642.259782838392 7139.63699229659
642.729968040574 7111.91010752889
643.456814562513 7085.54171570834
643.93242430873 7057.19587344781
644.202614562451 7195.90316418757
644.735839099729 7169.03700345437
645.386902963103 7141.17829848955
645.984014988629 7114.10005623031
646.744331899559 7087.63799085531
647.262078681584 7061.49252424472
647.923013890803 7033.88759096032
647.523060739038 7201.94863982279
648.182514774704 7175.63239305264
648.906099280365 7149.42340153776
649.783156635786 7123.10757783109
650.456545305659 7097.46318150296
651.426212351303 7072.58166781788
652.219878747842 7048.20378515459
653.505421361993 7024.97596565525
651.309883037197 7215.59165598134
652.387284585914 7192.66847900393
653.317233522513 7168.6860611022
654.139648738339 7143.75784708818
655.292398005034 7122.50422442555
656.275751899922 7099.18478438073
657.732772217314 7077.683243156
659.079781165742 7057.50793872443
660.644985389455 7037.4592816853
655.794409433356 7244.91721897464
656.785086094613 7223.30512142815
657.952903213203 7202.60019318215
659.226078776643 7182.11555595202
660.496289891972 7163.6572863429
662.066258623153 7146.76251043277
663.91931368209 7131.15583819237
665.832813852617 7115.20896516187
668.239153532153 7102.95671098752
671.085130514408 7094.79262055744
};
\addlegendentry{Two-Forking, $s=60$}
\addplot [semithick, color3, mark=square, mark size=1.5, mark options={solid}, only marks]
table {%
632.083046015167 7238.3813889576
633.501727893165 7223.11277663234
633.818227745472 7183.55669583334
635.032871825905 7210.32509039057
635.381708354339 7171.20818199509
635.615810584855 7131.2313112741
636.8282612315 7198.77212270851
637.24406017179 7159.54113905204
637.517714839491 7120.990642569
637.890776341838 7082.46553307423
638.848483406259 7191.14297874738
639.431714823593 7153.61725262403
639.729504022582 7113.40078666086
640.071139990864 7074.8247566252
640.651575445645 7037.03116968923
641.188020939984 7186.55557532964
641.705364340608 7147.33053668591
642.433639016106 7109.85875427656
642.80725367479 7072.79311069538
643.360726710975 7033.82068214279
644.034871142081 6998.82799657991
644.118050042659 7185.58469708512
644.767896808136 7148.48206193916
645.387584411541 7111.87618113565
645.942389597423 7074.78455520229
646.81474981473 7037.26726729919
647.580750529162 7001.73824862824
648.346597859464 6966.12319891726
647.619685831751 7191.66297839187
648.251323386547 7155.20916105185
649.071568101308 7119.79270087073
649.932716382287 7084.10441168485
650.724670479235 7047.7352968796
651.879065320959 7013.2816129291
652.826008038561 6978.36579203849
654.147649930845 6945.46590801127
651.408644908105 7207.54505856264
652.414193766054 7172.53762588941
653.344963423878 7138.78872284952
654.586810726044 7105.85846809818
655.711125859957 7073.24012137122
657.010377677017 7040.00510070707
658.719609185619 7009.81734655119
660.231885328083 6978.74808955342
662.575237418178 6950.70209682196
655.857569058385 7235.45287709045
656.993744647845 7204.11358671784
658.344475752543 7173.77825584252
659.816548109961 7144.43660590283
661.453122782115 7115.01346573785
663.251646791063 7089.04921994621
665.451300439573 7064.86146274539
668.030999121284 7042.86723127866
670.988072368322 7022.33652405165
674.817031836091 7010.54452130585
};
\addlegendentry{Two-Forking, $s=70$}
\addplot [semithick, color=black, mark=*, mark size=2]
table {%
654.883289294144 7267.41295229045
650.53098823605 7241.34793184904
646.848890757193 7228.78383198056
643.511379070829 7222.13724671847
640.904978773837 7225.02144258951
638.487207859213 7228.3129233973
636.519432038631 7238.44414185552
634.68847334505 7248.89355658449
633.275064895989 7262.00821191369
631.847422466753 7278.44423526674
630.555630681685 7295.41227755153
};
\addlegendentry{Single-Forking};
\end{axis}

\end{tikzpicture}
\caption{
This figure illustrates achievable points on Intel DevCloud cluster for empirical and task averaged server utilization cost and empirical average of service completion time for two-forking with different values of forked servers $m_0$, $m_1$ and $N-m_0-m_1$ at forking points $t_0=0$, $t=30$ and $s \in \{40, 50, 60, 70\}$, respectively. 
For comparison, we also plot the tradeoff points for single-forking at forking time $t=30$ varying the number of initial servers $n_0$. 
}
\label{Fig:TradeoffIntelDevCloudTFPjob600IFT10}
\end{figure}


\section{Conclusions and Future Work}
\label{sec:concl}
This paper considers a multi-fork analysis for running cloud computing jobs. 
We analytically computed the mean service completion time and the mean server utilization cost for multiple computation jobs, 
when the job execution time at each server is assumed to be \emph{i.i.d.} with a shifted exponential distribution. 
We show that having multiple forking points for speculative execution of jobs provide significantly improved tradeoff between the  two performance metrics. 
As a special case, we also show that starting with a single server in speculative execution of tasks is sub-optimal. 
This paper considers replication as a strategy for speculative execution. 

We empirically verified that the insights derived from the shifted exponential distribution continue to hold when the job execution times at individual servers have heavy-tailed distributions such as Pareto and Weibull. 
We also conducted this study on a real compute cluster, and verified that the empirical distribution of the job execution time has a constant shift and light tails. 
This implies that a shifted exponential distribution capture the service time well in real compute clusters. 
As a result, the insights derived from the analytical study continue to hold on the studied compute cluster as well. 

Recently, coding-theory-inspired approaches have been
applied to mitigate the effect of straggling~\cite{Tandon2017ICML, Aktas2018Sigmetrics, Ye2018ICML}. Single fork analysis with coding has been studied in~\cite{Aktas2018Sigmetrics}, 
where a single copy of the job is started at $t=0$. 
Considering multi-fork analysis with such general coding flexibilities remains an important future direction. 
We hope that the framework provided in this article can be utilized to quantify the performance gains of multi-forked coded replicas.

Further, this work considers the performance metrics for a single job system.  
Analysis of overall completion time of jobs when there is a sequence of job arrivals is an open problem. 
When the system load is low, the request queue has a single job with high probability, 
and our work provides insights for the queueing system in this regime. 





\begin{appendices}
\section{Proof of Theorem \ref{thm:Monotone}}\label{apdx:proof:stoc}
We know that $\E S = \int_{u \in \R_+}\bar{F}_S(u)du$, and we will show the monotonicity on the complementary distribution function of service times in the two cases. 
Since $\bar{F}_S = 1- (1-\bar{F}_{S_1})^K$ is a monotonically increasing function of $\bar{F}_{S_1}$, 
it suffices to show the monotonicity of the complementary distribution function of service times for single task in the two cases. 
\begin{enumerate}[(i)]
	\item 
	From the hypothesis on two forking time sequences $t, t'$, 
	and denoting $t'_{-1} = t_0 = 0$, we can write 
	\EQ{
		I_i(t) 
		\subseteq \cup_{j=0}^iI_j(t').
	}
	That is, for any time $u \in I_i(t) = [t_i, t_{i+1})$, we can write $t \in I_j(t')$ for some $j \le i$. 
	Let $j(u)$ and $i(u)$ denote the forking stage at time $u$ for the two forking sequences $t', t$ respectively, then $j(u) \le i(u)$ for all times $u$.
	Furthermore, since $t_\ell \le t'_\ell$ for each stage $\ell$, and the complementary distribution function $\bar{F}$ is non-increasing, we have $\bar{F}(u-t_\ell) \le \bar{F}(u-t'_\ell)$ for $u \ge t'_\ell$. 
	Therefore, we can write the following inequality for the complementary distribution for the service time of a single task for two forking instant sequences as 
	\EQ{
		\bar{F}_{S_1^{(t)}}(u) = \prod_{\ell=0}^{i(u)}\bar{F}(u-t_\ell)^{n_\ell} \le  \prod_{\ell=0}^{j(u)}\bar{F}(u-t'_\ell)^{n_\ell} = \bar{F}_{S_1^{(t')}}(u). 
	}
	\item For any time $u \in I_i$, from the hypothesis $n'_j \le n_j$, it follows that  
	\EQ{
		\bar{F}_{S_1^{(n)}}(u) = \prod_{\ell=0}^{i}\bar{F}(u-t_\ell)^{n_\ell} \le  \prod_{\ell=0}^{i}\bar{F}(u-t_\ell)^{n'_\ell} = \bar{F}_{S_1^{(n')}}(u).
	}
\end{enumerate}

\section{Proof of Theorem \ref{thm:OptSingleFork}}\label{apdx:OptSingleFork}
For the proof of this theorem, we utilize Proposition~\ref{prop:PartialDer} for the partial derivatives of two performance metrics, with respect to initial fraction of servers and normalized forking time for single forking. 
\begin{enumerate}[(i)]
	\item 
	Since $\pd{\E S}{t_1} = \frac{1}{c}\pd{\E S}{u} > 0$ for forking time $t_1 > 0$, the result follows.
	\item 
	We can write the derivative of mean service completion time with respect to initial number of servers as $\pd{\E S}{n_0} = \frac{1}{N}\pd{\E S}{x}$. 
	We write the scaled partial derivative of mean service completion time with respect to initial fraction of servers $x$ as a function of number of parallel tasks as
	\EQ{
		h(K) \triangleq \frac{1}{c}\pd{\E S}{x}(K) + u(1-(1-e^{-\alpha xu})^K).
	}
	We will show that $h(K) \le 0$ for all $K \ge 1$, which would give us the result. 
	From the partial derivative of mean service completion time with respect to initial fraction of servers $x$, we get 
	\EQ{
		h(K) = -\sum_{k=1}^K\frac{(1-e^{-\alpha xu})^k}{k\alpha x^2} + \frac{u}{x}(1-(1-e^{-\alpha xu})^K).
	}
	Since $e^x \ge 1 +x$ for all $x \in \R$, we see that $h(1) = -\frac{(1-(1+\alpha xu)e^{-\alpha xu})}{\alpha x^2} < 0$ for forking instant $u > 0$. 
	Furthermore, we can write 
	\EQ{
		\frac{h(K) - h(K-1)}{(1-e^{-\alpha xu})^{K-1}} = -\frac{(1-e^{-\alpha xu})}{K\alpha x^2} + \frac{u}{x}e^{-\alpha xu}
	}
	We see that the right hand side of the above equation is an increasing function of $K$. 
	Let $K^\ast = \inf\set{K \in \N: h(K) - h(K-1) \ge 0}$. 
	Then $h(K) - h(K-1) < 0$ for all $K < K^\ast$, and $h(K)-h(K-1) \ge 0$ for all $k \ge K^\ast$. 
	It follows then $h(K) \le h(1)\vee h(\infty)$. 
	We observe that the limiting value 
	\eq{
		&\lim_{K \to \infty}h(K) = -\sum_{k \in \N}\frac{(1-e^{-\alpha xu})^k}{k\alpha x^2} + \frac{u}{x}\\
		&=\frac{u}{x}+\frac{1}{\alpha x^2}\ln(e^{-\alpha xu}) = \frac{u}{x} - \frac{u}{x} = 0
	}
	\item 
	Since $\pd{\E W_1}{t_1} = \frac{1}{c}\pd{\E W_1}{u} < 0$ for forking time $t_1 > 0$, the result follows.
	\item 
	We show in Appendix~\ref{sec:ConvexUtilization}, 
	that the mean server utilization $\E W_1$ is a strict convex function of initial server fraction $x$. 
	Hence, it follows that there exists a unique optimal initial fraction $x^\ast \in [\frac{1}{N}, 1]$ that minimizes the mean server utilization cost. 
	Convexity of mean server utilization cost with respect to initial fraction of servers $x$ implies that $\pd{\E W_1}{x}$ is increasing in $x$ for any fixed forking instant $u > 0$. 
	Next, we show in Appendix~\ref{sec:PositiveUtilizationUnit} that the $\left.\pd{\E W_1}{x}\right\rvert_{x=1} > 0$ for all forking instants $u > 0$. 
	In Appendix~\ref{sec:SignUtilizationSingle},  we show that partial derivative of mean server utilization cost at initial fraction of servers $x = \frac{1}{N}$, 
	defined as a function of normalized forking time $u$ 
	\EQ{
		g(u) \triangleq \left(\frac{\mu}{\alpha\lambda}\right)\left.\pd{\E W_1}{x}\right\rvert_{x = \frac{1}{N}}
	}
	can take both positive and negative values depending on the normalized forking time $u$.  
	In particular, we show that 
	\EQ{
		\sgn\left(g(u) \right) = \SetIn{u > v_3} - \SetIn{u < v_3}.
	}
	
	Therefore, for $u \ge v_3$, we have $g(u) > 0$, and hence the mean server utilization cost $\E W_1$ is an increasing function of initial server fraction $x \in [\frac{1}{N}, 1]$. 
	It follows that the mean server utilization cost $\E W_1$ is minimized at the unique optimal initial fraction $x^\ast = \frac{1}{N}$.  
	In the other case when  $u < v_3$, we have $g(u) < 0$, 
	and the mean server utilization cost is minimized at the unique optimal initial server fraction $x^\ast \in [\frac{1}{N}, 1]$ is given by the solution of the implicit equation~\eqref{eqn:ImplicitFraction} such that $\left.\pd{\E W_1}{x}\right\rvert_{x=x^\ast} = 0$. 
\end{enumerate}

\section{Convexity of mean server utilization with initial fraction of servers}
\label{sec:ConvexUtilization}
We show the strict convexity of mean server utilization with respect to initial fraction of servers, 
by showing that the second partial derivative is always positive. 
For $u > 1$, we can compute this second partial derivative $\frac{\partial^2\E W_1}{\partial x^2}$ to be equal to
\eq{
&= \frac{2\lambda e^{-\alpha xu}}{x^3\mu}\Big[
\alpha x(u-1)(e^{\alpha  x}-1-\alpha x)
+ \frac{\alpha^2x^3}{2}(2u-1)\\
&+\frac{\alpha^2x^2(1-x)}{2}(u-1)^2(e^{\alpha  x}-1) 
+e^{\alpha  x}- 1 -\alpha x -\frac{\alpha^2x^2}{2}
\Big].
}
The second partial derivative of mean server utilization $\frac{\partial^2\E W_1}{\partial x^2}$ for single task with respect to fraction of initial servers for normalized forking time $u < 1$ is equal to
\EQ{
\frac{2\lambda e^{-\alpha xu}}{x^3\mu}\left[e^{\alpha xu}-(1+\alpha xu + \frac{\alpha^2x^2u^2}{2})+ \frac{\alpha^2x^3u^2}{2}\right].
}
From the Taylor series expansion of exponential function $e^{\alpha  x}$, 
we obtain that $\frac{\partial^2\E W_1}{\partial x^2} > 0$ for all initial fraction $x \in [1/N, 1]$ and normalized forking time $u > 0$. 

\section{Positivity of partial derivative of mean server utilization with initial fraction of servers at unit fraction}
\label{sec:PositiveUtilizationUnit}

We will show that the partial derivative $\left.\pd{\E W_1}{x}\right\rvert_{x=1} > 0$ 
for all forking instants $u > 0$. 
To this end, we observe that at $x =1$, we can write 
\EQ{
\left.\pd{\E W_1}{x}\right\rvert_{x=1} = \begin{cases}
\frac{\alpha\lambda}{\mu}\left[1 - \frac{e^{-\alpha u}}{\alpha}(e^{\alpha}-1)\right],& u \ge 1\\
\frac{\lambda}{\mu}\left[e^{-\alpha u}-1+ \alpha u \right], & u \in (0,1].
\end{cases}
}
We observe that $\pd{\E W_1}{x} \ge 0$ at $x =1$ for $u \in (0, 1]$, 
since $e^{-\alpha u} \ge 1 - \alpha u$. 
For $u \ge 1$, we observe that $e^{-\alpha u} \le e^{-\alpha}$ and hence  
\EQ{
1 - \frac{e^{-\alpha u}}{\alpha}(e^{\alpha}-1) \ge 1-\frac{1}{\alpha} + \frac{e^{-\alpha}}{\alpha} \ge 0,
}
where the last inequality follows from the fact that $e^{-\alpha} \ge 1 - \alpha$. 
Therefore, we deduce that $\pd{\E W_1}{x} \ge 0$ at $x =1$ for $u  \ge 1$.

\section{Partial derivative of mean server utilization with initial fraction of servers at smallest fraction}
\label{sec:SignUtilizationSingle} 

We will show that the sign of $g(u) = \left.\pd{\E W_1}{x}\right\rvert_{x = \frac{1}{N}}$ 
depends on the normalized forking time $u$, the constant $c\mu$ and the number of servers $N$.  
Recall that $\frac{\alpha}{N} = \mu c$, 
and we observe that for $u \ge 1$ the function $g(u)$ equals
\EQ{
1 - e^{-c\mu u}\left((N-1)(u-1)+\frac{N}{c\mu}\right)(e^{c\mu}-1) + (N-1)e^{-c\mu u}. 
}
Further, we have for $u \le 1$
\EQ{
g(u) = u + (N-1)ue^{-c\mu u} - \frac{N}{c\mu}(1 - e^{-c\mu u}).
}

We observe that the limiting values are $\lim_{u \downarrow 0}g(u) = 0, \lim_{u \to \infty}g(u) = 1$. 
Further, we observe that $g$ is continuous at $u = 1$, 
such that 
\EQ{
g(1) = 1 + (N-1)e^{-c\mu} - \frac{N}{c\mu}(1-e^{-c\mu}).
}
We observe that $\lim_{c\mu \to 0}g(1) = 0$ and $\lim_{c\mu \to \infty}g(1) = 1$. 
In general, $g(1) \ge 0$ if and only if 
\EQ{
(N-c\mu)(e^{c\mu}-1) \le c\mu N.
} 

\subsection{Behavior of $g'(u)$}
\label{subsec:DerBehavior}

We can write the derivative $g'(u)$ with respect to normalized forking time $u$ as  
\EQ{
\begin{cases}
e^{-c\mu u}\left[c\mu(N-1)((u-1)e^{c\mu}-u) +(e^{c\mu}-1)\right], & u \ge 1\\
1 - e^{-c\mu u}((N-1)c\mu u +1), &u \le 1.
\end{cases}
}
We can verify that the limiting values are $\lim_{u \downarrow 0}g'(u) = 0$ and $\lim_{u \to \infty}g'(u) = 0$, 
and $g'$ is continuous at $u=1$. 
We can also write the scaled second derivative $h(u) \triangleq g''(u)\frac{e^{c\mu u}}{c\mu}$ as equal to
\EQ{
\begin{cases}
(N-2)(e^{c\mu}-1) -c\mu(N-1)((u-1)e^{c\mu}-u), & u > 1\\
(N-1)c\mu u-(N-2), &u < 1.
\end{cases}
}
We observe that $h(0) = -(N-2)$ and it is linearly increasing in the interval $[0, 1]$. 
We next observe that $h$ (and hence $g''$) is discontinuous at $u=1$ for $N > 2$ with a positive jump of $(N-2)e^{c\mu}$, 
which results in $\lim_{u \downarrow 1}h(u) > 0$.  
Further, $\lim_{u \uparrow 1} h(u) > 0$ if and only if $c\mu  > 1 - \frac{1}{(N-1)}$.  
Further, we observe that $h$ is linearly decreasing in the interval $[1, \infty]$ with $\lim_{u \to \infty}h(u) = -\infty$.  
We define the following two normalized time instants 
\meq{2}{
&v_0 =  \frac{1}{c\mu} - \frac{1}{c\mu(N-1)},&
&v_1 = 1 + v_0 + \frac{1}{e^{c\mu}-1}, 
}
where $v_0$ and $v_1$ are the expressions obtained by setting the two expressions of $h(u)$ to zero corresponding to $u<1$ and $u>1$ cases respectively. 
We observe that $v_1 > 1$ for all parameter values $c, \mu, N$, 
whereas $v_0$ can be larger or smaller than unity depending on the parameters. 
In particular, $v_0 \le 1$ if and only if 
$
c\mu \ge 1 - \frac{1}{(N-1)}. 
$
Since $g''(u) = e^{-c\mu u}h(u)/c\mu$, we have $\sgn(g''(u)) = \sgn(h(u))$, where 
\EQ{
\sgn(h(u)) = \SetIn{u \in (v_0\wedge 1, 1]\cup[1, v_1)} - \SetIn{u \in [0, v_0 \wedge 1)\cup(v_1, \infty)}, 
} 
and we conclude that the second derivative of $g$ is discontinuous at unity, and can have one or two zeros. 
Specifically, we \emph{always} have $g''(v_1) = 0$ with $v_1  > 1$ and we \emph{can} have $g''(v_0) = 0$ if $c\mu \ge 1 - \frac{1}{(N-1)}$. 
We further remark that, $g''(u)$ is negative in the intervals $(0, v_0\wedge 1)$ and $(v_1, \infty)$ and is non-negative in the interval $[v_0\wedge 1, v_1]$. 

Therefore, the first derivative $g'(u)$ decreases monotonically from $g'(0) = 0$ to $g'(v_0\wedge 1)$, 
and then monotonically increases from $g'(v_0\wedge 1)$ to $g'(v_1)$, followed by the decrease to $g'(\infty) = 0$. 
From the continuity of $g'$, we infer that $g'(v_1) > 0 > g'(v_0\wedge 1)$, 
and hence we can conclude that there exists a unique $v_2 \in [v_0\wedge 1, v_1]$ such that $g'(v_2) = 0$. 
Further, the first derivative $g'$ is negative in the interval $(0, v_2)$ and positive in the interval $(v_2, \infty)$. 
That is, 
\EQ{
\sgn(g'(u)) = -\SetIn{0 < u  < v_2} + \SetIn{ u  > v_2}.
}
Note that the unique value of $v_2$ is determined by the solution to $g'(u) = 0$ in two complementary cases when $v_0 \le u \le 1$ and $u \ge 1$. 
Hence, we can write this normalized forking time threshold $v_2$ to be equal to 
$1 - \frac{1}{c\mu(N-1)} + \frac{1}{e^{c\mu}-1}$ when $v_2 \ge 1$ and the solution $v_2$ to the implicit equation 
$k(u) \triangleq e^{c\mu u} - 1 - (N-1)c\mu u = 0$ if $v_2 \in [v_0 , 1]$. 
We see that the function $k(u)$ has a zero in $[v_0, 1)$ if and only if $v_0 < 1$ and $k(1) > 0$. 
We can see that when $v_0 \ge 1$ then $k(1) \le 0$ from the Taylor expansion of exponential function
\footnote{If $c\mu \le 1 - \frac{1}{N-1}$, then we have $\frac{e^{c\mu} - 1}{c\mu} = \sum_{k =0}^{\infty}\frac{(c\mu)^{k}}{(k+1!} \le \sum_{k =0}^{\infty}(c\mu)^k = \frac{1}{1-c\mu} \le (N-1)$, i.e. $k(1) \le 0$. }, 
and hence 
\EQ{
\set{v_0 < 1}\cap \set{k(1) > 0} = \set{k(1) > 0}.
} 
Therefore, we can see that $v_2 \ge 1$ iff $k(1) \le 0$ iff  $\frac{1}{c\mu}\ln(1+c\mu(N-1)) \ge 1$, 
and in this case there is no non-trivial solution to the above implicit equation in $(0, 1]$. 
Hence, we can write the normalized forking time threshold $v_2$ as 
\EQ{
\begin{cases}
1 - \frac{1}{c\mu(N-1)} + \frac{1}{e^{c\mu}-1}, &\frac{1}{c\mu}\ln(1+c\mu(N-1)) \ge 1,\\
\frac{1}{c\mu}\ln(1 + (N-1)c\mu v_2), &\frac{1}{c\mu}\ln(1+c\mu(N-1)) \le 1.
\end{cases}
}

\subsection{Behavior of $g(u)$}
\label{subsec:FunBehavior}
From the behavior of $g'(u)$, 
we conclude that the function $g$ is monotonically decreasing from $g(0) = 0$ to $g(v_2)$, and monotonically increasing from $g(v_2)$ to $\lim_{u \to \infty}g(u) = 1$. 
Therefore, it follows that there exists a unique $v_3 \in [v_2, \infty)$ such that $g(v_3) = 0$. 
However, there is no direct relation between $v_3$ and $v_1$. 
We conclude that the function $g = \left.\pd{\E W_1}{x}\right\rvert_{x=\frac{1}{N}}$ is negative when the normalized forking time $u < v_3$ and non-negative otherwise. 

We can write the normalized forking time threshold $v_3$ as the solution to an implicit equation obtained from equating $g(u) = 0$ for two cases when $u \le 1$ and $u \ge 1$. 
Since $g$ is increasing in $[v_3, \infty)$ with $g(v_3) = 0$, 
it follows that $v_3 \le 1$ if and only if $g(1) \ge 0$. 
From the continuity of $g$ at $u=1$, we can evaluate $g(1)$ and observe that $g(1) \ge 0$ if and only if $(N-c\mu)(e^{c\mu}-1) \le c\mu N$. 
Therefore, we can write the following implicit equation for the normalized forking time threshold $v_3$, where $e^{c\mu v_3} + N-1$ is equal to
\EQ{
\begin{cases}
\left((N-1)(v_3-1)+\frac{N}{c\mu}\right)(e^{c\mu}-1), &(1-\frac{c\mu}{N})\frac{(e^{c\mu}-1)}{c\mu} >1,\\
\frac{N}{1- \frac{c\mu v_3}{N}},  & (1-\frac{c\mu}{N})\frac{(e^{c\mu}-1)}{c\mu} \le 1.
\end{cases}
}
We define $f(x) \triangleq (N-x)(e^{x}-1)-Nx$, and observe that $v_3 > 1$ iff $f(c\mu) > 0$. 
This implies a necessary conditions for $v_3 > 1$ is to have $c\mu \in (0, N)$.  
We can verify that the two derivatives of $f$ are 
\meq{2}{
&f'(x) = (N-1)(e^x-1)-xe^x,&&f''(x) = (N-2-x)e^x.  
}
It follows that $\sgn(f'') = \SetIn{x < N-2} -\SetIn{x > N-2}$ and hence the first derivative $f$ is increasing for $x \in [0, N-2)$ and decreasing for $x > N-2$.  
That is for $N \ge 2$, the first derivative $f'$ increases monotonically from $f'(0) = 0$ to $f'(N-2) = e^{N-2}-(N-1) \ge 0$ and then decreases monotonically to $f'(N) = -e^N-(N-1) < f'(N-1) = -(N-1)$. 
Let $x^\# \in (N-2, N-1)$ be the unique point where $f'(x^\#) = 0$, then $\sgn(f') = \SetIn{x \in (0, x^\#)} -\SetIn{x \in (x^\#, N)}$. 
Hence, the function $f$ is monotonically increasing for $x \in (0, x^\#)$ from $f(0) = 0$ to $f(x^\#)$, and then monotonically decreasing in $(x^\#, N)$ to $f(N) = -N^2$.  
Hence, there exists a threshold $x^\prime \in (N-2, N)$ on the rate $c\mu$ such that $v_3 \ge 1$ for $c\mu \le x^\prime$, 
where $x^\prime$ is the unique non-zero solution to the implicit equation $f(x) = 0$. 

Taking the number of servers $N=12$, the shift of the shifted exponential service as $c = 1$, 
we have plotted the function $g, g'$, and $g''$ as a function of normalized forking time $u$ for the service rate $\mu = 0.8$ in Fig.~\ref{Fig:NatureOfganddg_v0ge1} and for the service rate $\mu = 2$ in 
Fig.~\ref{Fig:NatureOfganddg_v0le1}. 

\begin{figure}[hhh]
\centering
\input{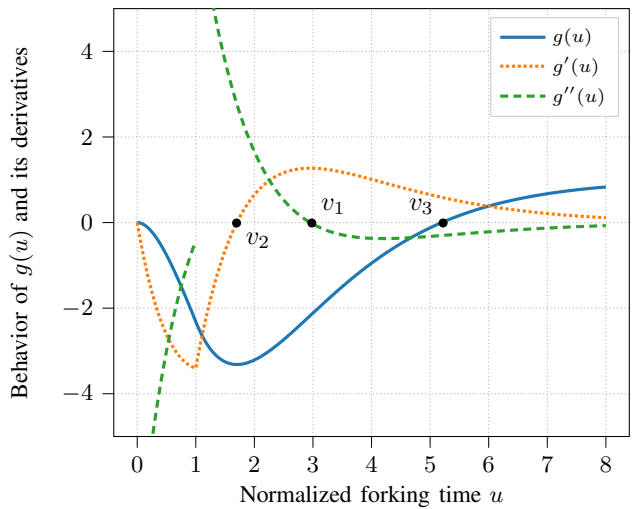}
\caption{ 
For a fixed number of servers $N=12$, the service shift $c=1$, and the exponential service rate $\mu=0.8$, 
this graph plots $g(u), g'(u), g''(u)$ as a function of normalized forking time $u$.  
For this set of system parameters,  
we have $v_0 = \frac{25}{22} > 1$ and $v_1 = 1 + \frac{25}{22}+ \frac{1}{e^{0.8}-1}$. 
The discontinuity in $g''$ at $u =1$ is given by $(N-2)c\mu = 8$. 
Further, we see that $e^{0.8}-1 \le 8.8$ and hence we have $v_2 = v_1 - \frac{5}{4} = 1 - \frac{5}{44} + \frac{1}{e^{0.8}-1}$. 
}
\label{Fig:NatureOfganddg_v0ge1}
\end{figure}

\begin{figure}[hhh]
\centering
\input{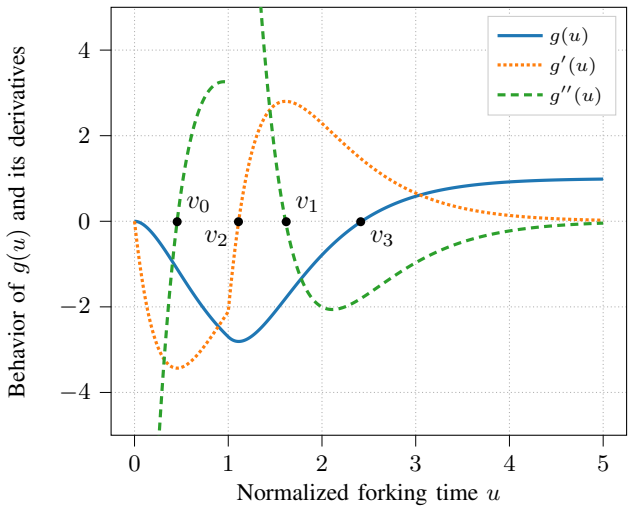}
\caption{
For a fixed number of servers $N=12$, the service shift $c=1$, and the exponential service rate $\mu=2$, this graph plots $g(u), g'(u), g''(u)$ as a function of normalized forking time $u$. 
For this set of system parameters, we have $v_0 = \frac{5}{11} < 1$ and $v_1 = 1 + \frac{5}{11}+ \frac{1}{e^{2}-1}$. 
The discontinuity in $g''$ at $u =1$ is given by $(N-2)c\mu = 20$. 
Further, we see that $e^{2}-1 \le 22$ and hence we have $v_2 = v_1 - \frac{1}{2} = 1 - \frac{1}{22} + \frac{1}{e^{2}-1}$. 
}
\label{Fig:NatureOfganddg_v0le1}
\end{figure}

\subsection{Behavior of normalized forking point threshold $v_3$}
\label{subsec:ImpactParam}

We list down the impact of system parameters $K, N, \mu, c$ on this normalized forking point threshold $v_3$, from its analytical expression.  
That is, for $c\mu > x^\prime$, the optimal forking threshold $v_3$ is the solution to the equation 
\EQ{
\frac{e^{c\mu u} -1}{c\mu u} = \frac{1}{1- \frac{c\mu u}{N}}.
}
For large $c\mu$, the solution for this is $v_3=0$, which implies that when the amount of work done in a single shift by a single server is large, starting with a single server is optimal. 
Thus,  $n^\ast_0 \to 1$ as $c\mu \to \infty$. 
For $c\mu \le x^\prime$, the optimal forking threshold $v_3$ is the solution to the equation 
\EQ{
e^{c\mu u} + N-1 = (e^{c\mu} -1)\left((u-1)(N-1) + \frac{N}{c\mu}\right).
}
For small $c\mu$ the optimal threshold $v_3$ is very large, and can be written approximately as the solution to the equation 
\EQ{
e^{c\mu v}  = 1 + (N-1)c\mu v.
} 
The solution of this equation is $v = \frac{y}{c\mu}$ where $e^y = 1+ (N-1)y$. 

We have plotted the normalized forking point threshold $v_3$ along with its approximation in Fig.~\ref{Fig:NormalizedFPThldv3VsProductcMu}, 
as a function of the product $c\mu$ for a fixed number of servers $N=12$. 
Since $x' \ge N-2$, we have $c\mu < x'$ in this plot, and we observe that the approximation matches the numerically computed threshold $v_3$. 
We conclude that when the amount of work done in a single shift by a single server is small, $cv_3 \propto \frac{1}{\mu}$ and for the forking point $t_1 \le cv_3 \approx \frac{y}{\mu}$ the optimal number of initial servers $n_0^\ast \ge 1$.  

\begin{figure}[hhh]
\centering
\input{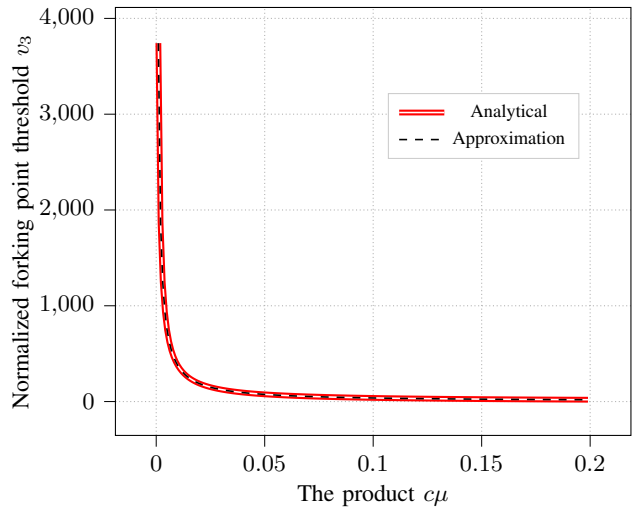}
\caption{ 
For fixed number of servers $N=12$, this graph plots the normalized forking point threshold $v_3$ and its approximation as a function of $c\mu$. 
The service rate $\mu$ is set to $0.01$ and service shift $c$ is varied. 
}
\label{Fig:NormalizedFPThldv3VsProductcMu}
\end{figure}

Since the threshold $x' \in (N-2, N)$, this threshold increases as the total number of servers $N$ increases. 
Therefore for any finite $c\mu$ and large $N$, we have $c\mu \le x'$ and hence the normalized forking point threshold $v_3 > 1$. 
In this setting, the normalized forking point threshold $v_3$ is approximately given by the solution to the equation
\EQ{
e^{c\mu u} = N u(e^{c\mu}-1). 
}
Setting $x = -c\mu u < 0$, we can rewrite the above equation as $xe^x = -\frac{c\mu}{N(e^{c\mu}-1)}$. 
Further, recall that the equation $y = xe^x$ is equivalent to $x = W(y)$ where $W$ is the Lambert-W relation~\cite{corless1996lambertw}. 
Therefore, we have $-c\mu u = W(-\frac{c\mu}{N(e^{c\mu}-1)})$. 
We further note that the Lambert-W is a double-valued relation on $(-1/e, 0)$. 
The relation $W$ has two real branches in this regime, 
and they are represented by single valued functions $W_0$ and $W_{-1}$ with additional constraints $W_0 \ge -1 \ge W_{-1}$. 
We observe that $-\frac{c\mu}{N(e^{c\mu}-1)} \le 0$ and decreases to zero as $N$ grows large, and hence the interval $(-1/e, 0)$ is of interest for large number of servers.  
Since $x' \ge N-2$, and hence for any fixed $c\mu$, we have $c\mu < x'$ as $N$ increases, and the normalized forking threshold $v_3 > 1$. 
This implies that the lower branch $W_{-1}$ of the Lambert-W relation gives us the right solution. 
Thus, we have 
\EQ{
v_3 \approx -\frac{1}{c\mu}W_{-1}\left(-\frac{c\mu}{N (e^{c\mu}-1)}\right). 
}

\begin{figure}[hhh]
\centering
\input{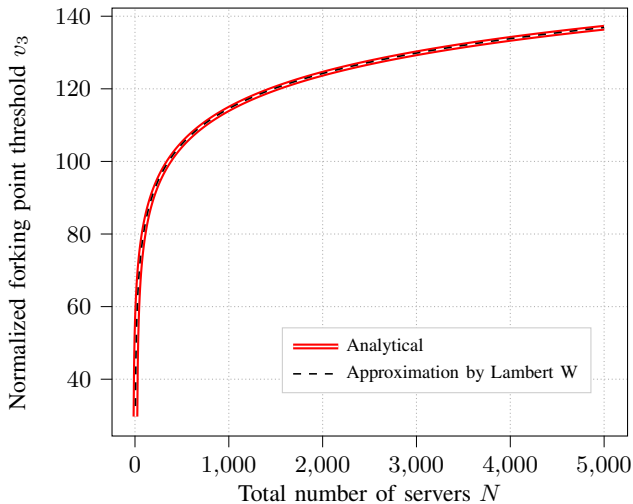}
\caption{ 
This graph plots the normalized forking point threshold $v_3$ and its approximation by Lambert-W relation as a function of number of servers $N$,  
for shifted exponential service distribution with shift $c=8$ and rate $\mu=0.01$. 
}
\label{Fig:NormalizedFPThldv3VsNumServers}
\end{figure}

We have plotted the numerically computed normalized forking point threshold $v_3$ along with its approximation from Lambert-W relation in Fig.~\ref{Fig:NormalizedFPThldv3VsNumServers}, as a function of the number of servers $N$.  
The plot suggests that threshold is logarithmically increasing with the number of servers $N$. 
Indeed, the authors of \cite{barry1993class, barry2000analytical} provided an approximation for the lower branch of the Lambert-W relation as  
\EQ{
W_{-1}(z) \approx -1-\sigma - \frac{2}{M_1}\left[1-\frac{1}{1+ \frac{M_1\sqrt{\sigma/2}}{1+M_2\sigma \exp(M_3\sqrt{\sigma})}}\right],
}
where $\sigma = -1-\ln(-z)$, $M_1=0.3361$, $M_2=-0.0042$, and $M_3=-0.0201$. 
This  approximation has a maximum relative error of only 0.025\%  \cite{barry2000analytical}. 
Using this approximation, we can analytically see that $v_3$ increases logarithmically with total number of servers $N$ for large $N$.  

\section{Numerical results of single forking for heavy-tailed distributions}
\label{apdx:HeavyTailSFP}
It is not straightforward to compute the mean service completion time and mean server utilization cost analytically for general distribution of job execution times. 
In fact these computations remain challenging for heavy-tailed distributions such as the Pareto and the Weibull distributions. 
However, we can compute these performance metrics empirically to verify that the insights obtained by the analytical study of shifted exponential distribution continue to hold in other cases.  

For this numerical study, we select the identical system parameters to those chosen for study with shifted exponential distribution. 
That is, we take $K = 10$ parallel tasks, and the total available servers per task as $N=12$, together with the cost of server utilization per unit time as $\lambda  = 1$. 
\subsection{Pareto distribution}
\label{subsec:ParetoSFP}
We take the execution times at each server to be an \emph{i.i.d.} random sequence having a Pareto distribution with scale $x_m$ and shape $\alpha > 1$, 
such that the complementary execution time distribution is given by 
\EQ{
\bar{F}(x) = P\set{X > x} = \left(\frac{x_m}{x}\right)^\alpha\SetIn{x \ge x_m}.
}
For the following numerical studies, we take mean of the Pareto distribution as $m = \frac{x_m\alpha}{\alpha-1}= 0.8$, and scale $x_m = 0.1*m$. 
We verify that the insights obtained from the shifted exponential distribution of execution time, 
continue to hold in this case. 
To this end, we plot the empirical mean of service completion time in Fig.~\ref{Fig:MeanServiceParetoSFP}, 
and empirical mean of server utilization cost in Fig.~\ref{Fig:MeanCostParetoSFP}, 
both as functions of initial servers $n_0 \in \{1, \dots, 11\}$ for values of forking times in $t_1 \in \{m, 2m, \dots, 9m\}$. 
As expected, the mean service completion time increases and the mean server utilization cost decreases with increase in the forking time $t_1$. 
We also notice the decrease in the mean service completion time as the number of initial server $n_0$ increases. 
Interestingly, we still have an optimal number of initial servers $n_0^\ast \ge 1$ that minimizes the mean server utilization cost for different values of forking points $t_1 > m$. 

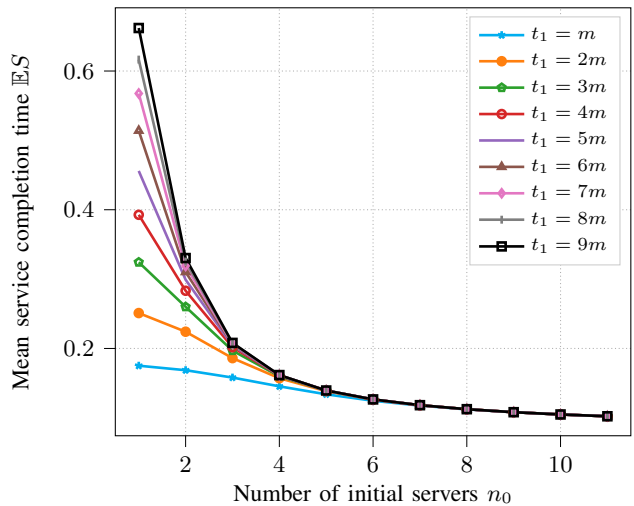
\begin{figure}[hhh]
\centering
\begin{tikzpicture}

\definecolor{color0}{rgb}{0.12156862745098,0.466666666666667,0.705882352941177}
\definecolor{color1}{rgb}{1,0.498039215686275,0.0549019607843137}
\definecolor{color2}{rgb}{0.172549019607843,0.627450980392157,0.172549019607843}
\definecolor{color3}{rgb}{0.83921568627451,0.152941176470588,0.156862745098039}
\definecolor{color4}{rgb}{0.580392156862745,0.403921568627451,0.741176470588235}
\definecolor{color5}{rgb}{0.549019607843137,0.337254901960784,0.294117647058824}
\definecolor{color6}{rgb}{0.890196078431372,0.466666666666667,0.76078431372549}
\definecolor{color7}{rgb}{0.737254901960784,0.741176470588235,0.133333333333333}

\begin{axis}[
font=\small,
legend cell align={left},
legend style={draw=white!80.0!black, font=\scriptsize},
tick align=outside,
tick pos=left,
x grid style={white!69.01960784313725!black, densely dotted},
xlabel={Number of initial servers $n_0$},
xmajorgrids,
xmin=0.5, xmax=11.5,
xtick style={color=black},
y grid style={white!69.01960784313725!black, densely dotted},
ylabel={Mean service completion time $\E S$},
ymajorgrids,
ymin=0.0741404190846677, ymax=0.691161198542081,
ytick style={color=black}
]

\addplot [semithick, cyan, mark=star, mark size=1.5, mark options={solid}, line width=1pt]
table {%
1 0.175126183930618
2 0.168647327295846
3 0.158025512797236
4 0.145396964266735
5 0.133885613636278
6 0.124613464774345
7 0.117569029920989
8 0.112163470813312
9 0.108030558034828
10 0.104805736691023
11 0.102201418608451
};
\addlegendentry{$t_1=m$}
\addplot [semithick, color1, mark=*, mark size=1.5, mark options={solid}, line width=1pt]
table {%
1 0.25102329143346
2 0.224131207993809
3 0.186055832342216
4 0.157069822918165
5 0.138383097409416
6 0.126260453504139
7 0.118107404408264
8 0.112401909646071
9 0.108114203573765
10 0.104818795593906
11 0.102210963582293
};
\addlegendentry{$t_1=2m$}
\addplot [semithick, color2, mark=pentagon, mark size=1.5, mark options={solid}, line width=1pt]
table {%
1 0.324281682585965
2 0.26001651187997
3 0.196904155218078
4 0.159853891144972
5 0.139065026045008
6 0.126474684089882
7 0.118195923481255
8 0.112423467346884
9 0.108093287354966
10 0.104834079573721
11 0.102203306180634
};
\addlegendentry{$t_1=3m$}
\addplot [semithick, color3, mark=o, mark size=1.5, mark options={solid}, line width=1pt]
table {%
1 0.392702923081315
2 0.283196136438107
3 0.201742761143336
4 0.160829467197722
5 0.139253175508146
6 0.126444100305267
7 0.118194369202438
8 0.112399905996743
9 0.108079866453144
10 0.104813197721765
11 0.102204920420168
};
\addlegendentry{$t_1=4m$}
\addplot [semithick, color4, mark=, mark size=1.5, mark options={solid}, line width=1pt]
table {%
1 0.455968329089246
2 0.299095871246918
3 0.204266583544672
4 0.161296107486552
5 0.139352560671704
6 0.126505438465105
7 0.118211053324207
8 0.112348137249458
9 0.108137895253384
10 0.104806253387726
11 0.102220281488972
};
\addlegendentry{$t_1=5m$}
\addplot [semithick, color5, mark=triangle, mark size=1.5, mark options={solid}, line width=1pt]
table {%
1 0.514015820135655
2 0.310203040635031
3 0.205938252249873
4 0.161346950093049
5 0.139295815256313
6 0.126506369980329
7 0.118171951429904
8 0.112405289709002
9 0.10810315884246
10 0.104834883785916
11 0.102208251350065
};
\addlegendentry{$t_1=6m$}
\addplot [semithick, color6, mark=diamond, mark size=1.5, mark options={solid}, line width=1pt]
table {%
1 0.567597766504757
2 0.31897411730922
3 0.206617669807937
4 0.16142451535251
5 0.139339517280832
6 0.126530313226628
7 0.118196387247315
8 0.112373748248864
9 0.108105542113051
10 0.104840391372425
11 0.10222698862481
};
\addlegendentry{$t_1=7m$}
\addplot [semithick, white!49.80392156862745!black, mark=|, mark size=1.5, mark options={solid}, line width=1pt]
table {%
1 0.616475177495095
2 0.325274379619962
3 0.207412370492729
4 0.161578966648194
5 0.139287359878481
6 0.126532877066356
7 0.118217213426747
8 0.112384742013046
9 0.108093461428759
10 0.104829588740269
11 0.102195697497066
};
\addlegendentry{$t_1=8m$}
\addplot [semithick, black, mark=square, mark size=1.5, mark options={solid}, line width=1pt]
table {%
1 0.662008332945187
2 0.330340124663116
3 0.208007710226241
4 0.161620375154907
5 0.139384510462546
6 0.126524564289059
7 0.118211954903276
8 0.112441226364196
9 0.108095562254935
10 0.104820223806016
11 0.102213851218083
};
\addlegendentry{$t_1=9m$}
\end{axis}

\end{tikzpicture}
\caption{
This graph displays the mean service completion time $\E S$ as a function of initial number of servers $n_0 \in \{1, \dots, 11\}$ for single forking of $K = 10$ parallel tasks at different forking times $t_1 \in \{m, 2m, \dots, 9m\}$ when the single task execution time at servers are \emph{i.i.d.}  with Pareto distribution of shift $x_m = 0.08$ and shape $\alpha = 10/9$. 
}
\label{Fig:MeanServiceParetoSFP}
\end{figure}

\begin{figure}[hhh]
\centering
\begin{tikzpicture}

\definecolor{color0}{rgb}{0.12156862745098,0.466666666666667,0.705882352941177}
\definecolor{color1}{rgb}{1,0.498039215686275,0.0549019607843137}
\definecolor{color2}{rgb}{0.172549019607843,0.627450980392157,0.172549019607843}
\definecolor{color3}{rgb}{0.83921568627451,0.152941176470588,0.156862745098039}
\definecolor{color4}{rgb}{0.580392156862745,0.403921568627451,0.741176470588235}
\definecolor{color5}{rgb}{0.549019607843137,0.337254901960784,0.294117647058824}
\definecolor{color6}{rgb}{0.890196078431372,0.466666666666667,0.76078431372549}
\definecolor{color7}{rgb}{0.737254901960784,0.741176470588235,0.133333333333333}

\begin{axis}[
font=\small,
legend cell align={left},
legend style={at={(0.97,0.03)}, anchor=south east, draw=white!80.0!black, font=\scriptsize},
tick align=outside,
tick pos=left,
x grid style={white!69.01960784313725!black, densely dotted},
xlabel={Number of initial servers $n_0$},
xmajorgrids,
xmin=0.5, xmax=11.5,
xtick style={color=black},
y grid style={white!69.01960784313725!black, densely dotted},
ylabel={Mean server utilization cost $\E W_1$},
ymajorgrids,
ymin=0.2, ymax=0.999400172231525,
ytick style={color=black}
]

\addplot [semithick, cyan, mark=star, mark size=1.5, mark options={solid}, line width=1pt]
table {%
1 0.757959351451249
2 0.62703382049562
3 0.578503136237515
4 0.577380415308398
5 0.603966897009197
6 0.647054980590793
7 0.700907225042399
8 0.761379262012
9 0.826597007617828
10 0.894933570945484
11 0.965517856698378
};
\addlegendentry{$t_1=m$}
\addplot [semithick, color1, mark=*, mark size=1.5, mark options={solid}, line width=1pt]
table {%
1 0.510140298282687
2 0.374148680156209
3 0.374819674996413
4 0.424312142515513
5 0.491867886503321
6 0.56606904027794
7 0.643035999702653
8 0.721325184306927
9 0.800049968786476
10 0.879135491024668
11 0.958397361907368
};
\addlegendentry{$t_1=2m$}
\addplot [semithick, color2, mark=pentagon, mark size=1.5, mark options={solid}, line width=1pt]
table {%
1 0.421133602240607
2 0.32167537000758
3 0.351367039654789
4 0.414947126261051
5 0.488288242578396
6 0.564835303651505
7 0.642666779500124
8 0.721179517435312
9 0.800006135775263
10 0.87913321321693
11 0.958404507519811
};
\addlegendentry{$t_1=3m$}
\addplot [semithick, color3, mark=o, mark size=1.5, mark options={solid}, line width=1pt]
table {%
1 0.378095823463358
2 0.303724191937976
3 0.345806482135918
4 0.413434850435231
5 0.487892567600782
6 0.564665691310561
7 0.642639023740362
8 0.721125004378344
9 0.799953369029382
10 0.879113141399087
11 0.958382945101768
};
\addlegendentry{$t_1=4m$}
\addplot [semithick, color4, mark=, mark size=1.5, mark options={solid}, line width=1pt]
table {%
1 0.354344758370341
2 0.296217041353556
3 0.344006433445615
4 0.413112139792684
5 0.487877535864437
6 0.564701986254474
7 0.642677610323173
8 0.721067017867986
9 0.800021307657742
10 0.879097774482979
11 0.958436508761482
};
\addlegendentry{$t_1=5m$}
\addplot [semithick, color5, mark=triangle, mark size=1.5, mark options={solid}, line width=1pt]
table {%
1 0.339613353674177
2 0.292396384279081
3 0.343348532491171
4 0.412918909692909
5 0.487783118353574
6 0.56469450503514
7 0.642600075434941
8 0.721155877221008
9 0.799996650567208
10 0.879140816585608
11 0.958400925838421
};
\addlegendentry{$t_1=6m$}
\addplot [semithick, color6, mark=diamond, mark size=1.5, mark options={solid}, line width=1pt]
table {%
1 0.330972130411199
2 0.290731623092404
3 0.342956136290385
4 0.412869229327429
5 0.487765083510794
6 0.564748517165325
7 0.642632226754113
8 0.721096732143132
9 0.800010685348237
10 0.879170218321584
11 0.958429341401286
};
\addlegendentry{$t_1=7m$}
\addplot [semithick, white!49.80392156862745!black, mark=|, mark size=1.5, mark options={solid}, line width=1pt]
table {%
1 0.324960923741317
2 0.289545041384216
3 0.342887795333176
4 0.412931183443796
5 0.487771571011192
6 0.564740360509213
7 0.642658274201448
8 0.721097187011791
9 0.799979167316185
10 0.879157999393842
11 0.958373457231914
};
\addlegendentry{$t_1=8m$}
\addplot [semithick, black, mark=square, mark size=1.5, mark options={solid}, line width=1pt]
table {%
1 0.321171290320272
2 0.288917853805569
3 0.34294991891923
4 0.412904424224672
5 0.487817436979832
6 0.564736338762564
7 0.642633778423145
8 0.721167724681014
9 0.799989545843917
10 0.879123623842841
11 0.958426910876037
};
\addlegendentry{$t_1=9m$}
\end{axis}

\end{tikzpicture}
\caption{
This graph displays the mean server utilization cost $\E W$ as a function of initial number of servers $n_0 \in \{1, \dots, 11\}$ for single forking of $K = 10$ parallel tasks at different forking times $t_1 \in \{m, 2m, \dots, 9m\}$ when the single task execution time at servers are \emph{i.i.d.}  with Pareto distribution of shift $x_m=0.08$ and shape $\alpha=10/9$. 
}
\label{Fig:MeanCostParetoSFP}
\end{figure}
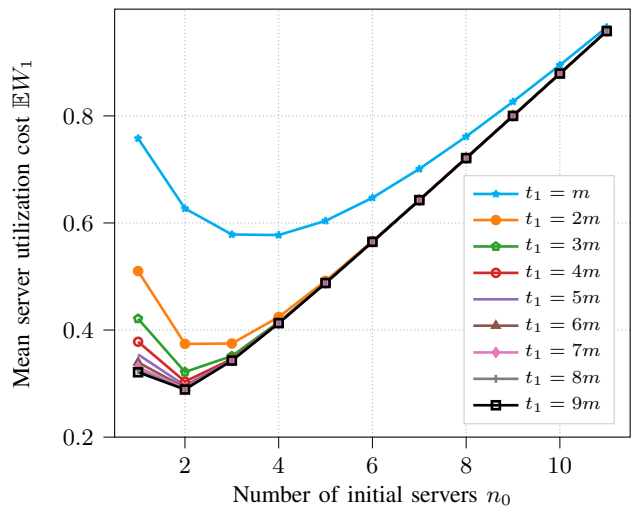

These insights are identical to the one derived from the shifted exponential service case, 
and we have an interesting tradeoff between the two metrics when we vary the number of initial servers and the forking time.  
First observation is that forking time gives a true tradeoff between these two metrics. 
Second and more interesting observation is that there exist a minimum number of initial servers for each forking time, until which point we can decrease both the mean service completion time and the mean server utilization cost. 
We demonstrate this tradeoff by plotting the empirical mean of service completion time with respect to empirical mean of server utilization cost for $\lambda =1$ in Fig.~\ref{Fig:TradeoffParetoSFP} for initial number of servers $n_0 \in \{1, \dots, 11\}$ and for forking times $t_1 \in \{m, 2m, \dots, 9m\}$. 
It is clear that there is an optimal number of initial servers for each forking time $t_1$. 
Further, the mean server utilization decreases with forking time $t_1$, 
though at the cost of increase in mean service completion time. 

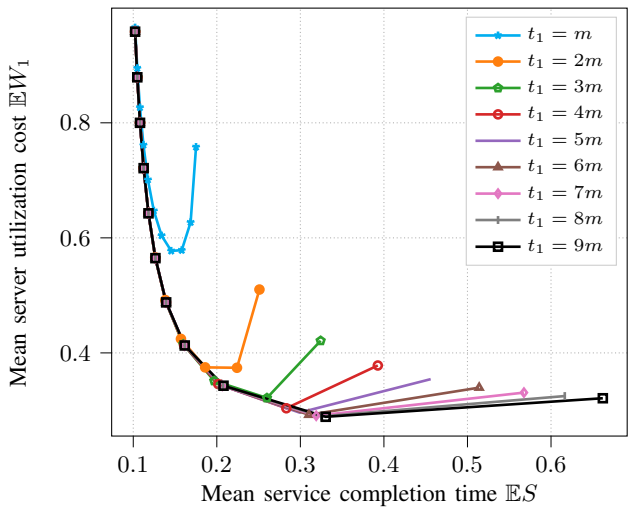
\begin{figure}[hhh]
	\centering
\begin{tikzpicture}

\definecolor{color0}{rgb}{0.12156862745098,0.466666666666667,0.705882352941177}
\definecolor{color1}{rgb}{1,0.498039215686275,0.0549019607843137}
\definecolor{color2}{rgb}{0.172549019607843,0.627450980392157,0.172549019607843}
\definecolor{color3}{rgb}{0.83921568627451,0.152941176470588,0.156862745098039}
\definecolor{color4}{rgb}{0.580392156862745,0.403921568627451,0.741176470588235}
\definecolor{color5}{rgb}{0.549019607843137,0.337254901960784,0.294117647058824}
\definecolor{color6}{rgb}{0.890196078431372,0.466666666666667,0.76078431372549}
\definecolor{color7}{rgb}{0.737254901960784,0.741176470588235,0.133333333333333}

\begin{axis}[
font=\small,
legend cell align={left},
legend style={draw=white!80.0!black, font=\scriptsize},
tick align=outside,
tick pos=left,
x grid style={white!69.01960784313725!black, densely dotted},
xlabel={Mean service completion time $\E S$},
xmajorgrids,
xmin=0.0741404190846677, xmax=0.691161198542081,
xtick style={color=black},
y grid style={white!69.01960784313725!black, densely dotted},
ylabel={Mean server utilization cost $\E W_1$},
ymajorgrids,
ymin=0.255422463470822, ymax=0.999400172231525,
ytick style={color=black}
]

\addplot [semithick, cyan, mark=star, mark size=1.5, mark options={solid}, line width=1pt]
table {%
0.175126183930618 0.757959351451249
0.168647327295846 0.62703382049562
0.158025512797236 0.578503136237515
0.145396964266735 0.577380415308398
0.133885613636278 0.603966897009197
0.124613464774345 0.647054980590793
0.117569029920989 0.700907225042399
0.112163470813312 0.761379262012
0.108030558034828 0.826597007617828
0.104805736691023 0.894933570945484
0.102201418608451 0.965517856698378
};
\addlegendentry{$t_1 = m$}
\addplot [semithick, color1, mark=*, mark size=1.5, mark options={solid}, line width=1pt]
table {%
0.25102329143346 0.510140298282687
0.224131207993809 0.374148680156209
0.186055832342216 0.374819674996413
0.157069822918165 0.424312142515513
0.138383097409416 0.491867886503321
0.126260453504139 0.56606904027794
0.118107404408264 0.643035999702653
0.112401909646071 0.721325184306927
0.108114203573765 0.800049968786476
0.104818795593906 0.879135491024668
0.102210963582293 0.958397361907368
};
\addlegendentry{$t_1 = 2m$}
\addplot [semithick, color2, mark=pentagon, mark size=1.5, mark options={solid}, line width=1pt]
table {%
0.324281682585965 0.421133602240607
0.26001651187997 0.32167537000758
0.196904155218078 0.351367039654789
0.159853891144972 0.414947126261051
0.139065026045008 0.488288242578396
0.126474684089882 0.564835303651505
0.118195923481255 0.642666779500124
0.112423467346884 0.721179517435312
0.108093287354966 0.800006135775263
0.104834079573721 0.87913321321693
0.102203306180634 0.958404507519811
};
\addlegendentry{$t_1 = 3m$}
\addplot [semithick, color3, mark=o, mark size=1.5, mark options={solid}, line width=1pt]
table {%
0.392702923081315 0.378095823463358
0.283196136438107 0.303724191937976
0.201742761143336 0.345806482135918
0.160829467197722 0.413434850435231
0.139253175508146 0.487892567600782
0.126444100305267 0.564665691310561
0.118194369202438 0.642639023740362
0.112399905996743 0.721125004378344
0.108079866453144 0.799953369029382
0.104813197721765 0.879113141399087
0.102204920420168 0.958382945101768
};
\addlegendentry{$t_1 = 4m$}
\addplot [semithick, color4, mark=, mark size=1.5, mark options={solid}, line width=1pt]
table {%
0.455968329089246 0.354344758370341
0.299095871246918 0.296217041353556
0.204266583544672 0.344006433445615
0.161296107486552 0.413112139792684
0.139352560671704 0.487877535864437
0.126505438465105 0.564701986254474
0.118211053324207 0.642677610323173
0.112348137249458 0.721067017867986
0.108137895253384 0.800021307657742
0.104806253387726 0.879097774482979
0.102220281488972 0.958436508761482
};
\addlegendentry{$t_1 = 5m$}
\addplot [semithick, color5, mark=triangle, mark size=1.5, mark options={solid}, line width=1pt]
table {%
0.514015820135655 0.339613353674177
0.310203040635031 0.292396384279081
0.205938252249873 0.343348532491171
0.161346950093049 0.412918909692909
0.139295815256313 0.487783118353574
0.126506369980329 0.56469450503514
0.118171951429904 0.642600075434941
0.112405289709002 0.721155877221008
0.10810315884246 0.799996650567208
0.104834883785916 0.879140816585608
0.102208251350065 0.958400925838421
};
\addlegendentry{$t_1 = 6m$}
\addplot [semithick, color6, mark=diamond, mark size=1.5, mark options={solid}, line width=1pt]
table {%
0.567597766504757 0.330972130411199
0.31897411730922 0.290731623092404
0.206617669807937 0.342956136290385
0.16142451535251 0.412869229327429
0.139339517280832 0.487765083510794
0.126530313226628 0.564748517165325
0.118196387247315 0.642632226754113
0.112373748248864 0.721096732143132
0.108105542113051 0.800010685348237
0.104840391372425 0.879170218321584
0.10222698862481 0.958429341401286
};
\addlegendentry{$t_1 = 7m$}
\addplot [semithick, white!49.80392156862745!black, mark=|, mark size=1.5, mark options={solid}, line width=1pt]
table {%
0.616475177495095 0.324960923741317
0.325274379619962 0.289545041384216
0.207412370492729 0.342887795333176
0.161578966648194 0.412931183443796
0.139287359878481 0.487771571011192
0.126532877066356 0.564740360509213
0.118217213426747 0.642658274201448
0.112384742013046 0.721097187011791
0.108093461428759 0.799979167316185
0.104829588740269 0.879157999393842
0.102195697497066 0.958373457231914
};
\addlegendentry{$t_1 = 8m$}
\addplot [semithick, black, mark=square, mark size=1.5, mark options={solid}, line width=1pt]
table {%
0.662008332945187 0.321171290320272
0.330340124663116 0.288917853805569
0.208007710226241 0.34294991891923
0.161620375154907 0.412904424224672
0.139384510462546 0.487817436979832
0.126524564289059 0.564736338762564
0.118211954903276 0.642633778423145
0.112441226364196 0.721167724681014
0.108095562254935 0.799989545843917
0.104820223806016 0.879123623842841
0.102213851218083 0.958426910876037
};
\addlegendentry{$t_1 = 9m$}
\end{axis}

\end{tikzpicture}
	\caption{ 
		This graph displays the mean server utilization cost $\E W$ as a function of mean service completion time $\E S$ when we vary the number of initial servers $n_0 \in \{1, \dots, 11\}$ for single forking of $K = 10$ parallel tasks at different forking times $t_1 \in \{m, 2m, \dots, 9m\}$ when the single task execution time at servers are \emph{i.i.d.}  with Pareto distribution of shift $x_m=0.08$ and shape $\alpha = 10/9$. 
		}
	\label{Fig:TradeoffParetoSFP}
\end{figure}

\subsection{Weibull distribution}
\label{subsec:WeibullSFP}
We take the execution time at each server to be and \emph{i.i.d.} random sequence having a Weibull distribution with scale $\gamma$ and shape $\beta$, 
such that the complementary execution time distribution is given by 
\EQ{
	\bar{F}(x) = P\set{X > x} = e^{-\left(\frac{x}{\gamma}\right)^\beta} \SetIn{x \ge 0}.
}
For the following numerical studies, 
we take scale of the Weibull distribution as $\gamma=16$, and shape $\beta=2$.
For this service distribution, 
we individually plot the empirical mean of service completion time in Fig.~\ref{Fig:MeanServiceWeibullSFP}, 
and the empirical mean of server utilization cost in Fig.~\ref{Fig:MeanCostWeibullSFP}, 
both as functions of initial servers $n_0 \in \{1, \dots, 11\}$ with values of forking times in $t_1 \in \{1, 3, 5, \dots, 19\}$. 
We demonstrate the tradeoff between these two performance metrics by plotting the empirical mean of service completion time with respect to empirical mean of server utilization cost for $\lambda =1$ in Fig.~\ref{Fig:TradeoffWeibullSFP}  for initial number of serves $n_0 \in \{1, \dots, 11\}$ for values of forking times in $t_1 \in \{1, 3, 5, \dots, 19\}$.

We reiterate that, as expected, the insights obtained in the Weibull service distribution case are identical to those obtained from the light-tailed distribution such as shifted exponential and heavy-tailed distribution such as Pareto.
This suggests that the insights derived from our analysis applies to random execution times with heavy tail distributions as well. 
\begin{figure}[hhh]
\scalebox{1}{
\begin{tikzpicture}

\definecolor{color0}{rgb}{0.12156862745098,0.466666666666667,0.705882352941177}
\definecolor{color1}{rgb}{1,0.498039215686275,0.0549019607843137}
\definecolor{color2}{rgb}{0.172549019607843,0.627450980392157,0.172549019607843}
\definecolor{color3}{rgb}{0.83921568627451,0.152941176470588,0.156862745098039}
\definecolor{color4}{rgb}{0.580392156862745,0.403921568627451,0.741176470588235}
\definecolor{color5}{rgb}{0.549019607843137,0.337254901960784,0.294117647058824}
\definecolor{color6}{rgb}{0.890196078431372,0.466666666666667,0.76078431372549}
\definecolor{color7}{rgb}{0.737254901960784,0.741176470588235,0.133333333333333}
\definecolor{color8}{rgb}{0.0901960784313725,0.745098039215686,0.811764705882353}

\begin{axis}[
font=\small,
legend cell align={left},
legend style={draw=white!80.0!black, font=\scriptsize},
tick align=outside,
tick pos=left,
x grid style={white!69.01960784313725!black, densely dotted},
xlabel={Number of initial servers $n_0$},
xmajorgrids,
xmin=0.5, xmax=11.5,
xtick style={color=black},
y grid style={white!69.01960784313725!black, densely dotted},
ylabel={Mean service completion time $\E S$},
ymajorgrids,
ymin=7.08946180740696, ymax=23.6285099541644,
ytick style={color=black}
]
\addplot [semithick, color7, mark=star, mark size=1.5, mark options={solid}, line width=1pt]
table {%
1 8.65468854450982
2 8.5749005216662
3 8.4537171165028
4 8.37089787875789
5 8.31947682296514
6 8.18266295570761
7 8.15864534567871
8 8.06899763283092
9 7.95815129529582
10 7.87480671600981
11 7.84123672316866
};
\addlegendentry{$t_1=1$}
\addplot [semithick, color1, mark=*, mark size=1.5, mark options={solid}, line width=1pt]
table {%
1 10.4400903493354
2 10.1502459967111
3 9.87897843776859
4 9.61189416562653
5 9.33546425277891
6 9.08255117347147
7 8.81569835105206
8 8.6197409613711
9 8.34737693484047
10 8.16702777900127
11 7.94954092330463
};
\addlegendentry{$t_1=3$}
\addplot [semithick, color2, mark=pentagon, mark size=1.5, mark options={solid}, line width=1pt]
table {%
1 12.2045132491143
2 11.6369129949686
3 11.1311866527346
4 10.6866104532583
5 10.2443834418679
6 9.79636197191396
7 9.36449469752003
8 9.02917584590282
9 8.66541200480386
10 8.31663752941778
11 8.01250067759097
};
\addlegendentry{$t_1=5$}
\addplot [semithick, color3, mark=o, mark size=1.5, mark options={solid}, line width=1pt]
table {%
1 13.914505856339
2 13.0853760881289
3 12.3134051918588
4 11.6226330075544
5 10.9359217878237
6 10.3728915697869
7 9.81981046713094
8 9.29870613849767
9 8.85883035771072
10 8.4396181886256
11 8.05047251072268
};
\addlegendentry{$t_1=7$}
\addplot [semithick, color4, mark=, mark size=1.5, mark options={solid}, line width=1pt]
table {%
1 15.5425832666564
2 14.419657947351
3 13.3531302710525
4 12.3906723838407
5 11.4605643999306
6 10.6902174402372
7 10.0393991710927
8 9.42541826555313
9 8.90272018264378
10 8.47247383085931
11 8.08298488871327
};
\addlegendentry{$t_1=9$}
\addplot [semithick, color5, mark=triangle, mark size=1.5, mark options={solid}, line width=1pt]
table {%
1 17.1507075683709
2 15.6349437847404
3 14.1692994104604
4 12.9013214901938
5 11.779225572753
6 10.8664812886239
7 10.1248785518668
8 9.4755422255672
9 8.97631086620506
10 8.48554256188592
11 8.08237408184987
};
\addlegendentry{$t_1=11$}
\addplot [semithick, color6, mark=diamond, mark size=1.5, mark options={solid}, line width=1pt]
table {%
1 18.6997520787773
2 16.6833522830039
3 14.8077864728119
4 13.2102551260235
5 11.908732547228
6 10.9542542100134
7 10.1340901582537
8 9.51319794029972
9 8.93999711397853
10 8.48100306791745
11 8.08010528369086
};
\addlegendentry{$t_1=13$}
\addplot [semithick, white!49.80392156862745!black, mark=|, mark size=1.5, mark options={solid}, line width=1pt]
table {%
1 20.2133062786789
2 17.567487210061
3 15.1954089209297
4 13.3041509933915
5 11.9533272364327
6 10.9195240812237
7 10.1821311381188
8 9.49400995816253
9 8.942286988131
10 8.43363392740025
11 8.08339497091358
};
\addlegendentry{$t_1=15$}
\addplot [semithick, cyan, mark=Mercedes star, mark size=1.5, mark options={solid}, line width=1pt]
table {%
1 21.5480983599482
2 18.1375938253471
3 15.3327690027403
4 13.3890634568415
5 11.9886026485047
6 10.961571427987
7 10.1712738170318
8 9.48296208916158
9 8.93440930466544
10 8.48101819926089
11 8.06929900125104
};
\addlegendentry{$t_1=17$}
\addplot [semithick, black, mark=square, mark size=1.5, mark options={solid}, line width=1pt]
table {%
1 22.8767350384026
2 18.5457241526913
3 15.43150245417
4 13.403033602994
5 11.9678398777175
6 10.9424731683001
7 10.1257313570364
8 9.48502064957881
9 8.93023171114152
10 8.49841485874241
11 8.08666923097526
};
\addlegendentry{$t_1=19$}
\end{axis}

\end{tikzpicture}}
\caption{
This graph displays the mean service completion time $\E S$ as a function of initial number of servers $n_0 \in \{1, \dots, 11\}$ for single forking of $K = 10$ parallel tasks at different forking times $t_1 \in \{1, 3, 5, \dots, 19\}$ when the single task execution time at servers are \emph{i.i.d.}  with Weibull distribution of scale $\gamma=16$ and shape $\beta=2$.
}
\label{Fig:MeanServiceWeibullSFP}
\end{figure}
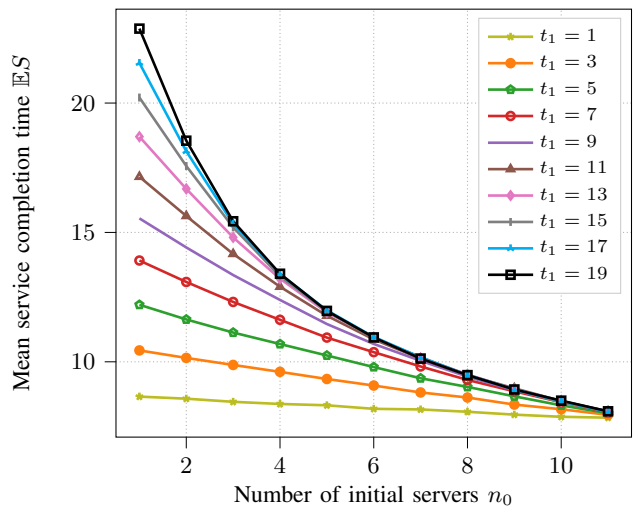%

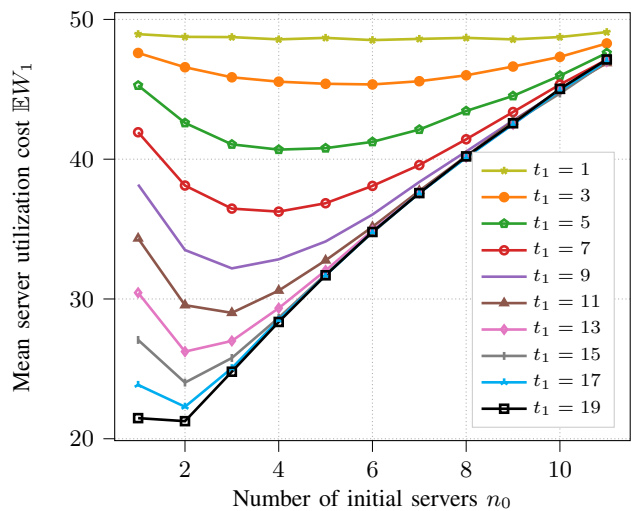
\begin{figure}[hhh]
\scalebox{1}{
\begin{tikzpicture}

\definecolor{color0}{rgb}{0.12156862745098,0.466666666666667,0.705882352941177}
\definecolor{color1}{rgb}{1,0.498039215686275,0.0549019607843137}
\definecolor{color2}{rgb}{0.172549019607843,0.627450980392157,0.172549019607843}
\definecolor{color3}{rgb}{0.83921568627451,0.152941176470588,0.156862745098039}
\definecolor{color4}{rgb}{0.580392156862745,0.403921568627451,0.741176470588235}
\definecolor{color5}{rgb}{0.549019607843137,0.337254901960784,0.294117647058824}
\definecolor{color6}{rgb}{0.890196078431372,0.466666666666667,0.76078431372549}
\definecolor{color7}{rgb}{0.737254901960784,0.741176470588235,0.133333333333333}
\definecolor{color8}{rgb}{0.0901960784313725,0.745098039215686,0.811764705882353}

\begin{axis}[
font=\small,
legend cell align={left},
legend style={at={(0.97,0.03)}, anchor=south east, draw=white!80.0!black, font=\scriptsize},
tick align=outside,
tick pos=left,
x grid style={white!69.01960784313725!black, densely dotted},
xlabel={Number of initial servers $n_0$},
xmajorgrids,
xmin=0.5, xmax=11.5,
xtick style={color=black},
y grid style={white!69.01960784313725!black, densely dotted},
ylabel={Mean server utilization cost $\E W_1$},
ymajorgrids,
ymin=19.8669073078214, ymax=50.482525008975,
ytick style={color=black}
]
\addplot [semithick, color7, mark=star, mark size=1.5, mark options={solid}, line width=1pt]
table {%
1 48.9447029149539
2 48.7524240615862
3 48.7341183314497
4 48.5716001466312
5 48.6797529947665
6 48.52059134381
7 48.606184399759
8 48.6797527674104
9 48.5710907729236
10 48.7393264208763
11 49.0909060225589
};
\addlegendentry{$t_1=1$}
\addplot [semithick, color1, mark=*, mark size=1.5, mark options={solid}, line width=1pt]
table {%
1 47.5986833469118
2 46.5784596648252
3 45.8584258654274
4 45.5455254884501
5 45.3948744543627
6 45.3463026004648
7 45.5781659833201
8 46.0015954838364
9 46.6286335525872
10 47.3193054740593
11 48.2813978804111
};
\addlegendentry{$t_1=3$}
\addplot [semithick, color2, mark=pentagon, mark size=1.5, mark options={solid}, line width=1pt]
table {%
1 45.2736592649229
2 42.5967774349177
3 41.0596610975453
4 40.6852590550431
5 40.7884399065496
6 41.2407015587611
7 42.1211062277038
8 43.4498156236447
9 44.5261284559823
10 45.9724919676508
11 47.5997492305569
};
\addlegendentry{$t_1=5$}
\addplot [semithick, color3, mark=o, mark size=1.5, mark options={solid}, line width=1pt]
table {%
1 41.9266395195584
2 38.1145112668051
3 36.4646911115157
4 36.2488051299392
5 36.8486096757564
6 38.0881982752722
7 39.5808869232394
8 41.424614217435
9 43.3735954162945
10 45.343039911517
11 47.1197895026862
};
\addlegendentry{$t_1=7$}
\addplot [semithick, color4, mark=, mark size=1.5, mark options={solid}, line width=1pt]
table {%
1 38.1812511423496
2 33.5023797658715
3 32.1884335620406
4 32.8352813801281
5 34.1125818422254
6 36.0378063708948
7 38.369361642583
8 40.549752511757
9 42.7624488500693
10 44.9116006705734
11 47.1135771166789
};
\addlegendentry{$t_1=9$}
\addplot [semithick, color5, mark=triangle, mark size=1.5, mark options={solid}, line width=1pt]
table {%
1 34.3326893707017
2 29.555393364851
3 29.0082865396428
4 30.6001600335628
5 32.761409714983
6 35.1595056423137
7 37.7368125548423
8 40.2175307245202
9 42.7165640320737
10 44.912618540819
11 46.8772047897603
};
\addlegendentry{$t_1=11$}
\addplot [semithick, color6, mark=diamond, mark size=1.5, mark options={solid}, line width=1pt]
table {%
1 30.4458655575992
2 26.2349782693274
3 26.9986950895397
4 29.3433710196973
5 32.0357130016247
6 34.9084108642925
7 37.5849726365072
8 40.1516184369753
9 42.5003118563676
10 44.89985990396
11 46.9869915267951
};
\addlegendentry{$t_1=13$}
\addplot [semithick, white!49.80392156862745!black, mark=|, mark size=1.5, mark options={solid}, line width=1pt]
table {%
1 27.0692437485257
2 24.0189669313894
3 25.7784537141359
4 28.6566631138083
5 31.7568990386075
6 34.7450055351564
7 37.5902033458439
8 40.1242500914596
9 42.6192896148718
10 44.699407214676
11 46.9828673592338
};
\addlegendentry{$t_1=15$}
\addplot [semithick, cyan, mark=Mercedes star, mark size=1.5, mark options={solid}, line width=1pt]
table {%
1 23.8540258498767
2 22.2879231467357
3 25.0499541605495
4 28.5045777221316
5 31.6856595438279
6 34.7914392786532
7 37.5731642054743
8 40.1148916403227
9 42.4997826031226
10 44.9150867153349
11 47.0141600226428
};
\addlegendentry{$t_1=17$}
\addplot [semithick, black, mark=square, mark size=1.5, mark options={solid}, line width=1pt]
table {%
1 21.4756253474563
2 21.2585262942374
3 24.7968139985863
4 28.356119834174
5 31.6947140613263
6 34.7928028792249
7 37.5754234456118
8 40.2007167735093
9 42.5712485771723
10 45.0302237441302
11 47.1339187727202
};
\addlegendentry{$t_1=19$}
\end{axis}

\end{tikzpicture}}
\caption{
This graph displays the mean server utilization cost $\E W$ as a function of initial number of servers $n_0 \in \{1, \dots, 11\}$ for single forking of $K = 10$ parallel tasks at different forking times $t_1 \in \{1, 3, 5, \dots, 19\}$ when the single task execution time at servers are \emph{i.i.d.}  with Weibull distribution of scale $\gamma=16$ and shape $\beta=2$. 
}
\label{Fig:MeanCostWeibullSFP}
\end{figure}%

\begin{figure}[hhh]
\scalebox{1}{
\begin{tikzpicture}

\definecolor{color0}{rgb}{0.12156862745098,0.466666666666667,0.705882352941177}
\definecolor{color1}{rgb}{1,0.498039215686275,0.0549019607843137}
\definecolor{color2}{rgb}{0.172549019607843,0.627450980392157,0.172549019607843}
\definecolor{color3}{rgb}{0.83921568627451,0.152941176470588,0.156862745098039}
\definecolor{color4}{rgb}{0.580392156862745,0.403921568627451,0.741176470588235}
\definecolor{color5}{rgb}{0.549019607843137,0.337254901960784,0.294117647058824}
\definecolor{color6}{rgb}{0.890196078431372,0.466666666666667,0.76078431372549}
\definecolor{color7}{rgb}{0.737254901960784,0.741176470588235,0.133333333333333}
\definecolor{color8}{rgb}{0.0901960784313725,0.745098039215686,0.811764705882353}

\begin{axis}[
font=\small,
legend cell align={left},
legend style={draw=white!80.0!black, font=\scriptsize},
tick align=outside,
tick pos=left,
x grid style={white!69.01960784313725!black, densely dotted},
xlabel={Mean service completion time $\E S$},
xmajorgrids,
xmin=7.08946180740696, xmax=23.6285099541644,
xtick style={color=black},
y grid style={white!69.01960784313725!black, densely dotted},
ylabel={Mean server utilization cost $\E W_1$},
ymajorgrids,
ymin=19.8669073078214, ymax=50.482525008975,
ytick style={color=black}
]
\addplot [semithick, color7, mark=star, mark size=1.5, mark options={solid}, line width=1pt]
table {%
8.65468854450982 48.9447029149539
8.5749005216662 48.7524240615862
8.4537171165028 48.7341183314497
8.37089787875789 48.5716001466312
8.31947682296514 48.6797529947665
8.18266295570761 48.52059134381
8.15864534567871 48.606184399759
8.06899763283092 48.6797527674104
7.95815129529582 48.5710907729236
7.87480671600981 48.7393264208763
7.84123672316866 49.0909060225589
};
\addlegendentry{$t_1=1$}
\addplot [semithick, color1, mark=*, mark size=1.5, mark options={solid}, line width=1pt]
table {%
10.4400903493354 47.5986833469118
10.1502459967111 46.5784596648252
9.87897843776859 45.8584258654274
9.61189416562653 45.5455254884501
9.33546425277891 45.3948744543627
9.08255117347147 45.3463026004648
8.81569835105206 45.5781659833201
8.6197409613711 46.0015954838364
8.34737693484047 46.6286335525872
8.16702777900127 47.3193054740593
7.94954092330463 48.2813978804111
};
\addlegendentry{$t_1=3$}
\addplot [semithick, color2, mark=pentagon, mark size=1.5, mark options={solid}, line width=1pt]
table {%
12.2045132491143 45.2736592649229
11.6369129949686 42.5967774349177
11.1311866527346 41.0596610975453
10.6866104532583 40.6852590550431
10.2443834418679 40.7884399065496
9.79636197191396 41.2407015587611
9.36449469752003 42.1211062277038
9.02917584590282 43.4498156236447
8.66541200480386 44.5261284559823
8.31663752941778 45.9724919676508
8.01250067759097 47.5997492305569
};
\addlegendentry{$t_1=5$}
\addplot [semithick, color3, mark=o, mark size=1.5, mark options={solid}, line width=1pt]
table {%
13.914505856339 41.9266395195584
13.0853760881289 38.1145112668051
12.3134051918588 36.4646911115157
11.6226330075544 36.2488051299392
10.9359217878237 36.8486096757564
10.3728915697869 38.0881982752722
9.81981046713094 39.5808869232394
9.29870613849767 41.424614217435
8.85883035771072 43.3735954162945
8.4396181886256 45.343039911517
8.05047251072268 47.1197895026862
};
\addlegendentry{$t_1=7$}
\addplot [semithick, color4, mark=, mark size=1.5, mark options={solid}, line width=1pt]
table {%
15.5425832666564 38.1812511423496
14.419657947351 33.5023797658715
13.3531302710525 32.1884335620406
12.3906723838407 32.8352813801281
11.4605643999306 34.1125818422254
10.6902174402372 36.0378063708948
10.0393991710927 38.369361642583
9.42541826555313 40.549752511757
8.90272018264378 42.7624488500693
8.47247383085931 44.9116006705734
8.08298488871327 47.1135771166789
};
\addlegendentry{$t_1=9$}
\addplot [semithick, color5, mark=triangle, mark size=1.5, mark options={solid}, line width=1pt]
table {%
17.1507075683709 34.3326893707017
15.6349437847404 29.555393364851
14.1692994104604 29.0082865396428
12.9013214901938 30.6001600335628
11.779225572753 32.761409714983
10.8664812886239 35.1595056423137
10.1248785518668 37.7368125548423
9.4755422255672 40.2175307245202
8.97631086620506 42.7165640320737
8.48554256188592 44.912618540819
8.08237408184987 46.8772047897603
};
\addlegendentry{$t_1=11$}
\addplot [semithick, color6, mark=diamond, mark size=1.5, mark options={solid}, line width=1pt]
table {%
18.6997520787773 30.4458655575992
16.6833522830039 26.2349782693274
14.8077864728119 26.9986950895397
13.2102551260235 29.3433710196973
11.908732547228 32.0357130016247
10.9542542100134 34.9084108642925
10.1340901582537 37.5849726365072
9.51319794029972 40.1516184369753
8.93999711397853 42.5003118563676
8.48100306791745 44.89985990396
8.08010528369086 46.9869915267951
};
\addlegendentry{$t_1=13$}
\addplot [semithick, white!49.80392156862745!black, mark=|, mark size=1.5, mark options={solid}, line width=1pt]
table {%
20.2133062786789 27.0692437485257
17.567487210061 24.0189669313894
15.1954089209297 25.7784537141359
13.3041509933915 28.6566631138083
11.9533272364327 31.7568990386075
10.9195240812237 34.7450055351564
10.1821311381188 37.5902033458439
9.49400995816253 40.1242500914596
8.942286988131 42.6192896148718
8.43363392740025 44.699407214676
8.08339497091358 46.9828673592338
};
\addlegendentry{$t_1=15$}
\addplot [semithick, cyan, mark=Mercedes star, mark size=1.5, mark options={solid}, line width=1pt]
table {%
21.5480983599482 23.8540258498767
18.1375938253471 22.2879231467357
15.3327690027403 25.0499541605495
13.3890634568415 28.5045777221316
11.9886026485047 31.6856595438279
10.961571427987 34.7914392786532
10.1712738170318 37.5731642054743
9.48296208916158 40.1148916403227
8.93440930466544 42.4997826031226
8.48101819926089 44.9150867153349
8.06929900125104 47.0141600226428
};
\addlegendentry{$t_1=17$}
\addplot [semithick, black, mark=square, mark size=1.5, mark options={solid}, line width=1pt]
table {%
22.8767350384026 21.4756253474563
18.5457241526913 21.2585262942374
15.43150245417 24.7968139985863
13.403033602994 28.356119834174
11.9678398777175 31.6947140613263
10.9424731683001 34.7928028792249
10.1257313570364 37.5754234456118
9.48502064957881 40.2007167735093
8.93023171114152 42.5712485771723
8.49841485874241 45.0302237441302
8.08666923097526 47.1339187727202
};
\addlegendentry{$t_1=19$}
\end{axis}

\end{tikzpicture}}
\caption{ 
This graph displays the mean server utilization cost $\E W$ as a function of mean service completion time $\E S$ when we vary the number of initial servers $n_0 \in \{1, \dots, 11\}$ for single forking of $K = 10$ parallel tasks at different forking times $t_1 \in \{1, 3, 5, \dots, 19\}$ when the single task execution time at servers are \emph{i.i.d.} with Weibull distribution of scale $\gamma=16$ and shape $\beta=2$. 
}
\label{Fig:TradeoffWeibullSFP}
\end{figure}
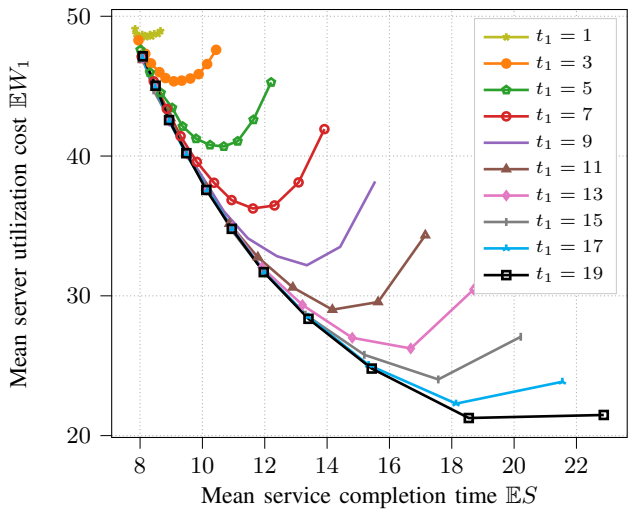

\section{Numerical results of multiple forking for heavy-tailed distributions}
\label{apdx:HeavyTailTFP}

In this section, we compare the optimal single forking to sub-optimal two-forking to understand the potential gains of multi-forking for \emph{i.i.d.} random execution times for two heavy-tailed distributions: 
Pareto and Weibull. 
For numerical studies in both the cases, we take $K=10$ tasks, sequentially forked on $N=12$ servers each, 
with server utilization cost rate $\lambda = 1$. 

\subsection{Pareto distribution}
\label{subsec:ParetoTFP} 
We assume the job execution times to have an \emph{i.i.d.} Pareto distribution with the mean $m = \frac{x_m\alpha}{\alpha-1}= 16$ and the shift $x_m = 0.52*m$, for the following numerical studies. 
We compare the performance of the single forking sequence $(0, n_0), (t, N-n_0)$ to that of the two-forking sequence $(0, m_0), (t, m_1), (s, N-m_0-m_1)$, when the second forking point is $s>t$. 
We plot the tradeoff curve between mean service completion time and mean server utilization cost for the single and two-forking sequences in Fig.~\ref{Fig:TradeoffParetoTFPAfter} for the values of forking points $t = 4$ and $s \in  \{5, 6, 7, 8\}$, varying the number of forked servers $n_0 \in [N]$ in single-forking case and $m_0, m_1$ in two-forking case.

\begin{figure}[hhh]
\centering
\begin{tikzpicture}

\definecolor{color0}{rgb}{0.12156862745098,0.466666666666667,0.705882352941177}
\definecolor{color1}{rgb}{1,0.498039215686275,0.0549019607843137}
\definecolor{color2}{rgb}{0.172549019607843,0.627450980392157,0.172549019607843}
\definecolor{color3}{rgb}{0.83921568627451,0.152941176470588,0.156862745098039}

\begin{axis}[
font=\small,
legend cell align={left},
legend style={at={(0.99,0.99)}, draw=white!80.0!black, font=\scriptsize},
tick align=outside,
tick pos=left,
x grid style={white!69.01960784313725!black, densely dotted},
xlabel={Mean service completion time $\E S$},
xmajorgrids,
xmin=9.14949739754036, xmax=18,
xtick style={color=black},
y grid style={white!69.01960784313725!black, densely dotted},
ylabel={Mean server utilization cost $\E W_1$},
ymajorgrids,
ymin=47.3850679359989, ymax=105,
ytick style={color=black}
]
\addplot [semithick, color0, mark=diamond, mark size=1.5, mark options={solid}, only marks]
table {%
9.6019703251523 95.8973907652901
9.7571574277057 92.5038262915882
9.73783678369872 91.4458394510986
9.93383389145061 89.1627008552355
9.94725146200483 88.1990349048365
9.95443935449228 87.2333994943725
10.1986822357343 86.1659820076006
10.1850826150541 85.1412433077577
10.198006656841 84.1684958653673
10.1964301332294 83.1719486652628
10.517698444132 83.4247410280941
10.5245153277195 82.4469578598377
10.5325155639277 81.4796435373557
10.5547316237283 80.5279941223942
10.5496606560005 79.4842333453077
10.9721750593249 81.2782075193727
10.973023609068 80.2684964882707
10.9704135118597 79.2794677613108
10.9687929804975 78.2588531215222
11.0003925750396 77.2940982923236
11.0222195609227 76.3800088915373
11.5263632865405 79.8710420886957
11.5401245770417 78.9303488791912
11.5471013083652 77.8867092918251
11.5721972456303 76.9288193880481
11.5879586027048 75.9561257570445
11.6098111439275 74.9255770801233
11.657726216381 73.989498028446
12.1788098211735 79.8362501113486
12.2248642641522 78.9724787920839
12.2336025505686 77.9345453320348
12.2745537510618 77.0595690955734
12.2946996109664 75.9709169916513
12.3627726463985 75.098477073563
12.4206527638895 74.2924959928139
12.5019265303981 73.3892049676873
12.7674102152858 82.2323612926071
12.8208269098081 81.3499786586507
12.8625146670199 80.364857642745
12.9180806232977 79.5236442652796
12.9762563259295 78.5599097868713
13.065576872074 77.8511150563524
13.1542855801615 77.0321583826636
13.2413380164542 76.2707164292916
13.3407063573258 75.5616048856731
13.1489010541314 88.8102585658845
13.2149615551644 87.8632483753039
13.2947884189717 87.175462660836
13.3648252864218 86.538453230606
13.4605899465804 85.8070524789207
13.5524554094877 85.1133708497696
13.6415394743366 84.4593169153874
13.7362072504546 83.9272335317981
13.8389171380131 83.7048729313006
13.9477773682102 83.5908279769649
};
\addlegendentry{Two-Forking, $s=5$}
\addplot [semithick, color1, mark=triangle, mark size=1.5, mark options={solid}, only marks]
table {%
9.58618836784661 94.8449829260684
9.72901747424202 91.4090958679
9.7654135564256 89.4927122729164
9.92608628858061 88.1800780586612
9.93644838575256 86.1909611825132
9.94508186840737 84.2172145413352
10.1829528221459 85.1386105399373
10.1826188938812 83.1384272192263
10.1950139810187 81.1681357685849
10.2031170440219 79.1542630746519
10.5429516027071 82.51483013444
10.506499415484 80.4088510308948
10.5203787394432 78.4577433563465
10.5575598488005 76.515838734122
10.5301498558711 74.4695544824808
10.9529540691809 80.2295993688406
10.9839933390376 78.3050841596861
10.9779524089415 76.3077341029211
10.9704602606449 74.2552199503227
10.9984542684772 72.3230108070209
11.0325248125969 70.3620658966034
11.5215087212511 78.8708618215623
11.5403816972995 76.8668127782157
11.5519524850544 74.8777322041446
11.6111704471611 73.0470527982461
11.6249163208484 71.0332157272022
11.6807794042568 69.0872883178562
11.7269554409463 67.1720515683489
12.1687899290702 78.8054940339315
12.1887654900803 76.7945914993389
12.2447303194542 74.9658414602408
12.2858335590229 73.0030076168292
12.3821190603664 71.274535789434
12.4538294131731 69.3285410684593
12.5357497438775 67.3987829631589
12.6541311937812 65.5952833378474
12.763629184501 81.2495092464078
12.8192488660363 79.3275198678352
12.906805432571 77.6122614229484
12.9807289914861 75.6239919280013
13.0729054458195 73.6183548996042
13.2262386345127 71.931261954582
13.3823516464639 70.3344070805906
13.5985350788478 68.830733206789
13.83542421111 67.332931858402
13.1582025626476 87.7665598067341
13.255840693199 85.9858510079599
13.3593549336009 84.3099585821631
13.4864356150186 82.6638782247404
13.6393917666889 81.0272203177021
13.8175771683645 79.6152377229281
14.0311363644282 78.3164170191604
14.2604949055072 77.1321945172534
14.5126350644162 76.3486210879403
14.7560592152841 76.1264150496321
};
\addlegendentry{Two-Forking, $s=6$}
\addplot [semithick, color2, mark=+, mark size=1.5, mark options={solid}, only marks]
table {%
9.58927657213014 93.8551607768574
9.74584967967448 90.4601648173737
9.75121105806887 87.4380781312902
9.94622774378008 87.2107422504633
9.95471134584717 84.2449623731731
9.94902959253674 81.1999806793209
10.2026832100923 84.2074134779418
10.1927046358877 81.1791562296608
10.1941020267362 78.1592912702299
10.210383930419 75.2197887438659
10.539379503854 81.4874523782174
10.5407233532372 78.4971171095827
10.5370028555849 75.4984490463816
10.5125559843155 72.4117775023132
10.5507406706166 69.4803965278665
10.9685038050558 79.2592351373191
10.9647388398325 76.2479870369911
10.9808906398527 73.2762270530279
10.9958990061051 70.3386832357857
11.0073150183991 67.3130315250729
11.0334623313685 64.3981990019081
11.5123727888487 77.7782603869892
11.5223192898522 74.8286547186289
11.5544777923042 71.920541604477
11.6010384277978 69.0309497665223
11.6359543774542 66.0670020539495
11.6811330535135 63.0286081065651
11.7676804167765 60.2566176327431
12.1800246149394 77.8497245083219
12.2046559113251 74.8202730020203
12.2556902032295 71.9270042990889
12.3079648735706 69.039951685465
12.3718605402592 66.1702778581992
12.475380747431 63.267788356561
12.5914479240016 60.4365559267655
12.7982063930235 57.8222231278765
12.7476931231653 80.083388072484
12.8311439155481 77.2903645462395
12.9116616581354 74.5661706951076
13.0119426729442 71.5902243243035
13.1227584192784 68.7996123111199
13.2924177218002 66.1148413905535
13.5157135950804 63.6076817762361
13.794771769499 61.1256457078654
14.1829392755474 58.9872515386628
13.1618712028854 86.8749635638593
13.2517694040437 83.9551151991794
13.3852302567349 81.4699511894015
13.5164263512481 78.7692375259709
13.6963866280136 76.0081162420238
13.9335509846604 73.7226942120566
14.2579066794991 71.6527543092109
14.5993450441828 69.7427318471912
15.0531136998428 68.4259411604488
15.4968670669915 67.9182867496423
};
\addlegendentry{Two-Forking, $s=7$}
\addplot [semithick, color3, mark=square, mark size=1.5, mark options={solid}, only marks]
table {%
9.58626062445002 92.8630471779241
9.73865857673258 89.4589150093923
9.74934860508947 85.4587136864106
9.93609804995876 86.2006422490682
9.93633674734843 82.1816136247425
9.94988869180736 78.2200266404394
10.2049379110696 83.2030180772189
10.1963605136968 79.1763199374589
10.2147424671897 75.2131316527551
10.2017280277649 71.1691318245771
10.5182372066589 80.443593798373
10.5303266969666 76.4396284747162
10.5362359400181 72.4755636621877
10.5468759666703 68.511149681659
10.544470697807 64.475322119237
10.9792094506267 78.2920115444362
10.981763086966 74.2718458931688
10.9767808162405 70.2498123677476
11.0187925890303 66.3706205217919
10.9834044348011 62.2487068811103
11.0508834175935 58.3988145270894
11.5321441043357 76.9025693151515
11.5512957449842 72.9006884007106
11.5630526722196 68.9607231378665
11.5901021179154 64.9492896234824
11.6291748099262 60.9686139043097
11.7087220581683 57.1359914037173
11.7680756462673 53.172352608473
12.1775701053215 76.9065231999485
12.227830106636 72.9845846952234
12.2435132204106 68.927170730388
12.3141175252668 65.0545399440176
12.3859179881664 61.261733941962
12.4694044765222 57.2048476480485
12.6369127159499 53.5647613232731
12.8729525604668 49.9088901109444
12.7696413609591 79.2347935277348
12.8301744066606 75.3668829021184
12.8947021625881 71.4442450389737
12.9967687620588 67.5443785082918
13.1347574782622 63.8850767431696
13.3170313541062 60.245840300711
13.5532255577702 56.6203597628725
13.908023829864 53.1832691833296
14.4313467333826 50.2935521438039
13.1635881389266 85.8943986617219
13.2555156484097 82.1475195526252
13.3687138259828 78.4025469732362
13.5094300849676 74.6409282303115
13.7077072638316 71.2344588652653
13.9646180745959 67.6522462488439
14.2901878319092 64.54536918999
14.805820599829 61.9956748366788
15.4201530040461 59.8830174129467
16.1827608342434 59.5281563522505
};
\addlegendentry{Two-Forking, $s=8$}
\addplot [semithick, color=black, mark=*, mark size=2]
table {%
13.0906666831394 89.8608582314354
12.6823039823615 82.8872374933173
12.122257697726 80.6833227397667
11.5160947642278 80.8881884071402
10.9506599637107 82.3538088723514
10.5035087699936 84.4909613348849
10.2105368082049 87.2896482471741
9.88074902479145 90.0607557259166
9.73676228898449 93.5014511657595
9.60539901581864 96.9205424782861
9.48441470405003 100.385333609855
};
\addlegendentry{Single-Forking};
\end{axis}

\end{tikzpicture}
\caption{ 
This figure illustrates achievable points with Pareto execution times for mean server utilization cost and mean service completion time for two-forking with different values of forked servers $m_0$, $m_1$ and $N-m_0-m_1$ at forking points $t_0=0$, $t=4$ and $s \in \{5, 6, 7, 8\}$, respectively. 
For comparison, we also plot the tradeoff points for single-forking at forking time $t=4$ varying the number of initial servers $n_0$.}
\label{Fig:TradeoffParetoTFPAfter}
\end{figure}
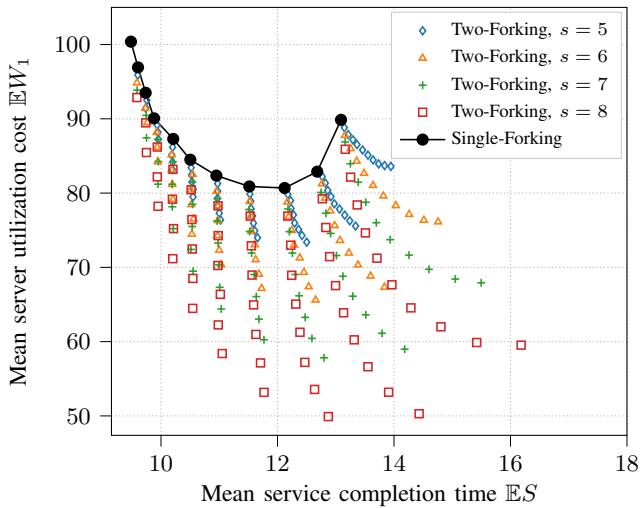

We observe that \emph{for any} choice of forked
servers $m_0, m_1$ and forking point $s > t$, the two-forking
system achieves better tradeoff points as compared to the single-forking system. 

\subsection{Weibull distribution}
\label{subsec:WeibullTFP}
We assume the system parameters for the Weibull distribution to be scale $\gamma=16$ and shape $\beta=2$. 
We compare the performance of the single forking sequence $(0, n_0), (t, N-n_0)$ to that of the two-forking sequence $(0, m_0), (t, m_1), (s, N-m_0-m_1)$, when the second forking point is $s>t$.  

We plot the tradeoff curve between mean service completion time and mean server utilization cost for the single and two-forking sequences in Fig.~\ref{Fig:TradeoffWeibullTFPAfter} for the values of forking points $t = 4$ and $s \in  \{5, 6, 7, 8\}$, varying the number of forked servers $n_0 \in [N]$ in the single-forking and $m_0, m_1$ ine th two-forking case.

In this case of Weibull distribution for job execution time as well, we observe that \emph{for any} choice of forked servers $m_0, m_1$ and forking point $s > t$, the two-forking system achieves better tradeoff points as compared to the single-forking system.

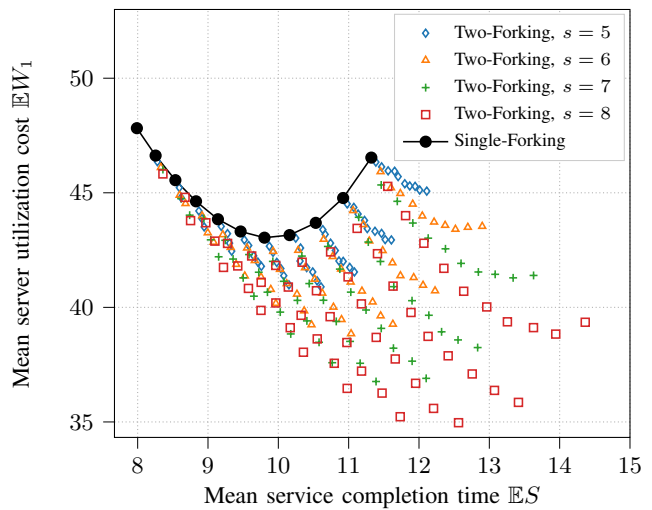
\begin{figure}[hhh]
\centering
\begin{tikzpicture}

\definecolor{color0}{rgb}{0.12156862745098,0.466666666666667,0.705882352941177}
\definecolor{color1}{rgb}{1,0.498039215686275,0.0549019607843137}
\definecolor{color2}{rgb}{0.172549019607843,0.627450980392157,0.172549019607843}
\definecolor{color3}{rgb}{0.83921568627451,0.152941176470588,0.156862745098039}

\begin{axis}[
font=\small,
legend cell align={left},
legend style={at={(0.99,0.99)}, draw=white!80.0!black, font=\scriptsize},
tick align=outside,
tick pos=left,
x grid style={white!69.01960784313725!black, densely dotted},
xlabel={Mean service completion time $\E S$},
xmajorgrids,
xmin=7.67018408804395, xmax=15,
xtick style={color=black},
y grid style={white!69.01960784313725!black, densely dotted},
ylabel={Mean server utilization cost $\E W_1$},
ymajorgrids,
ymin=34.3243996632461, ymax=53,
ytick style={color=black}
]
\addplot [semithick, color0, mark=diamond, mark size=1.5, mark options={solid}, only marks]
table {%
8.28046114665318 46.3572907010435
8.59255250922609 45.225972141765
8.61746375863963 44.7932792119632
8.87059704795574 44.2076221945424
8.91358745776322 43.8709335465048
8.9457359282946 43.5104925907284
9.18814624200725 43.5354583869411
9.27763656746831 43.2172883430765
9.30650873884389 42.8030468368718
9.33686212964761 42.4542026644787
9.51681376507247 42.9465825933688
9.57897785774734 42.6999890310915
9.64966655665075 42.2473503055109
9.7093174314609 41.9762303984608
9.75904418086521 41.7913792978674
9.87236547940577 42.6677836054018
9.92511494189918 42.4028362605313
9.9777923143048 41.9627348671674
10.0355439380796 41.6942475226444
10.0773861429381 41.3815634723778
10.1509417087346 40.9831430462948
10.2474484247889 43.0214748057232
10.315911806519 42.5905909340734
10.3086694984145 42.0214960596896
10.4031224524359 41.7365581933655
10.4890413082171 41.5452610125389
10.5514835865378 41.1572621056977
10.606732487289 40.8957316737687
10.6326763594976 43.391832019016
10.6660805519066 43.1847876439864
10.7425604892132 42.8019764978727
10.8106773998512 42.4799243636105
10.8621046490804 42.025538233712
10.9225921223222 42.0066502632991
11.0156245504369 41.7430995037626
11.0814681666058 41.5445599653275
10.9793426357833 44.5148592777217
11.0733449860206 44.362685708287
11.1340576847 44.0728838692713
11.2255978407454 43.7942344477459
11.2645842987421 43.4381461101954
11.3740900258366 43.3193383252928
11.4678593916302 43.2275630942559
11.5191631103589 42.9578521459902
11.6074649603154 42.9417784551152
11.386807394794 46.3394610765815
11.4684800578714 46.1311758284203
11.5612428482097 45.9582202743214
11.6536360687602 45.9406114191783
11.7036389913351 45.707394675615
11.7979412328634 45.3945633157684
11.868433149506 45.2978626465784
11.944852646004 45.2811086168872
12.0101029320247 45.1194828174929
12.1103387441158 45.0732090405494
};
\addlegendentry{Two-Forking, $s=5$}
\addplot [semithick, color1, mark=triangle, mark size=1.5, mark options={solid}, only marks]
table {%
8.32961132506642 46.103918083636
8.59639622557396 44.896100267432
8.68059755222226 44.5305083379338
8.91661154583134 43.9974194710844
8.99685199376901 43.2541110948845
9.08041028645307 42.8445854143961
9.21120693234066 43.1727899728526
9.32218293053975 42.6173004315665
9.40118854314559 41.8776277793814
9.52289636591419 41.3720191741458
9.55577254305641 42.5921320949625
9.66528443758657 42.0429756967612
9.75438111688733 41.3269800003345
9.89289566469596 40.7845991004807
9.9648936967455 40.1761240206728
9.92694463183602 42.4441874287242
10.0239619081773 41.6584833385982
10.1193724027017 41.0362134695893
10.2651732371636 40.5778846000021
10.3673883626226 39.8487923058094
10.4714855566388 39.2410762482973
10.2833932443829 42.4906656071269
10.3793105981671 41.7299951232242
10.5164751724873 41.2285443743204
10.6251135076379 40.6217100838406
10.7837091369378 40.0101021479828
10.898530665855 39.4688167836509
11.0345635323503 38.8505511941532
10.6452392225814 42.9733311167591
10.7717464662981 42.2177751306593
10.8922983592179 41.6730047522783
11.0287175074408 41.1058794091451
11.188808287784 40.7122487457928
11.341091735476 40.2363366831453
11.4545117941752 39.6313751927593
11.628552117179 39.2660621453861
11.0629065606821 44.2238214148706
11.209378111988 43.613131505289
11.2935329449542 42.8981097711592
11.483451369578 42.4785076418181
11.6228655795146 41.9470303623941
11.7360998083789 41.3831224582019
11.9246095384464 41.3017089376531
12.0578582935054 40.9656735397005
12.2281914740644 40.7240430694224
11.4536688118586 45.9224657832294
11.6016772865862 45.219214570797
11.7347928760884 44.9702590307324
11.9107375883672 44.5060241023582
12.0472618117519 43.9891568708984
12.2149039130113 43.7239578508535
12.3718628084333 43.5814958996375
12.5181299113947 43.4090206784249
12.6942403164635 43.4941975168886
12.899285716116 43.5359033717671
};
\addlegendentry{Two-Forking, $s=6$}
\addplot [semithick, color2, mark=+, mark size=1.5, mark options={solid}, only marks]
table {%
8.35913569690944 46.0155631922102
8.61815605189805 44.7445481133188
8.73929609849284 44.0231596889491
8.94109208247841 43.8739111139529
9.03854514447932 42.9481426710468
9.15130843576663 42.2075487779613
9.24515490879666 42.926754966559
9.3568840312712 42.108385773542
9.50037094132919 41.288897506053
9.65279666128672 40.4881376002206
9.58370951047822 42.3040281552809
9.72187517875981 41.5329810275081
9.84529848664961 40.6631035682373
10.0273230315227 39.7924374266914
10.1766701931844 38.8407005200031
9.92198250821028 42.0065512102517
10.0772246054747 41.1373582535661
10.2717064213925 40.3049834496591
10.4061783993142 39.409750657192
10.5790105564786 38.4754352474606
10.7666134100255 37.585037945062
10.3346714522915 42.2381387973739
10.4399907371091 41.0388749369702
10.6412592148121 40.3102472901447
10.8239634516311 39.3870196642939
11.0201858279622 38.5142932196856
11.1633438245375 37.5597262583976
11.3895925769081 36.764762217326
10.7042293806945 42.6575592144401
10.8745234358802 41.6694448393039
11.0305609639788 40.6645565111057
11.2472960131856 39.8847517995182
11.463456459473 39.0831008106475
11.6367593416381 38.2184968595584
11.9024004756259 37.6521317954262
12.1023895535809 36.9023865591644
11.1432205584453 43.9314009162524
11.2792587845669 42.8497444131216
11.4520768685742 42.003023463013
11.6414059387395 40.9034462387966
11.9010595756389 40.2927150788796
12.1399707752367 39.6551724071446
12.3212126784913 38.9329910683749
12.5552769626368 38.5796115425993
12.8338686186171 38.2455067065203
11.4627927057207 45.3423410994993
11.6914932336692 44.627760781967
11.9124090976173 43.680233669936
12.1307994297794 43.0200591538648
12.3763894427277 42.5552919429183
12.5979296537308 41.9196648021217
12.8441523174891 41.543921245873
13.0831985751067 41.4424405965487
13.3385820106984 41.2883907358093
13.6288873024988 41.3935630162795
};
\addlegendentry{Two-Forking, $s=7$}
\addplot [semithick, color3, mark=square, mark size=1.5, mark options={solid}, only marks]
table {%
8.35907839951041 45.8233610614087
8.67385750700558 44.8003232279036
8.75225015670079 43.785845693363
8.96942519559687 43.6926477911709
9.09839775846372 42.8911524881018
9.22004871701092 41.7436583871455
9.278338249377 42.7914540453972
9.42305996941533 41.8082289880392
9.57500929228028 40.8273673021549
9.75343761476485 39.8689521387458
9.62120031079005 42.2426987416652
9.75629204709663 41.092301730682
9.96501143540502 40.1926900501222
10.1693895951461 39.1130371769098
10.3565965150981 38.0420765329073
9.96340427311634 41.8301588610323
10.1422206762139 40.8856191039268
10.324273438737 39.6534712982926
10.553209820305 38.6247207696391
10.7955143399472 37.5594112032084
10.9791570655078 36.4646823259716
10.3383554583369 41.9700787512673
10.542054207806 40.7265215588338
10.7366495404641 39.5940642376273
10.9749750126267 38.4655061549939
11.1832262105739 37.2179191531992
11.4750902446238 36.2585937370715
11.7330544491783 35.2277379096134
10.7433078343731 42.4195482248585
10.9917726965457 41.330304080748
11.1810820071678 40.1523871128407
11.3920814430737 38.6893657714384
11.6624994074967 37.7477458727894
11.9521892767999 36.6859412441921
12.2069538897247 35.5906489793452
12.5622615170156 34.9670122448311
11.120744026027 43.4467251695175
11.4095961453063 42.3426372510788
11.6313181441956 40.9274378855598
11.8864038142048 39.7765339290015
12.1293448300024 38.7344370137237
12.4149685151899 37.8860855999584
12.7569586026781 37.0945688333981
13.0733955152441 36.3770781063842
13.4142307475436 35.8535718426001
11.555632941045 45.2798282206554
11.8093454729205 44.0031313881562
12.0697305646265 42.797701416295
12.3539531761034 41.7035559309007
12.6357506033666 40.702512868817
12.9635915037281 40.0189532700658
13.2570989077166 39.3701658322837
13.6280110115111 39.1158220299502
13.9445008360398 38.8336000902961
14.3636902760145 39.3503033414876
};
\addlegendentry{Two-Forking, $s=8$}
\addplot [semithick, color=black, mark=*, mark size=2]
table {%
11.322465074922 46.5354066772776
10.9214514914285 44.7721462937553
10.5323738140645 43.6885195190574
10.1601343211863 43.1535240010181
9.80281985280512 43.04095798933
9.46410879364957 43.3113522007682
9.14123601797566 43.8453987947239
8.82982768577365 44.6282204907939
8.53631003553624 45.5507761155286
8.25519713274231 46.6237765174868
7.98892247794731 47.8192638765309
};
\addlegendentry{Single-Forking};
\end{axis}

\end{tikzpicture}
\caption{ 
This figure illustrates achievable points with Weibull execution times for mean server utilization cost and mean service completion time for two-forking with different values of forked servers $m_0$, $m_1$ and $N-m_0-m_1$ at forking points $t_0=0$, $t=4$ and $s \in \{5, 6, 7, 8\}$, respectively. 
For comparison, we also plot the tradeoff points for single-forking at forking time $t=4$ varying the number of initial servers $n_0$.}
\label{Fig:TradeoffWeibullTFPAfter}
\end{figure}
\end{appendices}

\bibliographystyle{IEEEtran}
\bibliography{IEEEabrv,tnet2019}

\end{document}